\documentclass[12pt]{article}
\usepackage{enumitem}
\usepackage{xr}
\usepackage{amsmath, graphicx,algorithm,amsfonts, amsthm, amssymb, latexsym, slashbox, color, bm, algorithmic, appendix, setspace} %natbib,
\definecolor{mygray}{gray}{0.6}
\usepackage[numbers]{natbib}
\usepackage{hyperref}
\hypersetup{colorlinks=true, citecolor=blue, filecolor=blue, linkcolor=blue, urlcolor=blue}
\numberwithin{equation}{section}

\usepackage[left=1in,top=1in,right=1in,bottom=1in,nohead,paperwidth=8.5in, paperheight=11in]{geometry} 
\usepackage{eufrak}
\usepackage{hyperref}
\DeclareMathOperator*{\argmin}{\operatorname{argmin}}
\DeclareMathOperator*{\cov}{Cov}\DeclareMathOperator*{\var}{Var}

\newtheorem{assumption}{Assumption}[section]
\newtheorem{proposition}{Proposition}[section]
\theoremstyle{definition}
\newtheorem{remark}{Remark}[section]

\def\real{{\mathbb R}}
\def\st{\mathrm{~s.t.~}}

\def\var{\mathrm{Var}}

\def\wh{\widehat}

\newcommand{\vertiii}[1]{{\left\vert\kern-0.25ex\left\vert\kern-0.25ex\left\vert #1
		\right\vert\kern-0.25ex\right\vert\kern-0.25ex\right\vert}}

\def\a{\alpha}

\def\xa{x_{\a}}

\def\C{{\cal C^*}}
\def\L{{\cal L}}

\newtheorem{lem}{Lemma}

\def\hba{\hat\beta^{(\alpha)}}
\def\tba{\beta^{*(\alpha)}}

\newtheorem{lemma}{Lemma}[section]

\title{\bf Sparse Identification and Estimation of Large-Scale Vector AutoRegressive Moving Averages}
\author{Ines Wilms$^{a}$\footnote{Equal Contribution.}, Sumanta Basu$^{b*}$, Jacob Bien$^c$\footnote{Corresponding author. E-mail and URLs: jbien@usc.edu, \url{http://faculty.marshall.usc.edu/Jacob-Bien/} (J.\ Bien), i.wilms@maastrichtuniversity.nl, \url{https://sites.google.com/view/iwilms} (I.\ Wilms), sumbose@cornell.edu, \url{http://faculty.bscb.cornell.edu/\~basu/} (S.\ Basu), 
		matteson@cornell.edu,  \url{http://www.stat.cornell.edu/\~matteson/} (D.S.\ Matteson).
	}, and David S.\ Matteson$^b$ 
	\\ \textit{\small $^{a}$ Department of Quantitative Economics, Maastricht University, Maastricht, The Netherlands}
	\\ \textit{\small $^{b}$ Department of Statistics and Data Science, Cornell University, Ithaca, NY, USA}
	\\ \textit{\small $^{c}$ Data Sciences and Operations, University of Southern California, Los Angeles, CA, USA}
}
\date{ }
\begin{document}
	
	\def\spacingset#1{\renewcommand{\baselinestretch}%
		{#1}\small\normalsize} \spacingset{1}
	
	\maketitle
	
	\vspace{-12pt}
	\noindent
	{\bf  Abstract.}
	The Vector AutoRegressive Moving Average (VARMA) model is fundamental to the theory of multivariate time series; however, identifiability issues have led practitioners to abandon it in favor of the simpler but more restrictive Vector AutoRegressive (VAR) model. We narrow this gap with a new optimization-based approach to VARMA identification built upon the principle of parsimony. Among all equivalent data-generating models, we use convex optimization to seek the parameterization that is “simplest” in a certain sense. A user-specified strongly convex penalty is used to measure model simplicity, and that same penalty is then used to define an estimator that can be efficiently computed.  We establish consistency of our estimators in a double-asymptotic regime.  Our non-asymptotic error bound analysis accommodates both model specification and parameter estimation steps, a feature 
	that 
	is crucial for studying large-scale VARMA algorithms.  Our analysis also provides new results on penalized estimation of  infinite-order VAR, and elastic net regression under a singular covariance structure of regressors, which may be of independent interest.  We illustrate the advantage of our method over VAR alternatives on three real 
	data examples.
	\bigskip
	
	
	\newpage
	\spacingset{1.5} % Double spacing

\section{Introduction} \label{intro}
Learning regulatory dynamics and forecasting are two canonical problems in the analysis of multivariate time series, with widespread applications in economics, signal processing and biostatistics amongst others. 
In recent years, there has been increasing focus in networks or graphical models of time series  to describe how a multivariate time series' components interact with each other.
Vector AutoRegressions (VAR) estimated using parsimony-inducing regularization (penalties or priors) have become a popular alternative  \citep{Banbura10, de2008forecasting, kilian2017structural, diebold2014network} %
to factor modeling of high-dimensional time series, e.g., \citep{bai2008large}. % stock2011dynamic}.  
In the classical time series and signal processing literatures, Vector AutoRegressive  Moving Average (VARMA) models are known to provide a more parsimonious description of a linear time invariant system than VAR. However, in practice, their use  has been limited due to identification and estimation issues. The goal of this work is to overcome these challenges by theoretically and empirically investigating the large-scale VARMA as a competitive alternative to the VAR.

In a $\text{VARMA}_d(p,q)$ model, a stationary $d$-dimensional mean-zero vector time series $ {  y}_t$ is modeled as a function of its own $p$ past values and $q$ lagged error terms. More precisely,  
\begin{equation}
	{  y}_t =  \sum_{\ell=1}^p   \Phi_{\ell} {  y}_{t-\ell} + \sum_{m=1}^q   \Theta_{m} {  a}_{t-m} +  {  a}_{t}, \label{VARMA} 
\end{equation} 
where
$ \{   \Phi_{\ell} \in \mathbb{R}^{d \times d} \}_{\ell=1}^{p}$ are autoregressive parameter matrices,
$ \{   \Theta_{m} \in \mathbb{R}^{d \times d} \}_{m=1}^{q}$ are  moving average parameter matrices, and
${  a}_t$ denotes a $d$-dimensional mean-zero white noise  vector time series with $d\times d$ nonsingular contemporaneous covariance matrix $ {\Sigma}_a$. 
The primary focus of this work is to consider VARMA models where $d$ is moderate or large.
A VAR  is a special case of the VARMA without moving average coefficients ($\Theta_m = \mathbf{0}_{d \times d}$, for $m=1, \ldots, q$).

Although VARs are more intensively investigated  (e.g.,
\citep{Davis15, Nicholson17} for computational contributions; \citep{Basu15, Kock15, wu2016performance, basu2019low}  for theoretical contributions, and \citep{Matteson11, Gelper15} for applications),  several reasons exist for preferring the more general VARMA class. 
Unlike VAR, the class of VARMA is closed under marginaliztion and linear transformation \citep{Lutkepohl05}. 
In macroeconomics, VARMA is popular for its close link with linearized dynamic stochastic general equilibrium (DSGE) models \citep{kascha2012comparison, fernandez2007abcs}. 
A parsimonious finite order VARMA can capture the dynamics of a potentially infinite-order VAR, leading to improved estimation and forecasting accuracy. Empirically, VARMAs  have been shown to outperform VARs in terms of estimation and forecasting accuracy \citep{kascha2012comparison, Anthanasopoulos08}. Our empirical analysis also demonstrates such improvements (see Section \ref{applications}). 
\color{black} 
Importantly, we see that VARMA achieves this improved forecast accuracy using a more parsimonious description of the data than VAR.
\color{black} 

Despite its advantages over VAR, VARMA has not been very popular among practitioners due to its computational and theoretical challenges in model identification and specification. The model \eqref{VARMA} is not identifiable in general (see Section \ref{sec:id-prob}), i.e.\ there can be different combinations of AR and MA matrices $\{\Phi_{\ell}\}$ and $\{\Theta_{m}\}$ that lead to the same data generating process. The problem of \textit{model identification} refers to finding a ``simple'' element in this equivalence set $\mathcal{E}$ of all such AR-MA matrices (see Section \ref{identification} for formal definition), usually by specifying a number of restrictions on model parameters. The problem of \textit{model specification} refers to finding these restrictions along with the model orders $p, q$ in a data-driven fashion.

Arguably the most popular identification  procedure is the \textit{Echelon form identification} \citep{hannan1984, poskitt2016, Chan16}, which amounts to selecting a basis for the row space of a block Hankel matrix (see Section 4 of \citep{deistler1985}).
Specifying an Echelon form involves selecting \textit{Kronecker orders} (related to indices of rows that form the above basis) from a $O\left((p+q)^d \right)$-dimensional set, by comparing an equally large number of models. Data-driven strategies, involving a series of canonical correlation tests, or regressions based on model selection criteria (e.g., AIC, BIC, information theoretic criterion) were proposed \citep{akaike1974new, akaike1976canonical, poskitt1992}. However, all of these methods are computationally intensive and 
lack a formal asymptotic theory that combines specification and estimation. Assuming $d$ is fixed, \citep{poskitt1992} proved asymptotic theory for the specification step. Then, assuming Kronecker orders are known, consistency of parameter estimation was established. This procedure has been tested only on very small $d$, and finite sample performances are not clear (Section 3.4, \citep{lutkepohl2006forecasting}).

Other popular identification and specification methods include scalar component models \citep{tiao1989, athanasopoulos2008, athanasopoulos2012} and final equations form \citep{zellner1974, hannan1976, wallis1977}.
While these and other existing identification procedures \citep{hannan1984, poskitt2016, Chan16} require different sets of assumptions---sometimes more relaxed ones than we will consider---on the structure of the process, they inherently face the same limitations for large-scale models. The uncertainty and error in the data-driven specification stage is not accounted for in the analysis of the model parameter estimation stage.

These computational and theoretical challenges of aggregating the model selection and parameter estimation are akin to the variable selection challenges in linear regression, where shrinkage methods (e.g., ridge, lasso, elastic net) have been successfully used in combining selection and parameter estimation. A key advantage of these approaches is that they allow formal asymptotic analysis of the complete specification-plus-estimation procedure.

In this work, we show that these convex optimization based techniques of regularization and dimension reduction, by now ubiquitous in the field of high-dimensional statistics, provide new perspectives and solutions to large-scale VARMA identification and estimation problems with several attractive properties. 

\textbf{I. Automatic identification of parsimonious VARMA models.} We show that by devising a suitable convex penalty, we can identify a parsimonious element in the equivalence class $\mathcal{E}$ in an intuitive yet objective fashion (Section \ref{identification}). More formally, we can define the class of AR-MA matrices with minimum $\ell_1$-norm as a partially identified class of ``sparse'' VARMA models $\mathcal{RE} = \argmin_{(\Phi, \Theta)\in\mathcal E}\{ \sum_{\ell=1}^p \|\Phi_\ell\|_1 + \sum_{m=1}^q \|\Theta_m\|_1 \}$. We could also use a modified, strongly convex penalty $\argmin_{(\Phi, \Theta)\in\mathcal E}\{ \sum_{\ell=1}^p (\|\Phi_\ell\|_1 + \alpha \|\Phi_\ell\|_F^2) + \sum_{m=1}^q (\|\Theta_m\|_1 + \alpha \|\Theta_m\|_F^2) \}$ with a very small $\alpha \approx 0$ to identify a parsimonious element in $\mathcal{RE}$, viz.\ the \emph{unique} AR-MA matrices with minimum Frobenius norm (Proposition \ref{prop:yule-walker-identification}). 

\textbf{II. Computationally efficient estimation of VARMA models.}
Our identification strategy explicitly links the search for a unique, parsimonious model throughout the identification, specification and estimation stages. The same penalty used in our identification is used as a regularizer to define a natural VARMA estimator corresponding to this identified target (Section \ref{methodology}). 
We show on real and simulated data examples (Section \ref{applications} and Appendix G) that such parsimonious VARMA models lead to important gains in forecast accuracy compared to parsimoniously estimated VARs. An implementation of our fully-automated VARMA identification and estimation procedure is available in the \texttt{R} package \texttt{bigtime} \citep{bigtime}.

\textbf{III. Non-asymptotic theory for sparse VARMA.} We also provide a non-asymptotic theoretical analysis of our proposed sparse VARMA estimator  (Section \ref{sec:theory}).  Our analysis explicitly captures the complexity of model selection, and does not assume the identification restrictions are known {\it a priori} as in existing asymptotic analysis of VARMA \citep{dias2018, Dufour14}. 
While to the best of our knowledge, consistency of VARMA estimators has been studied only in the low-dimensional, fixed $d$ asymptotic regime \citep{kascha2012comparison, dufour2005asymptotic},
our error bound analysis shows consistent estimation is possible in a double-asymptotic regime $d, T \rightarrow \infty$. 
\color{black}
We provide two main results on consistency (Proposition \ref{prop:elastic-net-fixed-X-E-v2}). 
Our first result in the spirit of \textit{partial identification} \citep{manski2010partial, tamer2010partial} states that under suitable sparsity assumptions our algorithm provides a parsimonious VARMA estimator (small $\ell_1$-norm) whose distance from the equivalence class $\mathcal{E}$ asymptotically vanishes as long as $\log d / T \rightarrow 0$. Our second result 
on \textit{point identification} states that our estimator converges in probability to our identified target in $\mathcal{E}$ as long as $d^4 \log d /T \rightarrow 0$.  
\color{black}

	\section{\label{identification} Identification of the VARMA}
	We revisit the VARMA identification problem in Section \ref{sec:id-prob}, 
	then introduce an optimization-based, parsimonious identification strategy for VARMA  in Sections \ref{sec:id-new} and \ref{sparse_identif}. 
	
	\subsection{Identification Problem}
	\label{sec:id-prob}
	Consider the  $\text{VARMA}_d(p,q)$ of Equation \eqref{VARMA} with fixed autoregressive order $p$ and moving average order $q$. 
	The model can be written using compact lag operators   as
	${\Phi}(L) {  y}_t =  {\Theta}(L) {  a}_t$, 
	where the AR and MA operators are respectively given by
	\begin{equation}
		{\Phi}(L)  = {  I} -  \Phi_{1}L -  \Phi_{2}L^2 - \ldots -   \Phi_{p}L^p
		\ \ \text{and} \ \     {\Theta}(L)  = {  I} +  \Theta_{1}L +  \Theta_{2}L^2 + \ldots +   \Theta_{q}L^q, \nonumber
	\end{equation}
	with the lag operator $L^\ell$ defined as $L^\ell {  y}_t = {  y}_{t-\ell}$.
	We assume the model is stable and invertible 
	%(e.g., 
	%\citep{Brockwell91}, Chapter 11),  
	meaning respectively that
	$\text{det}\{\Phi(z)\} \neq 0 $  and $\text{det}\{\Theta(z)\} \neq 0 $ for all $|z| \leq 1$ $(z \in \mathbb{C})$. 
	The process $ \{{  y}_t\}$ then has
	an infinite-order VAR representation
	$\Pi(L) {  y}_t = {  a}_t$, 
	where 
	$\Pi(L) =  {\Theta}^{-1}(L)  {\Phi}(L) =  {  I} -  \Pi_{1}L -  \Pi_{2}L^2 - \cdots,$ 
	with $\text{det}\{\Pi(z)\} \neq 0 $ for all $|z| \leq 1$. 
	The $\Pi$-matrices can be computed recursively from the AR matrices $\{ \Phi_{\ell} \}$ and MA matrices $\{ \Theta_m \}$ (e.g.,  \citep{Brockwell91}, Chapter 11).
	The VARMA  is uniquely defined in terms of the operator $ {\Pi}(L)$, but not in terms of the AR and MA operators $ {\Phi}(L)$ and $ {\Theta}(L)$, in general. 
	That is, for a 
	given $ {\Pi}(L)$, $p$, and $q$, one can define an equivalence \color{black} class \color{black} of AR and MA matrix pairs, 
	\begin{equation}\label{eqn:eqv-class}
		\mathcal{E}_{p,q}( {\Pi}(L)) = \{ ( {\Phi},  {\Theta})  : 
		{\Phi}(L) =  {\Theta}(L) {\Pi}(L)  \}, \nonumber
	\end{equation}
	where $  \Phi = [  \Phi_{1} \cdots   \Phi_{p}]$ and 
	$  \Theta = [  \Theta_{1} \cdots   \Theta_{q}]$.  This \color{black} class \color{black} can, in general, consist of more than one such pair,
	implying that further  
	identification restrictions on the AR and MA matrices are needed for meaningful estimation.  
	
	In order to connect identification to estimation, we first provide an alternate characterization of the equivalence class $\mathcal{E}_{p,q}(\Pi(L))$ in terms of a Yule-Walker type equation. 
	\begin{proposition}[Yule-Walker type equation for VARMA]\label{prop:yule-walker-identification}
		Consider a white noise process $\{a_t \}_{t \in \mathbb{Z}}$ with mean zero and variance $ \Sigma_a$. For a stable, invertible linear filter $\Pi(L)$ that allows a $\text{VARMA}_d(p,q)$ representation $\Pi(L) = \Theta^{-1}(L) \Phi(L)$, consider the 
		process  $y_t = \Pi^{-1}(L) a_t$ and 
		define $z_t = \left[y_{t-1}^\top: \cdots: y_{t-p}^\top:a_{t-1}^\top: \cdots: a_{t-q}^\top \right]^\top$. Then, 
		$(\Phi, \Theta)\in \mathcal{E}_{p,q}(\Pi(L))$ if and only if $\beta_{d(p+q) \times d} :=  \left[ \Phi_1 :\ldots: \Phi_p: \Theta_1: \ldots : \Theta_q  \right]^\top$ is a solution to the system of equations $\rho_{zy} = \Sigma_z \beta$, where $\rho_{zy} = \mathbb{E}[z_t y_t^\top]$ and $\Sigma_z = \mathbb{E}[z_t z_t^\top]$.   That is, 
		\begin{equation}\label{eqn:yule-walker}
			\mathcal{E}_{p,q}(\Pi(L)) = \left\{ (\Phi, \Theta): \rho_{zy} = \Sigma_z \beta \right\}.
		\end{equation}
	\end{proposition}
	
	A proof of this proposition is provided in Appendix A.1.  Note that both $\rho_{zy}$ and $\Sigma_z$ can be expressed as functions of $\Pi$ and $\Sigma_a$ alone (i.e.\ they do not depend on $\Theta$ and $\Phi$), and hence are uniquely defined for the underlying process $y_t$. 
	While the AR($\infty$) representation given by $\Pi$ in Proposition \ref{prop:yule-walker-identification} is unique, it allows an equivalent characterization in terms of many ($\Phi$, $\Theta$) combinations. Each of these combinations is a solution to the (potentially) underdetermined system of equations in Proposition \ref{prop:yule-walker-identification}.
	
	A key consequence of this proposition is that our identification target can be defined by optimizing over the solution set of this Yule-Walker type equation. Further, we can use 
	sample analogues of $\rho_{zy}$ and $\Sigma_z$ in our estimation step to search for this target in a data-driven fashion.
	
	\subsection{Optimization-based Identification}
	\label{sec:id-new}
	We rely on strongly convex optimization to establish identification for VARMA models. 
	Among all feasible AR and MA matrix pairs, we look for the one that gives the most parsimonious VARMA representation. We measure parsimony through a pair of  convex regularizers, $\mathcal{P}_{\text{AR}}(\Phi)$ and $\mathcal{P}_{\text{MA}}(\Theta)$.
	Our identification results apply equally well to any convex function:  one may consider,  amongst others, the $\ell_1$-norm, the $\ell_2$-norm, the nuclear norm, and combinations thereof. 
	Our methodology also allows for a different choice of convex function for the AR and MA matrices if prior knowledge would allow a more informed modeling approach.  This might be particularly useful in economics, for instance, where one may be interested in a parsimonious AR structure for interpretability, but can allow for a non-sparse MA polynomial to increase forecast accuracy. 
	
	We now define the  \textit{regularized}  equivalence  class of VARMA representations as
	\begin{equation}
		\mathcal{RE}_{p,q}(  \Pi(L)) = \underset{{  \Phi}, {  \Theta}}{\operatorname{argmin}} \ \{ \mathcal{P}_{\text{AR}}({  \Phi}) + \mathcal{P}_{\text{MA}}({  \Theta}) \ \st  \ {  \Phi}(L) = {  \Theta}(L)   \Pi(L)\}. \label{Optim_set}
	\end{equation}
	This \color{black}  regularized   equivalence  class  is a subclass of the equivalence  class \color{black} $\mathcal{E}_{p,q}( {\Pi}(L))$, containing the \color{black}  regularized  \color{black} VARMA representations. 
	If the objective function in \eqref{Optim_set}  is strongly convex, then the \color{black}  regularized  \color{black} equivalence  \color{black} class \color{black} consists of one unique AR-MA matrix pair, in which case identification is established. 
	However, for the $\ell_1$-norm, for instance, the objective function is convex but not strongly convex. Hence, to ensure identification for this case, we add two extra terms to the objective function and consider 
	\begin{equation}
		({\Phi}^{(\alpha)}, {\Theta}^{(\alpha)})  =  \underset{{\Phi}, { \Theta}}{\operatorname{argmin}} 
		\{\mathcal{P}_{\text{AR}}({\Phi}) + \mathcal{P}_{\text{MA}}({  \Theta}) + \dfrac{\alpha}{2} \|{\Phi}\|_F^2 +  \dfrac{\alpha}{2}  \|{\Theta}\|_F^2  \ \st   {\Phi}(L) = {\Theta}(L)   \Pi(L) \}. \label{Optim_pair}
	\end{equation}
	Problem \eqref{Optim_pair} is strongly convex and thus has a \emph{unique} solution pair $({  \Phi}^{(\alpha)}, {  \Theta}^{(\alpha)})$ for each %value of 
	$\alpha>0$. 
	For any stable, invertible VARMA, we then define its unique \color{black}  regularized  \color{black} representation in terms of the AR-MA matrices as
	\begin{equation}
		({  \Phi}^{(0)}, {  \Theta}^{(0)}) = \underset{\alpha \rightarrow 0^+}{\text{lim}} ({  \Phi}^{(\alpha)}, {  \Theta}^{(\alpha)}). \label{uniquePhiTheta}
	\end{equation}
	The following proposition, proved in Appendix A.2, establishes that $({  \Phi}^{(0)}, {  \Theta}^{(0)})$ is in the \color{black}  regularized  equivalence class $\mathcal{RE}_{p,q}(  \Pi(L))$  \color{black} and furthermore is the \textit{unique} pair of autoregressive and moving average matrices in this set having the smallest Frobenius norm. This result is similar to a result in the regression context, which states that the LARS-lasso solution has the minimum $\ell_2$-norm over all lasso solutions (see \citep{tibshirani2012}, Lemma 7). 
	
	\begin{proposition}\label{thm:identification}
		The limit in \eqref{uniquePhiTheta} exists and is the unique pair in the set
		$\mathcal{RE}_{p,q}(  \Pi(L))$ \color{black} whose 
		Frobenius norm squared is smallest:  
		$$
		({  \Phi}^{(0)}, {  \Theta}^{(0)})=\underset{{  \Phi}, {  \Theta}}{\operatorname{argmin}} \{\|\Phi\|_F^2+\|\Theta\|_F^2\ \st\ (\Phi,\Theta)\in  \mathcal{RE}_{p,q}(  \Pi(L))  \}.
		$$
	\end{proposition}
	
	\subsection{\label{sparse_identif}Sparse Identification} \color{black}
	While our identification results apply equally well to any convex function, we give special attention to sparsity-inducing convex regularizers. In this case, the 
	regularized equivalence \color{black} class \color{black} in \eqref{Optim_set} is a sparse equivalence \color{black} class, \color{black} meaning that, in general, we would expect many of the elements of the AR and/or MA matrices to be exactly equal to zero. 
	
	\color{black}To guarantee the sparsest VARMA representation, one might consider taking
	$ \mathcal{P}_{\text{AR}}(\Phi) = \|\Phi\|_0$ and  $\mathcal{P}_{\text{MA}}(\Theta) = \|\Theta\|_0.$
	However,  since the $\ell_0$-penalty is non-convex, a unique solution cannot be guaranteed. One can construct examples in which there exist multiple equivalent, sparsest VARMAs, see \citep{Tsay14} 
	and Appendix A.3.1. 
	Strong convexity in \eqref{Optim_pair} is key to guaranteeing uniqueness of $({  \Phi}^{(\alpha)}, {  \Theta}^{(\alpha)})$.  For sparsity, we may therefore add to the $\ell_2$-norm in \eqref{Optim_pair} the $\ell_1$-norm $\mathcal{P}_{\text{AR}}(\Phi) = \|\Phi\|_1$  and $\mathcal{P}_{\text{MA}}(\Theta) =  \|\Theta\|_1$ as a sparsity-inducing convex heuristic.
	
	While our theory will focus on the $\ell_1$-norm, in the empirical sections we also  investigate a time-series specific alternative penalty,  the hierarchical lag (hereafter ``HLag") penalty \citep{Nicholson16, wilms2017nips}:
	$\mathcal{P}_{\text{AR}}(\Phi) = \sum_{i=1}^d \sum_{j=1}^d \sum_{\ell=1}^{p} \| {  \Phi}_{(\ell:p), ij} \|, \ \text{and } \ \mathcal{P}_{\text{MA}}(\Theta) = \sum_{i=1}^d \sum_{j=1}^d \sum_{m=1}^{q} \| {  \Theta}_{(m:q), ij} \|, $ with
	${  \Phi}_{(\ell:p), ij} = \left [{\Phi}_{\ell, ij} \ldots  {\Phi}_{p, ij} \right ] \in \mathbb{R}^{(p-\ell+1)}$ and 	${  \Theta}_{(m:q), ij} = \left [{\Theta}_{m, ij} \ldots  {\Theta}_{q, ij} \right ] \in \mathbb{R}^{(q-m+1)}$. 
	This penalty involves a lag-based hierarchical group lasso penalty (e.g., \citep{yan2017hierarchical}) on the AR (or MA) parameters.
	It allows for automatic lag selection  by forcing lower lags of a time series in one of the VARMA equations to be selected before its higher order lags and is thus built on the intuition of encouraging increased sparsity in $\Phi_{\ell}$ and $\Theta_{\ell}$ as the lag increases. 
	
	\section{\label{methodology}Sparse Estimation of the VARMA} \color{black}
	We estimate and determine the degree of parsimony of 
	VARMA parameters 
	by the use of convex regularizers. Since the $\text{VARMA}_d(p,q)$ of Equation \eqref{VARMA} \color{black} cannot be directly estimated as it  contains the unobservable (latent) lagged errors, 
	we proceed in two phases, in the spirit of \cite{Spliid83, Dufour14}, and references therein.
	In Phase-I, we approximate these unobservable errors.  In Phase-II, we estimate the VARMA  with the approximated lagged errors. 

	\subsection{\label{PhaseIestimators}Phase-I: Approximating the unobservable errors}
	The VARMA  of Equation \eqref{VARMA} has a pure  VAR($\infty$) representation if it is invertible (Section \ref{sec:id-prob}). 
	We  therefore approximate the errors ${  a}_t$ by the residuals  of a VAR($\widetilde p$) given by
	\begin{equation}
		{  y}_t = \sum_{\tau=1}^{\widetilde{p}} {  \Pi}_{\tau} {  y}_{t-\tau} + {  \varepsilon}_t, \label{VARptilde}
	\end{equation}
	for $ (\widetilde p+1) \leq t \leq T$, with $\widetilde p$ a finite number, 
	$ \{{  \Pi}_{\tau} \in \mathbb{R}^{d \times d} \}_{\tau=1}^{\widetilde p}$ the AR parameter matrices, and 
	${{  \varepsilon}}_t$ a 
	vector error series. 
	Denote the estimates by $\widehat{\Pi}_{\tau}$ and residuals by $\widehat{{  \varepsilon}}_t =y_t -  \sum_{\tau=1}^{\widetilde{p}} {  \widehat{\Pi}}_{\tau} {  y}_{t-\tau}$.
	
	Estimating the VAR($\widetilde p$) of Equation \eqref{VARptilde} is challenging since $\widetilde p$ needs to be sufficiently large such that the residuals $\widehat{{  \varepsilon}}_t$ accurately approximate the errors ${  a}_t$. 
	Since, a large number of parameters ($\widetilde pd^2$), relative to the time series length $T$, needs to be estimated, we use regularized estimation. 
	For ease of notation, first rewrite model \eqref{VARptilde} in compact matrix notation
	${  Y} =   {  \Pi} {  Z} + {{  E}}, $
	where 
	${  Y} = [ {  y}_{\widetilde p+1} \ldots {  y}_T] \in \mathbb{R}^{d\times ({T-\widetilde p})}, 
	{  Z} = [{  z}_{\widetilde p+1} \ldots {  z}_T] \in \mathbb{R}^{d\widetilde{p}\times ({T-\widetilde p})}, \text{with} \ 
	{  z}_t = [ {  y}^\top_{t-1} \ldots {  y}^\top_{t-{\widetilde{p}}} ]^\top \in \mathbb{R}^{(d\widetilde{p}\times 1)}, 
	{{  E}} = [ {  \varepsilon}_{\widetilde p+1}  \ldots {  \varepsilon}_T] \in \mathbb{R}^{d\times ({T-\widetilde p})},$ and
	${  \Pi} = [{  \Pi}_{1} \ldots {  \Pi}_{\widetilde{p}}] \in \mathbb{R}^{d\times d\widetilde{p}}.$ 
	The regularized autoregressive estimates ${  \widehat{\Pi}}$ are obtained as
	\begin{equation}
		\widehat{{  \Pi}} = \underset{{  \Pi}}{\operatorname{argmin}} \left \{ \dfrac{1}{2} \| {  Y} -  {  \Pi} {  Z}\|_F^2  + \lambda_{{ \Pi}} \mathcal{P}(\Pi) \right \}, \label{HVAReq}
	\end{equation}
	where we use the squared Frobenius norm  
	as loss function and \color{black} $\mathcal{P}(\Pi)$ is any convex regularizer.
	In our simulations and applications, we focus on sparsity-inducing regularizers ($\ell_1$-norm or  HLag penalty). 
	\color{black} 
	The penalty parameter $\lambda_{{  \Pi}}>0$ then regulates the degree of sparsity in ${  \widehat{{  \Pi}}}$: the larger $\lambda_{{  \Pi}}$, the sparser ${  \widehat{{  \Pi}}}$. 
	Problem \eqref{HVAReq} can be efficiently solved using Algorithm 1 in \cite{Nicholson16}.

	\subsection{\label{PhaseIIestimators}Phase-II: Estimating the VARMA}
	We continue with the approximated lagged errors $\widehat{{  \varepsilon}}_{t-1},\ldots,\widehat{{  \varepsilon}}_{t-q}$ instead of the true errors ${  a}_{t-1},\ldots,{  a}_{t-q}$ in Equation \eqref{VARMA}. The resulting model 
	\begin{equation}
		{  y}_t =  \sum_{\ell=1}^p {  \Phi}_{\ell} {  y}_{t-\ell} + \sum_{m=1}^q {  \Theta}_{m} \widehat{{  \varepsilon}}_{t-m} +  {{  u}}_{t}, \label{VARX} 
	\end{equation} 
	is a %lagged 
	regression 
	of  ${  y}_t$ on  ${  y}_{t-1}, \ldots, {  y}_{t-p}, \widehat{{  \varepsilon}}_{t-1},\ldots,\widehat{{  \varepsilon}}_{t-q}$ with vector error series ${{  u}}_t$. 
	To tackle the VARMA overparameterization problem  and establish identification simultaneously with estimation\color{black}, we again use regularization. 
	
	Rewrite the lagged regression \eqref{VARX} in compact matrix notation
	${  Y} =   {  \Phi} {  Z} + {  \Theta} {  X} +  {{  U}}, $
	where ${  Y} = [ {  y}_{\bar{o}+1} \ldots {  y}_T] \in \mathbb{R}^{d\times ({T-\bar{o}})},$ 
	${  Z} = [{  z}_{\bar o+1} \ldots {  z}_T] \in \mathbb{R}^{d\hat{p}\times ({T-\bar o})}$, with
	${  z}_t = [ {  y}^\top_{t-1} \ldots {  y}^\top_{t-\hat{p}} ]^\top \in \mathbb{R}^{(d\hat{p}\times 1)},$ 
	${  X} = [{  x}_{\bar{o}+1} \ldots {  x}_T] \in \mathbb{R}^{d\hat{q}\times ({T-\bar{o}})}$
	with 
	${  x}_t = [ \widehat{{  \varepsilon}}^\top_{t-1} \ldots \widehat{{  \varepsilon}}^\top_{t-\hat{q}} ]^\top \in \mathbb{R}^{(d\hat{q}\times 1)}$,
	with $\bar{o}=\text{max}(\hat{p}, \hat{q})$,
	for specified order $\hat{p}, \hat{q}$,
	${{  U}} = [ {  u}_{\bar o+1}  \ldots {  u}_T] \in \mathbb{R}^{d\times ({T-\bar o})}, 
	{  \Phi} = [ {  \Phi}_{1}\ldots {  \Phi}_{\hat{p}}] \in \mathbb{R}^{d\times d\hat{p}},$ and
	${  \Theta} = [ {  \Theta}_{1} \ldots {  \Theta}_{\hat{q}}] \in \mathbb{R}^{d\times d\hat{q}}.$ 
	The regularized VARMA estimates are obtained as:
	\begin{equation}
		(\widehat{{  \Phi}}^{(\alpha)}, \widehat{{  \Theta}}^{(\alpha)}) = \underset{{  \Phi},{  \Theta}}{\operatorname{argmin}}  \{ \dfrac{1}{2} \| {  Y} -  {  \Phi} {  Z} - {  \Theta}{  X}\|_F^2  + \lambda_{{ \Phi}} \mathcal{P}_{\text{AR}}(\Phi)  + \lambda_{{ \Theta}} \mathcal{P}_{\text{MA}}(\Theta)  +  \dfrac{\alpha}{2} (\lambda_{{ \Phi}}  \color{black} \|{  \Phi}\|_F^2 +   \lambda_{{ \Theta}} \color{black} \|{  \Theta}\|_F^2 \color{black}) \},
		\label{HVARXeq}
	\end{equation}
	where $\lambda_{{ \Phi}}, \lambda_{{ \Theta}}>0$ are two penalty parameters. \color{black}
	By adding the regularizers $\mathcal{P}_{\text{AR}}(\Phi)$ and $\mathcal{P}_{\text{MA}}(\Theta)$ to the objective function, estimation of large-scale VARMAs is feasible. The addition of the squared Frobenius norms makes the problem strongly convex, ensuring a unique solution in the same way as was done in the identification scheme \eqref{Optim_pair}.   \color{black}
	Optimization problem \eqref{HVARXeq} can be solved via the proximal gradient algorithm in Appendix F.  
	We investigate the forecast accuracy of the proposed VARMA on simulated data in Appendix G. 
	
	\subsection{Choosing Tuning Parameters}
	The estimation procedure involves three sets of  user-defined choices: (i) the
	maximum lag orders $\widetilde p, \widehat{p}, \widehat{q}$; 
	(ii) the penalty parameters $\lambda_{\Pi}, \lambda_\Phi, \lambda_\Theta$; and 
	(iii) the parameter $\alpha$ to ensure uniqueness. 
	We choose these in either a data-driven or computationally inexpensive manner. Below we motivate our choices and address implications of misspecification.
	
	{\bf The maximal lag orders $\widetilde p, \widehat{p},$ and $ \widehat{q}$.} We take $\widetilde p=\lfloor 1.5 \sqrt{T} \rfloor $ and $\widehat{p}=\widehat{q} = \lfloor 0.75 \sqrt{T} \rfloor$.  Our theoretical analysis suggests that $\widetilde{p} \asymp T^{\frac{1}{2}-\epsilon}$ (Proposition \ref{prop:est-error-main}), and for larger $d$, overselecting AR/MA orders only affects the estimation and prediction performance at a rate of $\log d$ (Proposition \ref{prop:est-error-phase-II}). To simplify practical implementation, we therefore set  these values at a slightly larger order $O(\sqrt{T}$).
	
	We perform a simulation study (Appendix G.4) to investigate misspecification of the maximal lag orders. 
	We find that, in general, overselecting is less severe than underselecting. The price to pay for overselection is smaller for the HLag penalty than for the 
	$\ell_1$-penalty since the former 
	performs automatic lag selection. As such, it can reduce the effective maximal order of each series in each equation of the VAR (Phase-I) and VARMA (Phase-II).
	
	{\bf The penalty parameters $\lambda_{\Pi}, \lambda_\Phi$ and $\lambda_\Theta$.} We select the penalty parameters using cross-validation. Below, we describe the selection of $\lambda_\Pi$ in Phase-I; in Phase-II, we proceed similarly but using a two-dimensional grid search for the penalty parameters $(\lambda_{{ \Phi}},\lambda_{{ \Theta}})$. 
	
	Following \cite{Friedman10}, we use a grid of ten penalty parameters starting from $\lambda_{\Pi, \text{max}}$, an estimate of the smallest value for which all parameters are zero, and then decreasing in log linear increments. 
	We then 
	use the following time series cross-validation approach:
	For each time point $t=S, \ldots, T-h$, with $S=\lfloor 0.9\cdot T\rfloor$ and forecast horizon $h$, we estimate the model and obtain parameter estimates. This results in ten different parameter estimates, one for each value of the penalty parameter in the grid. 
	From these  estimates, we compute $h$-step ahead  forecasts $\widehat{{  y}}^{(\lambda)}_{t+h}$ obtained with penalty parameter $\lambda$. 
	We select the value of  $\lambda_{{ \Pi}}$ that gives the most regularized model whose Mean Squared Forecast Error 
	\begin{equation}
		\text{MSFE}_h^{(\lambda)} = \dfrac{1}{T-h-S+1} \sum_{t=S}^{T-h} \dfrac{1}{d} \|{y}_{t+h} - \widehat{{y}}_{t+h}^{(\lambda)}\|^2, \nonumber \label{MSFE_eq}
	\end{equation}
	is within one standard error (see \cite{Hastie09}; Chapter 7) of the minimal MSFE.
	In simulations, we take $h=1$; in the forecast applications, we also consider other forecast horizons.
	
	{\bf The parameter $\alpha$.}
	We will sometimes refer to Equation \eqref{HVARXeq} as an ``elastic net" problem, although, unlike $\lambda_{{ \Phi}}$ and $\lambda_{{ \Theta}}$, the parameter $\alpha$ is not treated as a statistical tuning parameter; rather,  as a small positive value simply used to ensure uniqueness.
	Our simulation study in Appendix A.3.2 reveals that the addition of a small non-zero $\alpha$ indeed produces sparse VARMA estimates close to the unique $(\Phi^{(0)}, \Theta^{(0)})$ pair defined in Equation \eqref{uniquePhiTheta}.
		For $\alpha=0$, we still retrieve sparse VARMA estimates that are close to \textit{an} element in the sparse equivalence class. The resulting estimates are typically sparser (i.e.\ they have fewer non-zero components) than the estimates obtained with a small non-zero $\alpha$ since the target  $(\Phi^{(0)}, \Theta^{(0)})$ corresponds to 
		the pair with minimum Frobenius norm among all minimum-$\ell_1$ VARMA representations.
		Since our main objectives are to produce VARMA estimates that are close to the sparse equivalent class and have good out-of-sample forecast performance, we prefer to work with the sparser estimates and thus take $\alpha=0$ in practice, as we have done in our forecast applications (Section \ref{applications}) and simulations (Appendix G).

\section{Theoretical Properties}\label{sec:theory}
We establish consistency of our VARMA estimator with the lasso penalty in Phase-I and elastic net penalty in Phase-II under a double asymptotic regime where dimension $d$ grows with the sample size. Our Phase-II estimator is essentially an elastic net regression, but introduces additional complexities compared to i.i.d.\ or stochastic regression that need to be dealt with in the asymptotic analysis. The rows of the design matrix consist of consecutive observations from an \textit{approximate} version of the time series $z_t = [y_{t-1}^\top:\ldots: y_{t-p}^\top:a_{t-1}^\top:\ldots:a_{t-q}^\top]^\top$, with $a_t$ approximated by Phase-I residuals $\hat{\varepsilon}_t$. The error term in the regression involves $\hat{\varepsilon}_t$ which do not have an analytically tractable distribution. In addition, since $\Phi(L)y_t = \Theta(L)a_t$, the population covariance matrix of the predictors $\Sigma_z$ is potentially singular. It is not clear whether a restricted eigenvalue (RE) assumption, commonly used in high-dimensional regression \citep{powai2012}, holds in Phase-II regression.

We start by establishing in Section \ref{subsec:key-results} deterministic upper bounds on the estimation error of a generic elastic net regression under some sufficient conditions. A crucial step to verify these sufficient conditions is to derive upper bounds to control the approximation error of $a_t$ by $\hat{\varepsilon}_t$ in Phase-I. We do this in Section \ref{subsec:error-phase-I}. Finally, in Section \ref{subsec:error-phase-II} we show that these sufficient conditions for Phase-II elastic net regression are satisfied with high probability for random realizations from the VARMA model, and present estimation error bounds.

To maintain analytical tractability when tackling the VARMA specific complexities, we consider two modifications in Phase-II. First, we use  $\hat{y}_t := y_t - \hat{\varepsilon}_t$, the fitted values from Phase-I, instead of $y_t$, as response in  Phase-II. The analysis can be modified in a straightforward fashion to use $y_t$ as  response, although the resulting upper bounds become larger. 
Second, we consider a constrained version of the penalized Phase-II estimator with an additional side constraint on the $\ell_1$-norm of the regression coefficient. Equivalence of the constrained and penalized versions follows from duality of the convex programs. The additional side constraint on the regression coefficient is easy to implement in practice \citep{agarwal2010fast}, and has been used for technical convenience in earlier literature on high-dimensional statistics \citep{powai2012}.

We assume Gaussianity in our analysis, primarily to apply some concentration inequalities for Gaussian processes in our non-asymptotic error bound analysis. The results can be extended to non-Gaussian VARMA using recent concentration bounds for non-Gaussian linear processes \citep{sun2018large} with potentially slower convergence rate for processes with heavier tails than Gaussian, although the technical exposition becomes more cumbersome.

\textit{Notation.} 
We denote the sets of integers, real, and complex numbers by $\mathbb{Z}$, $\mathbb{R}$, and $\mathbb{C}$, respectively. We use $\|.\|$ to denote the Euclidean norm of a vector and the operator norm of a matrix. We reserve $\|.\|_0$, $\|.\|_1$ and $\|.\|_{\infty}$ to denote the number of nonzero elements, $\ell_1$ and $\ell_\infty$ norms of a vector or the vectorized version of a matrix, respectively, and $\|.\|_F$ to denote the Frobenius norm of a matrix.   
For a matrix-valued, possibly infinite-order lag polynomial $\mathcal{A}(L) = \sum_{\ell \ge 0} A_\ell L^\ell$, we define $\vertiii{\mathcal{A}}:= \max_{\theta \in [-\pi, \pi]} \| \mathcal{A}(e^{i\theta})\|$, and use $\mathcal{A}_{[k]}(L)$ and $\mathcal{A}_{-[k]}(L)$ to denote the truncated version $\sum_{\ell = 0}^k A_\ell L^\ell$ and the tail series $\sum_{\ell > k} A_\ell L^\ell$, respectively. We also use $\|\mathcal{A}\|_{2,1}$ to denote the sum of the operator norms of its coefficients, $\sum_{\ell \ge 0} \|A_\ell \|$. More generally, for any complex matrix-valued function $f(\theta)$ of frequencies   $\theta \in [-\pi, \pi]$ to $\mathbb{C}^{p \times p}$, we define $\vertiii{f}:= \max_{\theta \in [-\pi, \pi]} \|f(\theta)\|$. 
{In our theoretical analyses, we use $c_i$, $i = 0, 1, 2, \ldots$, to denote universal positive constants whose values do not rely on the model dimensions and parameters.}
For two model dependent positive quantities $A$ and $B$, we also use $A \succsim B$ to  mean that for any universal constant $c > 0$,  we have $A \ge cB$ for sufficiently large sample size. Finally, $A \asymp B$ means $A \succsim B$ and $A \precsim B$.

\begin{remark}[Measures of Dependence]\label{rem:dep-measures} We adopt the spectral density based measures of dependence introduced in \cite{Basu15} to capture the role of temporal dependence in our non-asymptotic error bounds. 
	For a $d$-dimensional centered stationary time series $\{x_t\}_{t \in \mathbb{Z}}$ with autocovariance function $\Gamma_x(h) = \cov(x_t, x_{t+h}) = \mathbb{E}[x_t x_{t+h}^\top]$, $h \in \mathbb{Z}$, we define the spectral density function  
	$f_x(\theta):= \frac{1}{2\pi} \sum_{\ell=-\infty}^{\infty} \Gamma_x(\ell) e^{-i\ell \theta}, \, \, \theta \in [-\pi, \pi]$. 
	The quantity $\vertiii{f_x}$ is taken as a measure of temporal and cross-sectional dependence in the time series $\{x_t\}$.
	For a stable, invertible VARMA process $y_t$ in  \eqref{VARMA} with $\Lambda_{\min}(\Sigma_a) > 0$, it is known that $f_y$ is non-singular on $[-\pi, \pi]$ and there exist two model dependent quantities $\bar{C}>0$ and $\bar{\rho} \in [0,1)$ such that $\| \Pi_{\tau} \| \le \bar{C} \, \bar{\rho}^\tau$, for all integers $\tau \ge 1$ \cite{dufour2005asymptotic}. This implies for any $\tilde{p} \ge 1$, we have  $\| \Pi_{-[\tilde{p}]} \|_{2,1} \le \bar{C} \bar{\rho}^{\tilde{p}} /(1-\bar{\rho})$. The quantities $\vertiii{f_y}, \vertiii{f_y^{-1}}$ and $\| \Pi_{-[\tilde{p}]} \|_{2,1}$ appear in our error bounds, and capture the effects of temporal dependence on the convergence rates. 
\end{remark}

\subsection{Elastic Net with Singular Gram Matrix}\label{subsec:key-results}
Consider an elastic net penalized regression problem where the population covariance matrix of the predictors is singular.
The problem is  non-identifiable in the sense that there is no ``true" coefficient vector.  Rather, the elastic net penalty itself is used to specify an identified target  among all equivalent data-generating models. The following proposition provides deterministic upper bounds on estimation and in-sample prediction errors under some sufficient conditions. The proof is in Appendix C.

\begin{proposition}
	\label{prop:elastic-net-fixed-X-E-v2} 
	Let $\Sigma\in\real^{D\times D}$ be a non-negative definite matrix with $\Lambda_{\min}(\Sigma) = 0$ and let $\rho\in\real^{D}$ be in the column space of $\Sigma$.  
	For some $\alpha \ge 0$, $y,\varepsilon\in\real^N$ and $X\in\real^{N\times D}$,
	consider the linear regression model  $y = X \tba + \varepsilon$ with identified target
	\begin{equation*}
		\tba:= \argmin_{\beta} \left\{ \mathcal{P}_\alpha (\beta) \st \Sigma \beta = \rho \right\},
	\end{equation*}
	where $\mathcal{P}_\alpha(\beta) := \|\beta\|_1 + (\alpha/2) \|\beta\|^2$, and define the estimator
	\begin{equation*}
		\hba := \argmin_{\beta: \|\beta\|_1 \le M } ~  \frac{1}{n} \|y-X \beta\|^2 + \lambda \mathcal{P}_{\alpha } (\beta),
	\end{equation*}
	for some $n$ and $M$, where $M\ge\|\tba\|_1$.  Then for any choice of $\lambda \ge 2\left\|X^\top \varepsilon/n \right\|_\infty$ and $q_n \ge  \left\|X^\top X/n - \Sigma \right\|_\infty $, the following holds:
	\begin{eqnarray*}
		&& \mbox{\textit{(a) {In-Sample} Prediction: }} \frac1{n}\| X\hba-X \tba\|^2\le \lambda \left[2M + \alpha M^2/2  \right], \\
		&& \mbox{\textit{(b) Partially-Identified Estimation: }} \min_{\beta:\Sigma\beta=\rho}\|\hba-\beta\|^2\le \frac{4 q_n M^2 + \lambda \left[ 2M + \alpha M^2 / 2  \right]}{\Lambda_{\min}^+\left(\Sigma \right)},
	\end{eqnarray*}
	where  $\Lambda_{\min}^+(\Sigma)$ is the smallest non-zero eigenvalue of $\Sigma$.
	
	\noindent In addition, define the constrained version of the estimator  
	\begin{equation*}
		\hba_{[C]} := \argmin_{\beta} \left\{ \mathcal{P}_{\alpha} (\beta) \st \frac{1}{n} \|y-X \beta\|^2 \le A_n, ~~ \|\beta\|_1 \le M \right\}.
	\end{equation*}
	%denote the constraint form of the estimator.  
	Then, for any $r_n \ge \frac1{n}\left\|X^\top \varepsilon \right\|_\infty$, and $s_n \ge  \left| \frac1{n}\left\| \varepsilon \right\|^2 - \sigma^2 \right|$, $A_n =\sigma^2 + s_n$ and $M\ge\|\tba\|_1$, we have 
	\begin{eqnarray*}
		\mbox{\textit{(c) Point-Identified Estimation: }} \left\|\hba_{[C]} - \tba\right\|^2 \le 2v_n + 2(\sqrt{{D}}/\alpha+ M)v_n^{1/2}, \nonumber
	\end{eqnarray*}
	where $v_n:=\frac{4Mr_n + 2s_n + 4M^2q_n}{\Lambda_{\min}^+(\Sigma)}$.
\end{proposition}

The VARMA estimator from Phase-II can be expressed in the above regression format (see Equation \eqref{eqn:vec-phase-II}) with $n = T-q$, $N = nd$, $\Sigma = \Sigma_z$ and $D = d^2(p+q)$. We will show that modulo some terms capturing the effect of temporal dependence, $\lambda, q_n, r_n$ can be chosen in the order of at most $O(\sqrt{\log D /n})$ with high probability. 

Under this setting, part (a) will imply in-sample prediction consistency in the high-dimensional regime $\log D/n \rightarrow 0$ as long as the identification target $\beta^{*(\alpha)}$ is \textit{weakly sparse}, i.e. its $\ell_1$-norm grows sufficiently slowly. Consequently, our VARMA forecasts will  asymptotically converge to the optimal forecasts.

Part (b) will ensure that the Euclidean distance of our VARMA estimator from the set of data-generating vectors $\{\beta: \Sigma_z \beta = \rho_{zy} \}$ converges to zero in the asymptotic regime  $\log D /n \rightarrow 0$, assuming weak sparsity of $\beta^{*(\alpha)}$. The rate of convergence also relies on the  curvature of the population loss captured by  $\Lambda_{\min}^+(\Sigma)$. 

Error bound for the point identification part (c) will imply that with an appropriate choice of $s_n$, consistent estimation of our identification target is possible in the double-asymptotic regime ${D}^2\log({D})/n\to0$, as long as $\beta^{*(\alpha)}$ is weakly sparse in the sense of small $\ell_1$-norm. This error bound also increases linearly with the inverse of $\alpha$, the parameter capturing curvature of the penalty function $\mathcal{P}_{\alpha}(\beta)$.

\begin{remark}  
	We focus on prediction and estimation instead of model selection consistency for two reasons. 
	First, model selection consistency in penalized regression holds only under incoherence or irrepresentable conditions \citep{zhao2006model}, which are stringent even for i.i.d.\ data, and are not known to hold with high probability for multivariate  stationary time series data.
	Second, since we work with \textit{an equivalence class of models} potentially having different sparsity patterns, it is not obvious how to define sparsity of a true model, in general.
	However, we have conducted a simulation experiment (Appendix A.3.2) to assess model selection properties of our estimator in finite samples, which shows promising results.
\end{remark}

\subsection{Approximation Error in Phase-I}\label{subsec:error-phase-I} 
Our main interest in this section is in approximating the errors  $a_t$ by the Phase-I residuals $\hat{\varepsilon}_t$ for use in Phase-II. As a by-product, we also provide estimation error bounds for VAR($\infty$) coefficients (see Proposition D.1).

Suppose we re-index data in the form $(y_{-(\tilde{p}-1)}, y_{-(\tilde{p}-2)}, \ldots, y_{-1}, y_0, y_1, \ldots, y_T)$. In Phase-I, we regress $y_t$ on its most recent $\tilde{p}$ lags:
\begin{equation}\label{eqn:inf-var}
	y_t = \sum_{\tau=1}^{\tilde{p}} \Pi_\tau y_{t-\tau} + \varepsilon_t, \mbox{ ~~ where ~~} \varepsilon_t = \left(a_t + \sum_{\tau=\tilde{p}+1}^\infty \Pi_\tau y_{t-\tau} \right).
\end{equation}

The autoregressive design takes the form 
$\mathcal{Y}_{T \times d} = \mathcal{X}_{T \times d \tilde{p}} B_{d \tilde{p} \times d} + E_{T \times d}$, where \linebreak
$\mathcal{Y} = [y_T:y_{T-1}:\ldots:y_1]^\top$, $\mathcal{X} = \left( \left( y_{T-i-j+1}\right)\right)_{1 \le i \le T, 1 \le j \le  \tilde{p}}$, $B = [\Pi_1:\ldots:\Pi_{\tilde{p}}]^\top$ and $E = [\varepsilon_T:\varepsilon_{T-1}:\ldots:\varepsilon_1]^\top$. 
Vectorizing this regression design with $T$ samples and $d^2 \widetilde{p}$ parameters, we have $Y = Z \beta^* + \text{vec}(E)$, where $Y = \text{vec}(\mathcal{Y})$, $Z = I \otimes \mathcal{X}$, and $\beta^* = \text{vec}(B)$. 
In Phase-I, we consider a lasso estimator 
\begin{equation}\label{eqn:l1-ls}
	\hat{\beta} = \argmin_{\beta \in \mathbb{R}^{d^2 \tilde{p}}} \frac{1}{T} \left\|Y - Z \beta \right\|^2 + \lambda \left\| \beta \right\|_1,
\end{equation}
where $\hat{\beta} = \text{vec}(\widehat{B})$ and $\widehat{B} = [\widehat{\Pi}_1: \ldots: \widehat{\Pi}_{\tilde{p}}]^\top$. We denote the residuals of the Phase-I regression as $\hat{\varepsilon}_t = y_t - \sum_{\tau=1}^{\tilde{p}} \widehat{\Pi}_\tau y_{t-\tau}$. 

Our next proposition provides upper bounds on the approximation error of $a_t$ by $\hat{\varepsilon}_t$ for a random realization of ($T+\widetilde{p}$) data points from the VARMA model \eqref{VARMA}. A complete proof is given in Appendix D. 

\begin{proposition}\label{prop:est-error-main}
	Consider any solution $\hat{\beta}$ of \eqref{eqn:l1-ls} using a random realization of $\{ y_t\}_{t=1-\tilde{p}}^T$ from the VARMA model \eqref{VARMA}. Choose $\tilde{p} \asymp T^{\frac{1}{2}-\epsilon}$ for some $\epsilon \in (0, 1/2)$, and $\lambda \ge \lambda_0$, where 
	$$
	\lambda_0 := 2\pi   \vertiii{f_y} \left[ 3A \, \max \left\{ \vertiii{\Pi_{[\tilde{p}]}}^2, \, 1 \right\} \sqrt{\log (d^2\tilde{p})/T} + \left\| \Pi_{-[\tilde{p}]} \right\|_{2,1} \right], \mbox{ for some } A > 1.
	$$
	Then, for $T \succsim \log d^2\tilde{p}$, there exist universal constants $c_i > 0$ such that with probability at least $1 - c_0 \exp \left[-(c_1 A^2 - 2) \log d^2 \tilde{p} \right]$,
	\begin{eqnarray*}
		\frac{1}{T} \sum_{t = 1}^T \left\| \hat{\varepsilon}_t - \varepsilon_t \right\|^2 &\le& \Delta^2_\varepsilon:= 2 \lambda \sum_{\tau=1}^{\tilde{p}} \left\| \Pi_\tau \right\|_1, \\
		\max_{1 \le j \le d} \, \,  \frac{1}{T} \sum_{t=1}^T \left( \hat{\varepsilon}_{tj} - a_{tj}\right)^2 &\le&  \Delta_a^2:= 4 \max \left\{\Delta_\varepsilon^2, 4 \pi \left\| \Pi_{-[\tilde{p}]} \right\|^2_{2,1} \vertiii{f_y} \right\}, \\
		\frac{1}{T} \sum_{t=1}^T \left\| \hat{\varepsilon}_t - a_t\right\|^2 &\le& 4 \max \left\{\Delta_\varepsilon^2, 4 \pi d \left\| \Pi_{-[\tilde{p}]} \right\|^2_{2,1} \vertiii{f_y} \right\}.
	\end{eqnarray*}
	If, in addition, $\{\Pi_1, \ldots, \Pi_{\tilde{p}} \}$ are sparse so that $k:= \sum_{\tau=1}^{\tilde{p}} \left\| \Pi_{\tau} \right\|_0 \precsim T$, then for any choice of $\lambda \ge 2 \lambda_0$ and $T \succsim \max \{\tilde{p}^2 \vertiii{f_y}^2 \vertiii{f_y^{-1}}^2, 1 \} k (\log d + \log \tilde{p})$, we can use a potentially tighter upper bound $\Delta_\varepsilon^2:= (128/\pi) \vertiii{f_y^{-1}} \, {k} \lambda^2$. 
\end{proposition}

\begin{remark}[Convergence Rate \& Truncation Bias]\label{rem:ph1-rem1}
	The error bounds $\Delta^2_{\varepsilon}$ and $\Delta_a^2$ scale with $\lambda_0$, which has two terms. 
	%We will focus on the 
	The first term decays polynomially with $T$. The second term $\left\| \Pi_{-[\tilde{p}]} \right\|_{2,1}$ captures the \textit{truncation bias} arising from using a VAR($\widetilde{p}$) approximation to a VAR($\infty$) process. When  $\widetilde{p} \asymp T^{\frac{1}{2} - \epsilon}$, this term decays exponentially with $T^{\frac{1}{2}-\epsilon}$ since
	\begin{equation}
		\left\| \Pi_{-[\tilde{p}]} \right\|_{2,1} \le \frac{\bar{C}}{1-\bar{\rho}} \bar{\rho}^{\tilde{p}} = \frac{\bar{C}}{1-\bar{\rho}} \exp \left[ - T^{\frac{1}{2}-\epsilon} \log(1/\bar{\rho}) \right],
	\end{equation}
	where $\bar{C}, \, \bar{\rho}$ are as defined in Remark \ref{rem:dep-measures}. 
	This bias also appears in our Phase-II analysis.
\end{remark}
\begin{remark}[Choice of $\tilde{p}$, Slow \& Fast Rates, and RE Condition]\label{rem:ph1-rem2}
	As long as $\widetilde{p}$ increases polynomially fast with $T$, the truncation bias vanishes as $T \rightarrow \infty$ and the approximation errors $\Delta_\varepsilon$ and $\Delta_a$ decay with $T$ at a rate $O(\sqrt{\log d / T})$. However,  under sparsity of $\Pi$ and choosing $\tilde{p} \asymp T^{1/2 \, - \, \epsilon}$,  a suitable Restricted Eigenvalue (RE) condition holds with high probability (see Appendix D for details), and these approximation errors decay at a faster rate $O(\log d / T)$. The choice of $(1/2 - \epsilon)$ in the exponent ensures that $T \succsim \tilde{p}^2$ holds asymptotically. This choice of $\widetilde{p}$ matches with low-dimensional VARMA analysis presented in \citep{dufour2005asymptotic}.
\end{remark}

\subsection{Prediction and Estimation Error in Phase-II}\label{subsec:error-phase-II}
For simplicity of exposition, we assume that $p$ and $q$ are known and $\widetilde{p} > p+q$. It will be evident from our analysis that similar conclusions hold as long as we replace these with any upper bounds of $p$ and $q$.  Without loss of generality, we also assume that the Phase-II regressions are run with the following re-indexing of  observations:
\begin{equation}\label{eqn:phase2-reformat}
	y_t = \sum_{\ell=1}^p \Phi_\ell y_{t-\ell} + \sum_{m=0}^q \Theta_m \hat{\varepsilon}_{t-m} + u_t, ~~~~\mbox{ for } t = 1, 2, \ldots, n,~~~~n = T-q,
\end{equation}
where $u_t = \Theta(L)(a_t - \hat{\varepsilon}_t)$, and $\Theta_0 = I$. 
As mentioned earlier, we consider a variant of the Phase-II regression where the fitted values from Phase-I, $\hat{y}_t = y_t - \hat{\varepsilon}_t$, are used as response instead of $y_t$. The autoregressive moving average design then takes the form 
\begin{equation*}
	\underbrace{\left[\begin{array}{c} \hat{y}_n^\top \\ \hat{y}_{n-1}^\top \\ \vdots \\ \hat{y}_1^\top \end{array}\right]}_{\mathcal{Y}_{n \times d}}
	= 
	\underbrace{\left[ \begin{array}{cccccc} 
			y_{n-1}^\top & \ldots & y_{n-p}^\top & \hat{\varepsilon}_{n-1}^\top & \ldots & \hat{\varepsilon}_{n-q}^\top \\ 
			y_{n-2}^\top & \ldots & y_{n-1-p}^\top & \hat{\varepsilon}_{n-2}^\top & \ldots & \hat{\varepsilon}_{n-1-q}^\top \\ 
			\vdots & \vdots & \vdots & \vdots & \vdots & \vdots \\
			y_{0}^\top & \ldots & y_{1-p}^\top & \hat{\varepsilon}_{0}^\top & \ldots & \hat{\varepsilon}_{1-q}^\top 
		\end{array} \right]}_{\tilde{\mathcal{Z}}_{n \times d(p+q)}}
	\underbrace{
		\left[ \begin{array}{c}
			\Phi^\top \\ \Theta^\top
		\end{array} \right]}_{B_{d(p+q) \times d}} + 
	\underbrace{\left[\begin{array}{c} u_n^\top \\ \vdots \\ u_1^\top \end{array} \right]}_{\mathcal{U}_{n \times d}},
\end{equation*}

where $\Phi = [\Phi_1:\ldots:\Phi_p]$, and $\Theta:= [\Theta_1:\ldots:\Theta_q]$. 
Vectorizing the above regression problem with $n$ samples and $d^2 ({p+q})$ parameters, we have 
\begin{equation}\label{eqn:vec-phase-II}
	\underbrace{\text{vec}(\mathcal{Y})}_{Y} = \underbrace{\left(I \otimes \tilde{\mathcal{Z}} \right)}_{\tilde{Z}} \underbrace{\text{vec}(B)}_{\beta^*} + \underbrace{\text{vec}(\mathcal{U})}_{U}.
\end{equation}

In order to apply Proposition \ref{prop:elastic-net-fixed-X-E-v2} on this regression problem with $N = nd$ and $D = d^2(p+q)$, we first provide suitable choices of $q_n, s_n$ and $r_n$ (same as choice of $\lambda$) that hold with high probability for a random realization of $(T+\tilde{p})$ consecutive observations from the VARMA process. To this end, note that $\left\| \left( I \otimes \tilde{\mathcal{Z}} \right)^\top \left( I \otimes \tilde{\mathcal{Z}} \right)/n - I \otimes \Sigma_z \right\|_\infty = \left\|  \tilde{\mathcal{Z}}^\top \tilde{\mathcal{Z}}/n - \Sigma_z \right\|_\infty$.

In Section \ref{subsec:error-phase-I}, we have discussed how the approximation errors $\Delta_a$,  $\Delta_\varepsilon$ and the truncation bias term $\| \Pi_{-[\tilde{p}]} \|_{2,1}$ decay with the sample size. In this proposition, we show that $q_n, r_n$ and $s_n/d$ can be chosen to be a linear combination of the above terms and $\sqrt{\log d^2(p+q)/n}$, where the coefficients of this linear combination depend on model parameters and capture the role of temporal dependence in these convergence rates.

\begin{proposition}\label{prop:phase2-qrs}
	Consider the Phase-II regression 
	\eqref{eqn:vec-phase-II} 
	with design matrix $I \otimes \tilde{\mathcal{Z}}$ and error vector $\text{vec}(\mathcal{U})$. Set $\sigma^2_j = e_j^\top \var\left( \Theta(L) \Pi_{-[\tilde{p}]}(L)y_t \right) e_j$, for $j=1, \ldots, d$. Then there exist universal constants $c_i > 0$ such that the event  \begin{equation}\label{eqn:E-varma-phase2}
		\mathcal{E}:= \left\{ 
		\left\| \tilde{\mathcal{Z}}^\top \tilde{\mathcal{Z}}/n - \Sigma_z \right\|_\infty \le q_n, \frac{1}{n} \left\| \tilde{\mathcal{Z}}^\top \mathcal{U} \right\|_\infty \le r_n, 
		\left| \frac{1}{n} \left\| \text{vec}(\mathcal{U}) \right\|^2 - \sum_{j=1}^d \sigma^2_j \right| \le s_n
		\right\}
	\end{equation}
	holds with probability at least $1 - c_0 \exp \left[-(c_1 A^2 - 2) \log d^2(p+q) \right]$, where 
	\begin{eqnarray*}
		q_n &=& \varphi_{q,1} \sqrt{\frac{\log \, d^2(p+q)}{n}} + \varphi_{q,2} \left( \Delta_a + \Delta_a^2 \right),  \\
		r_n &=& \varphi_{r,1} \sqrt{\frac{\log \, d^2(p+q)}{n}} + \varphi_{r,2} \left( \Delta_\varepsilon + \Delta_\varepsilon^2 + \left\|  \Pi_{-[\tilde{p}]} \right\|_{2,1}\right), \\
		s_n/d &=& \varphi_{s,1} \sqrt{\frac{\log \, d^2(p+q)}{n}} + \varphi_{s,2} \left( \Delta_\varepsilon + \Delta_\varepsilon^2 \right),
	\end{eqnarray*}
	and $\varphi_{q,1}, \varphi_{q,2}, \varphi_{r,1}, \varphi_{r,2}, \varphi_{s,1}, \varphi_{s,2}$ are functions of the model parameters
	\begin{eqnarray*}
		\varphi_{q,1} &=& 2 \pi \vertiii{f_y} \left(p+q \vertiii{\Pi_{[\tilde{p}]}}^2 \right)^2, \\
		\varphi_{q,2} &=& \max \, \{2q, 2 \sqrt{2\pi q} \vertiii{f_y}^{1/2} \left(p+q \vertiii{\Pi_{[\tilde{p}]}}^2 \right)^{1/2} \}, \\
		\varphi_{s,1} &=& 2 \pi \vertiii{\Theta} \left\| \Pi_{-[\tilde{p}]} \right\|_{2,1}^2 \vertiii{f_y},  \\
		\varphi_{s,2} &=& \max \, \left\{ 2 \| \Theta\|_{2,1}^2, 4 \sqrt{2 \pi} \vertiii{\Theta}^{1/2} \left\| \Pi_{-[\tilde{p}]} \right\|_{2,1} \vertiii{f_y}^{1/2} \left\|  \Theta \right\|_{2,1}\right\}, \\
		\varphi_{r,1} &=& c_1 \vertiii{f_y} A \,  \max \left\{1, \vertiii{\Theta}^2 \left\| \Pi_{-[\tilde{p}]} \right\|_{2,1}^2, \vertiii{\Pi_{[\tilde{p}]}}^2 \right\}, \\
		\varphi_{r,2} &=& c_2 \vertiii{f_y} \left\|  \Theta \right\|_{2,1} \max \! \{ 1, \left\| \Pi_{[\tilde{p}]} \right\|_{2,1} \}.
	\end{eqnarray*}
\end{proposition}

\noindent Using Proposition  \ref{prop:yule-walker-identification}, the identification target in \eqref{Optim_pair} with an elastic net penalty becomes 
\begin{equation} \label{Optim_pair_alt_rep}
	({  \Phi}^{(\alpha)}, {  \Theta}^{(\alpha)})  =  \underset{{  \Phi}, {  \Theta}}{\operatorname{argmin}} \ \left \{\| [ \Phi: \Theta] \|_1 + \dfrac{\alpha}{2} \| [\Phi:\Theta]\|_F^2 \st  \text{vec}(\rho_{zy}) = (I \otimes \Sigma_z) \text{vec}(\beta) \right \}, 
\end{equation}
where $\rho_{zy}, \Sigma_z$ and $\beta$ are as defined in Proposition \ref{prop:yule-walker-identification}.
We consider the penalized and constrained versions  of the estimator
\begin{eqnarray*}
	&& \text{vec}\left([\hat{\Phi}^{(\alpha)}: \hat{\Theta}^{(\alpha)}]^\top \right) = \argmin_{\left \| \beta \right\|_1 \le M} \, \frac{1}{n} \left\| \text{vec}(\mathcal{Y}) - (I \otimes \tilde{\mathcal{Z}}) \beta \right\|^2 + \lambda \, \mathcal{P}_{\alpha} (\beta) \\
	&& \text{vec}\left([\hat{\Phi}^{(\alpha)}_{[C]}: \hat{\Theta}^{(\alpha)}_{[C]}]^\top \right) = \argmin_{\left \| \beta \right\|_1 \le M} \, \left\{  \mathcal{P}_{\alpha} (\beta) \mbox{ s.t. } \frac{1}{n} \left\| \text{vec}(\mathcal{Y}) - (I \otimes \tilde{\mathcal{Z}}) \beta \right\|^2 \le A_n \right\}.
\end{eqnarray*}

A direct application of Proposition \ref{prop:elastic-net-fixed-X-E-v2} with the choices of $q_n r_n, s_n$ in Proposition \ref{prop:phase2-qrs} then leads to  the following upper bounds on the prediction and estimation error of the penalized and constrained versions of our two-phase VARMA estimator. 
\begin{proposition}[VARMA Estimation and Prediction Errors]\label{prop:est-error-phase-II}
	Consider a random realization of $T+\tilde{p}$ consecutive observations $\{y_1, \ldots, y_{T+\tilde{p}} \}$ from a stable, invertible Gaussian VARMA model \eqref{VARMA}, and let $n = T-q$ denote the sample size in Phase-II. 
	Denote $K_y := \max \{\vertiii{f_y}, 
	\left\| \Pi \right\|_{2,1},  \left\| \Theta^{(\alpha)} \right\|_{2,1} \}$. \\
	\noindent (a) \underline{Forecast Error}: Let $y_t^* = \sum_{\ell=1}^p \Phi_\ell y_{t-\ell} + \sum_{m=1}^q \Theta_m a_{t-m}$ and $\tilde{y}_t = \sum_{\ell=1}^p \widehat{\Phi}_\ell y_{t-\ell} + \sum_{m=1}^q \widehat{\Theta}_m \hat{\varepsilon}_{t-m}$ denote the optimal and the penalized VARMA forecasts respectively. Then, for a choice of  $\lambda \asymp  K_y^3 \max \left\{ \sqrt{\log d^2(p+q) \, /n }, \Delta_\varepsilon \right\}$, and  $M \ge \| \Phi^{(\alpha)} \|_1 + \|\Theta^{(\alpha)} \|_1$ for some   $\alpha \ge 0$,
	\begin{equation*}
		\frac{1}{n} \sum_{t=1}^n \left\| \tilde{y}_t - y^*_t\right\|^2 =  O_\mathbb{P} \left(K_y^3 M^2   \,  \max \left\{\sqrt{\frac{\log d^2 (p+q)}{n}}, \| \Pi_{-[\tilde{p}]} \|_{2,1} , \Delta_{\varepsilon}\right\} \right).
	\end{equation*}

	\noindent (b) \underline{Partially-identified Estimation}:  With the same choice of $\lambda$, $M$ and $\alpha$ in (a), the penalized estimator is partially identified and satisfies
	\begin{equation*}
		\small    \min_{(\Phi, \Theta) \in \mathcal{E}_{p,q}(\Pi(L))} \, \left \| \left(\wh{\Phi}^{(\alpha)}, \wh{\Theta}^{(\alpha)} \right) - \left({\Phi}, {\Theta} \right)\right \|_F^2 = O_\mathbb{P} \left(  \frac{K_y^3 M^2 }{ \Lambda_{\min}^+(\Gamma_z(0))} \,  \max \left\{\sqrt{\frac{\log d^2 (p+q)}{n}}, \| \Pi_{-[\tilde{p}]} \|_{2,1} , \Delta_{\varepsilon}\right\} \right).
	\end{equation*}
	
	\noindent (c) \underline{Point-identified Estimation}: For a choice of $A_n \asymp K_y^3 \| \Pi_{-[\tilde{p}]} \|_{2,1}^2 \max \{d \sqrt{\log d^2(p+q)/n}, \Delta_\varepsilon \}$ and any $\alpha > 0$, the constrained version of the estimator is point identified and satisfies
	\begin{equation*}
		\left \| \left(\wh{\Phi}^{(\alpha)}_{[C]}, \wh{\Theta}^{(\alpha)}_{[C]} \right) - \left({\Phi}^{(\alpha)}, {\Theta}^{(\alpha)} \right)\right \|_F^2 = O_\mathbb{P} \left(  \frac{K_y^3 M^2 }{ \alpha  \, \sqrt{\Lambda_{\min}^+(\Gamma_z(0))}} \,  \max \left\{d^3 \sqrt{\frac{\log d^2 (p+q)}{n}}, \| \Pi_{-[\tilde{p}]} \|_{2,1} , \Delta_{\varepsilon}\right\}^{1/2} \right).
	\end{equation*}
\end{proposition}

Part (a) of this proposition ensures that as long as the identification target is parsimonious in the sense of small $\ell_1$-norm and the penalty parameter is chosen appropriately, the VARMA forecasts converge to the optimal forecasts (which uses any element from the equivalence class $\mathcal{E}_{p,q}(\Pi)$) in the asymptotic regime $\log d /n \to 0$. The truncation bias term $\| \Pi_{-[\tilde{p}]}\|_{2,1}$ and the approximation error from Phase-I $\Delta_\varepsilon$ also converges to zero in this asymptotic regime, as shown in Section \ref{subsec:error-phase-I}. The convergence rates are further affected by the strength of temporal dependence in the VARMA process, as captured by the term $K_y$.

In addition, part (b) ensures that the distance of our penalized estimator from the equivalence class also asymptotically vanishes in this high-dimensional regime. Further, the convergence rates are affected by the minimum positive eigenvalue of the variance-covariance matrix of the process $z_t$, which captures the curvature of the loss function.

Part (c) shows that our constrained estimator converges in probability to our identification target, but in a low-dimensional regime $d^3 \sqrt{\log d}/n \to 0$. This slow rate is a consequence of the fact that we did not assume sparsity on the entire equivalence class $\mathcal{E}_{p,q}(\Pi)$, so searching for the correct identification target within this equivalence class still has a complexity of the order of $d^2$. The tuning parameter $\alpha$ also affects the convergence rate, since this captures the degree of curvature of the term $\mathcal{P}_\alpha(.)$ in the loss function. However, taking a sequence of $\alpha_n$ that converges to $0$ at a rate slower than $d^3 \sqrt{\log d^2(p+q)/n}$, we can still guarantee consistent estimation of the target $(\Phi^{(0)}, \Theta^{(0)})$ with the minimum Frobenius norm. 
	
	\section{\label{applications}Forecast Applications}
	We present three forecast applications:
	
	\textbf{(i) Demand forecasting.} Weekly sales data (in dollars) are collected for $d=16$ product categories of Dominick's Finer Foods from January 1993 to July 1994 ($T=76$). Data are taken from \url{https://research.chicagobooth.edu/kilts/marketing-databases/dominicks}. 
	To ensure stationarity, we take each series in log differences and consider sales growth. 
	Augmented Dickey-Fuller tests help support that the sales growth series are stationary.
	
	\textbf{(ii) Volatility forecasting.}  We collect monthly realized variances for $d=17$ stock market indices, from January 2009 to December 2016  ($T=96$).
	Realized variances, computed from five minute returns, are obtained from \url{http://realized.oxford-man.ox.ac.uk/data/download}
	and log-transformed following standard practice.
	Augmented Dickey-Fuller tests help support  that the log-realized variances are stationary. 
	
	\textbf{(iii) Macro-economic forecasting}.  We consider $d=168$ quarterly macro-economic series of length $T=60$ ending in 2008, Quarter 4. 
	Data are taken from \href{http://qed.econ.queensu.ca/jae/2013-v28.2/koop/}{the Journal of Applied Econometrics Data Archive,}
	a full list of the series is available in \cite{Koop13} (Data Appendix), along with the transformations to make them approximately stationary. 
	
	In all considered cases, the number of time series $d$ is large relative to the time series length $T$. 
	First, we discuss the model parsimony of the estimated VARMA and VAR with HLag penalties.
	Secondly, we compare their forecast accuracy for different forecast horizons.
	
	\subsection{Model Parsimony} 
	Since the sparse VARMA and VAR estimators with HLag penalties both perform automatic lag selection, they give information on the effective maximum AR and MA orders. Consider the $d \times d$ moving average lag matrix $\widehat{L}_{\widehat{{  \Theta}}}$ of the estimated VARMA  whose elements are 
	$\widehat{L}_{\widehat{{  \Theta}}, ij} = \text{max}\{ m: \widehat{{\Theta}}_{m, ij}\neq 0\},$ 
	where  $\widehat{L}_{\widehat{{  \Theta}}, ij}=0$ if $\widehat{{\Theta}}_{m, ij}=0$ for all $m=1\ldots, {\hat{q}}$. This lag matrix shows the maximal MA lag for each series $j$ in each equation $i$ of the corresponding estimated VARMA. 
	If entry $ij$ is zero, this means that all lagged MA coefficients of time series $j$ on time series $i$ are estimated as zero.
	If entry $ij$ is, for instance, three, this means that the third lagged moving average term of  series $j$ on  series $i$ is estimated as non-zero, but the forth and higher as all zero.
	Similarly, one can construct the autoregressive lag matrix $\widehat{L}_{\widehat{{  \Phi}}}$ of the estimated VARMA  and the autoregressive lag matrix
	$\widehat{L}_{\widehat{{  \Pi}}}$ of the estimated VAR.
	
	\begin{figure}[t]
		\centering
		\includegraphics[width=\textwidth]{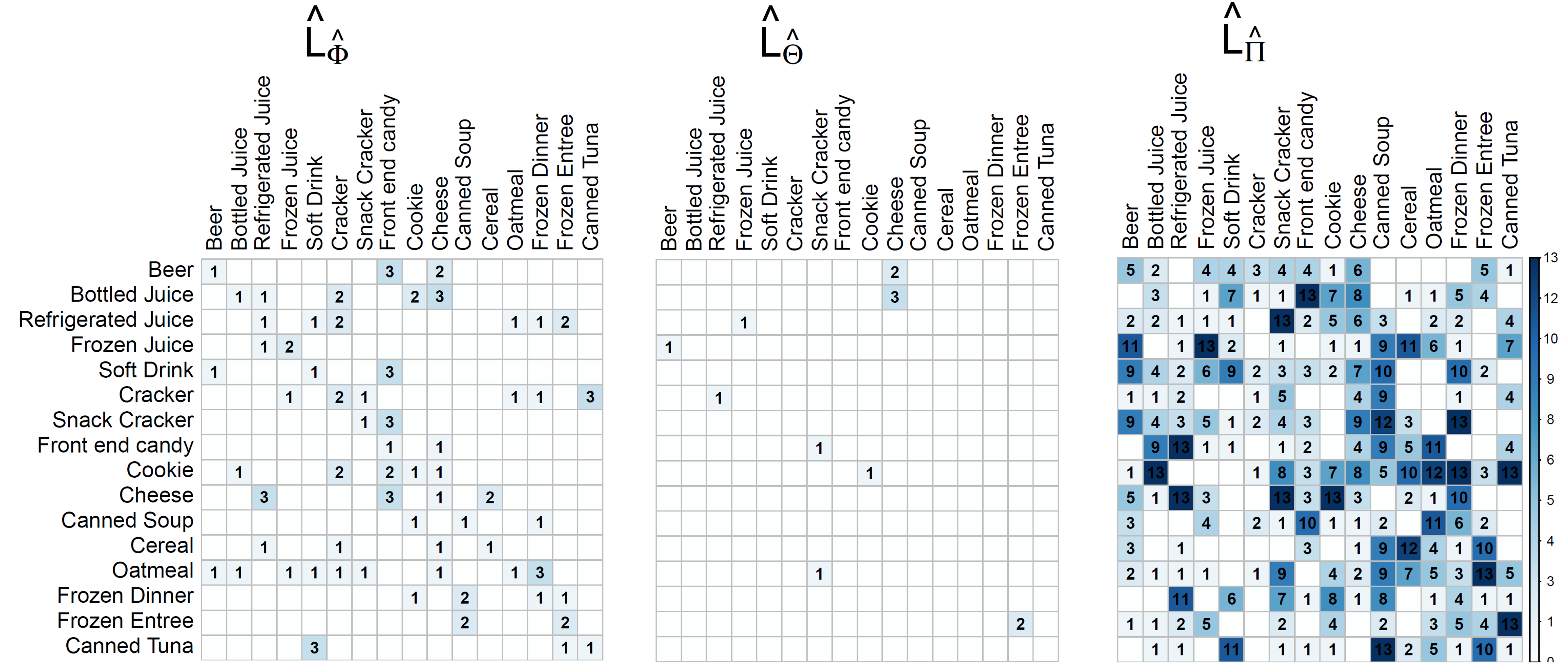}
		\caption{Demand data set: AR-lag matrix (left) and  MA-lag matrix (middle) of the estimated VARMA, and AR-lag matrix of the estimated VAR (right). \label{Demand_lhat}}
	\end{figure}

	Figure \ref{Demand_lhat} shows the lag matrices of the estimated VARMA and VAR on the demand data. 
	Similar findings are obtained for the other data sets and therefore omitted.
	The MA lag matrix of the  VARMA (middle panel) is very sparse: 
	247 out of 256 entries are equal to zero.    
	By adding just 
	few MA terms to the model, serial correlation in the error terms is captured.
	As a result, a more parsimonious VARMA model is obtained: 107 out of the
	$3,\!072$ (around 3\%) estimated VARMA  parameters  are non-zero. 
	In contrast, 877 out of the  
	$3,\!328$ (around 25\%)  estimated VAR parameters are non-zero.
	We find the more parsimonious VARMA to often give more accurate forecasts than the VAR, as discussed next.
	
	\subsection{Forecast Accuracy} 
	We compare the forecast accuracy of  VARMA to VAR  through an  expanding window forecast exercise. Let $h$ be the forecast horizon. At each time point $t=S,\ldots, T-h$, we sparsely estimate the VARMA and VAR. We take $S$ such that forecasts are computed for the last 25\% of  observations. We estimate the model on the standardized series and obtain $h$-step-ahead  forecasts and corresponding forecast errors $e_{i,t+h}^{(i)}=y_{i,t+h}-\widehat{y}_{i,t+h}$ for each series $1\leq i \leq d$. The overall forecast performance is measured by computing the Mean Squared Forecast Error for a particular forecast horizon $h$, as in Equation \eqref{MSFE_eq}.
	For the weekly marketing data set, we take 
	$h=1, 8, 13$.
	For the monthly volatility data set, we take $h=1, 6, 12$.
	For the quarterly macro-economic data set, we take $h=1, 4, 8$. 
	To assess the difference in forecast performance between VARMA and VAR, we use a Diebold-Mariano (DM-) test (\citep{Diebold95}).

	\begin{table}
		\caption{Mean Squared Forecast Errors  at different forecast horizons for the two estimators on the three data sets. $P$-values of the Diebold-Mariano tests are given in parentheses. 
			\label{MSFEApplications}}
		\centering
		\begin{tabular}{lcccccccccccc} \hline 
			Estimator && \multicolumn{3}{c}{Weekly}   && \multicolumn{3}{c}{Monthly}   && \multicolumn{3}{c}{Quaterly}	\\
			&& \multicolumn{3}{c}{Demand Data}   && \multicolumn{3}{c}{Volatility Data}   && \multicolumn{3}{c}{Macro-economic Data} \\
			&& $h=1$ &$h=8$ & $h=13$ && $h=1$ &$h=6$ & $h=12$ && $h=1$ &$h=4$ & $h=8$\\ \hline 
			VARMA && 0.473 & 0.578 & 0.550 && 0.781 & 1.080 & 1.065 && 0.974 & 1.152 & 1.281 \\
			VAR && $\underset{(0.141)}{0.499}$ & $\underset{(0.041)}{0.703}$ & $\underset{(<0.001)}{0.715}$ && $\underset{(0.142)}{0.728}$ & $\underset{(0.050)}{1.209}$ & $\underset{(0.007)}{1.429}$ && $\underset{(0.412)}{0.977}$ & $\underset{(0.080)}{1.170}$ & $\underset{(0.003)}{1.401}$ \\ \hline 
		\end{tabular}
	\end{table}
	
	The MSFEs 
	on the three data sets 
	are given in
	Table \ref{MSFEApplications}. 
	Across all considered data sets and  
	horizons, VARMA gives either a significantly lower MSFE than the VAR estimator (in 5 out of 9 cases  at the 5\% level, in 1 case  at the 10\% level) 
	or  performs equally well  (in 3 out of 9 cases). 
	The gain in forecast accuracy over VAR is typically the largest for the longest forecast horizons.
	VARMA not only gives a lower MSFE averaged over the considered time points, but it also attains the lowest MSFE for the large majority of time points. 
	For the demand data at  horizon $h=13$, for instance, it outperforms VAR for all time points except two.
	The sparse VARMA method is thus a valuable
	addition to the forecaster's toolbox for large-scale multivariate time series models. It exploits the serial correlation between the error terms and, as a consequence, often gives more parsimonious forecast models with competitive or better forecast accuracy than  a sparse VAR. 
	
	\section{\label{conclusion}Conclusion}
	We present sparse identification and estimation  for  VARMA models.
	Our estimator, available in the \texttt{R} package \texttt{bigtime}, 
	is naturally aligned with our identified target through the use of sparsity-inducing convex regularizers and can be computed efficiently even for large-scale VARMAs. 
	Under a double-asymptotic regime where both $d, T \rightarrow \infty$, we prove consistency of our two-step sparse VARMA estimation for stable, invertible Gaussian VARMA processes. 
	Simulation and real data analyses show that our sparse VARMA model can produce better forecasts compared to sparse VAR by fitting  more parsimonious models. 
	
	There are several questions we did not address. 
	Our two-stage 
	procedure can be generalized to an iterative method, as in \citep{dias2018}. However, developing a double-asymptotic theory for such an iterative method is complex and left for future research.
	The convergence rates of our point-identified Phase-II estimator can be potentially sharpened under restricted eigenvalue  assumptions. Identifying a class of sparse VARMAs for which such assumptions hold with high probability is an  interesting theoretical question. Inference of model parameters can be pursued by adopting debiasing approaches \citep{javanmard2014confidence, van2014asymptotically}, and are left for future research.

\section*{Acknowledgments}
We thank the editor and reviewers for their thorough review and highly appreciate their comments which substantially improved the quality of the manuscript.
The authors wish to thank Profs.\ Christophe Croux, George Michailidis, Suhasini Subba Rao and Ruey S.\ Tsay for stimulating discussions and helpful comments.
IW was supported by  the European Union's Horizon 2020 research and innovation programme under the Marie Sk\l{}odowska-Curie grant agreement  No 832671.
SB was supported by NSF award DMS-1812128 and NIH awards 1R01GM135926-01 and 1R21NS120227-01. 
JB was supported by an NSF CAREER award (DMS-1748166).
DSM was supported by NSF (1455172, 1934985, 1940124, 1940276), Xerox PARC, the Cornell University Atkinson Center for a Sustainable Future (AVF-2017), USAID, and the Cornell University Institute of Biotechnology \& NYSTAR. 
\bigskip

\newpage
\begin{appendices}
	\setcounter{figure}{0}
	\setcounter{table}{0}
	\renewcommand{\thefigure}{A\arabic{figure}}
	\renewcommand{\thetable}{A\arabic{table}}
	\clearpage
\section*{Supplement to ``Sparse Identification and Estimation of Large-Scale Vector AutoRegressive Moving Averages"}
We present the proofs of sparse identification in Section \ref{app:identification}. Proofs of key technical ingredients required for Phase-I and II analyses are in Section \ref{sec:technical-lemmas}, along with some additional lemmas to control the error due to using $\hat{\varepsilon}_t$ instead of $\varepsilon_t$ in Phase-II. 
Sections \ref{sec:pf-elnet}, \ref{app:phase-I} and \ref{app:phase-II} contain results for error bound analysis in elastic net, Phase-I and II, respectively. Section \ref{Algorithmdetail} contains details of Phase-I and II algorithms. 
Section \ref{Simulation} presents the results on several numerical experiments. 
\bigskip

\textit{Notations. } 
We denote the sets of integers, real, and complex numbers by $\mathbb{Z}$, $\mathbb{R}$, and $\mathbb{C}$, respectively. We use $\|.\|$ to denote the Euclidean norm of a vector and the operator norm of a matrix. We reserve $\|.\|_0$, $\|.\|_1$ and $\|.\|_{\infty}$ to denote the number of nonzero elements, $\ell_1$ and $\ell_\infty$ norms of a vector or the vectorized version of a matrix, respectively, and $\|.\|_F$ to denote the Frobenius norm of a matrix. The symbol $\mathbb{S}^{d-1}$ is used to denote the vectors $v \in \mathbb{R}^d$ with $\|v\|=1$. We use $\Lambda_{\max}(.)$ and $\Lambda_{\min}(.)$ to denote the maximum and minimum eigenvalues of a (symmetric or Hermitian) matrix. We use $|.|$ to denote the absolute value of a real number or complex number. We use $V^*$ to denote the conjugate transpose of a complex matrix, vector or scalar $V$. For a matrix-valued, possibly infinite-order lag polynomial $\mathcal{A}(L) = \sum_{\ell \ge 0} A_\ell L^\ell$, we define $\vertiii{\mathcal{A}}:= \max_{\theta \in [-\pi, \pi]} \| \mathcal{A}(e^{i\theta})\|$, and use $\mathcal{A}_{[k]}(L)$ and $\mathcal{A}_{-[k]}(L)$ to denote the truncated version $\sum_{\ell = 0}^k A_\ell L^\ell$ and the tail series $\sum_{\ell > k} A_\ell L^\ell$, respectively. We also use $\|\mathcal{A}\|_{2,1}$ to denote the sum of the operator norms of its coefficients, $\sum_{\ell \ge 0} \|A_\ell \|$. More generally, for any complex matrix-valued function $f$ of frequencies from  $[-\pi, \pi]$ to $\mathbb{C}^{p \times p}$, we define $\vertiii{f}:= \max_{\theta \in [-\pi, \pi]} \|f(\theta)\|$. In our theoretical analyses, we use $c_i$, $i = 0, 1, 2, \ldots$, to denote universal positive constants whose values do not rely on the model dimensions and parameters. Their values are allowed to change from equation to equation. For example, we will use $c_0$ instead of $2c_0, c_0+2$ etc. within a proof to keep the notations simple.  For two model dependent positive quantities $A$ and $B$, we also use $A \succsim B$ to  mean that for any universal constant $c > 0$,  we have $A \ge cB$ for sufficiently large sample size. Finally, $A \asymp B$ means $A \succsim B$ and $A \precsim B$.

\textit{Measures of Dependence}. We adopt the spectral density based measures of dependence introduced in \cite{Basu15} to conduct our non-asymptotic analysis. For a $d$-dimensional centered stationary Gaussian time series $\{x_t\}_{t \in \mathbb{Z}}$ with autocovariance function $\Gamma_x(h) = \cov(x_t, x_{t+h}) = \mathbb{E}[x_t x_{t+h}^\top]$, $h \in \mathbb{Z}$, we assume the spectral density function  
$f_x(\theta):= \frac{1}{2\pi} \sum_{\ell=-\infty}^{\infty} \Gamma_x(\ell) e^{-i\ell \theta}, \, \, \theta \in [-\pi, \pi]$,
exists, is non-singular a.e. on $[-\pi, \pi]$, and $\vertiii{f_x} < \infty$. The quantity $\vertiii{f_x}$ %, defined at the end of Section 1, 
is taken as a measure of temporal and cross-sectional dependence in the time series $\{x_t\}$. We say that the time series $x_t$ is \textit{stable} if $\vertiii{f_x} < \infty$. More generally, for any pair of $d$-dimensional centered, stable time series $\{ x_t\}$ and $\{y_t \}$, the cross-spectral density is defined as 
$f_{x,y}(\theta) = \frac{1}{2\pi} \sum_{\ell=-\infty}^{\infty} \Gamma_{x, y}(\ell) e^{-i \ell \theta},$ 
where $\Gamma_{x,y}(h) = \cov(x_t, y_{t+h})$, for $h \in \mathbb{Z}$. If the joint process $w_t = [x_t^\top, y_t^\top]^\top$ is stable, i.e. it satisfies $\vertiii{f_w} < \infty$, it follows that $\vertiii{f_{x,y}}^2 \le \vertiii{f_x} \vertiii{f_y}$.

For a stable, invertible VARMA process $y_t$ in  \eqref{VARMA} with $\Lambda_{\min}(\Sigma_a) > 0$, it is known that $f_y$ is non-singular on $[-\pi, \pi]$ and there exist two model dependent quantities $\bar{C}>0$ and $\bar{\rho} \in [0,1)$ such that $\| \Pi_{\tau} \| \le \bar{C} \, \bar{\rho}^\tau$, for all integers $\tau \ge 1$. This implies for any $\tilde{p} \ge 1$, we have  $\| \Pi_{-[\tilde{p}]} \|_{2,1} \le \bar{C} \bar{\rho}^{\tilde{p}} /(1-\bar{\rho})$. The quantities $\vertiii{f_y}, \vertiii{f_y^{-1}}$ and $\| \Pi_{-[\tilde{p}]} \|_{2,1}$ appear in our error bounds, and captures the effects of temporal dependence on the convergence rates.

\section{Proofs for Sparse Identification}\label{app:identification}
\subsection{Yule-Walker type Equations for VARMA \label{app:pf-yule-walker}}
\begin{proof}[Proof of Proposition \ref{prop:yule-walker-identification}]
	Define 
	\begin{equation*}
		\tilde{\mathcal{E}}_{p,q}(\Pi(L)):= \left\{ (\Phi, \Theta):\rho_{zy} = \Sigma_z \beta \right\}.
	\end{equation*}
	We first show that $\mathcal{E}_{p,q}(\Pi(L)) \subseteq \tilde{\mathcal{E}}_{p,q}(\Pi(L))$. To this end, note that any $(\Phi, \Theta) \in \mathcal{E}_{p,q}(\Pi(L))$ satisfies $y_t = \beta^\top z_t + a_t$. Therefore, $\mathbb{E} \left[y_t z_t^\top \right] = \beta^\top  \mathbb{E} \left[z_t z_t^\top \right]+ \mathbb{E} \left[ a_t z_t^\top \right]$. Since $\mathbb{E} \left[ a_t z_t^\top \right] = 0$, this implies $\rho_{zy}^\top = \beta^\top \Sigma_z$.
	
	Next we show that $\tilde{\mathcal{E}}_{p,q}(\Pi(L)) \subseteq \mathcal{E}_{p,q}(\Pi(L))$. To this end, note that the set $\mathcal{E}_{p,q}(\Pi(L))$ can be characterized precisely as the set of matrix AR and MA parameters $\Phi$ and $\Theta$ which satisfy almost surely (a.s.)
	\begin{equation}\label{eqn:iden2}
		y_t = \beta^\top z_t + a_t, ~~~ t \in \mathbb{Z}, 
	\end{equation}
	for a process $y_t = \Pi^{-1}(L) a_t$, where $a_t \stackrel{i.i.d.}{\sim} (0, \Sigma_a)$ is a white noise process. 
	
	Now, consider a solution of the Yule-Walker type linear systems of equation $\beta  \in \tilde{\mathcal{E}}_{p,q}(\Pi(L))$. Since ${\mathcal{E}}_{p,q}(\Pi(L)) \subseteq \tilde{\mathcal{E}}_{p,q}(\Pi(L))$ and ${\mathcal{E}}_{p,q}(\Pi(L)) \neq \phi$, this solution takes the form $\beta = \beta^* + \delta$, where 
	$\beta^* = \left[\Phi_1^*: \ldots: \Phi_p^*: \Theta_1^*: \ldots: \Theta_q^* \right]^\top \in \mathcal{E}_{p,q}(\Pi(L))$ is a particular solution of the linear systems, and 
	$\delta = \left[\delta_{11}: \ldots: \delta_{1p}: \delta_{21}: \ldots: \delta_{2q} \right]^\top$ satisfies  $\Sigma_z \delta = \mathbf{0}_{d(p+q) \times d}$.
	
	This implies $\delta^\top \Sigma_z \delta = \mathbf{0}_{d \times d}$, i.e. $var(\delta^\top z_t) = \mathbf{0}_{d \times d}$. In other words, $\delta^\top z_t$ is almost surely a constant. Since $\mathbb{E}[z_t] = 0$, we conclude that $\delta^\top z_t = 0$ a.s. 
	
	Now, consider any centered linear process $y_t = \Pi^{-1}(L) a_t$, as mentioned above. Then, since $\beta^* \in \mathcal{E}_{p,q}(\Pi(L))$, for any $t \in \mathbb{Z}$ we have 
	\begin{eqnarray*}
		y_t &=& \Phi_1^*y_{t-1}+ \ldots + \Phi_p^* y_{t-p} + \Theta_1^* a_{t-1} + \ldots + \Theta_q^* a_{t-q}+ a_t.  
	\end{eqnarray*}
	Also, since, $\delta^\top z_t = 0 \mbox{ a.s.}$, we have 
	\begin{eqnarray*}	
		y_t &=& (\Phi_1^*+\delta_{11})y_{t-1} + (\Phi_2^*+\delta_{12}) y_{t-2} + \ldots + (\Phi_p^*+\delta_{1p}) y_{t-p} \\
		&~& + (\Theta_1^*+\delta_{21}) a_{t-1} + \ldots + (\Theta_q^*+\delta_{2q}) a_{t-q}+ a_t \mbox{ a.s.}
	\end{eqnarray*}
	
	It follows from \eqref{eqn:iden2} that $\beta \in \mathcal{E}_{p,q}(\Pi(L))$, proving $\tilde{\mathcal{E}}_{p,q}(\Pi(L)) \subseteq \mathcal{E}_{p,q}(\Pi(L))$.
	
\end{proof}

\subsection{Optimization-based Identification \label{App.identif}}
Consider the convex minimization problem
$$\C=\arg\min_{x\in\L}f(x)$$
where $f:\real^n\to[0,\infty)$ is a
convex function and $\L\subseteq\real^n$ is an affine space.
We assume that $\C$
is non-empty (i.e. the minimum is attained) and let
$$
x^*=\arg\min_{x\in\L}\|x\|^2\st x\in\C,
$$
which is unique since this is a strongly convex problem.

Defining $f(x,\a)=f(x)+\frac{\a}{2}\|x\|^2$, we see that $f(\cdot,\a)$
is $\a$-strongly convex for each $\a>0$ and
therefore there is a unique minimizer
$$
\xa:=\arg\min_{x\in\L}f(x,\a).
$$
\begin{proposition} \label{th_sparseid}
	The sequence of minimizers of $f(\cdot,\a)$ converge, as $\a\to0^+$,
	to the unique minimizer of $f(\cdot)$ that has smallest $\ell_2$-norm:  $\lim_{\a\to0^+}\xa= x^*$.
\end{proposition}
\begin{proof}
	We begin with a lemma.
	\begin{lem}
		\begin{align*}
			\lim_{\a\to0^+}\frac{f(x_{\a},\a)-f(x^*,\a)}{\a}=0.
		\end{align*}
		\label{lem:inf}
	\end{lem}
	\begin{proof}
		By definition of $x^*$,
		$$
		\|x^*\|^2=\min_{x\in\L}\|x\|^2\st f(x)\le f^*,
		$$
		where $f^*=\min_{x\in\L}f(x)$.  This can be equivalently expressed
		(see, e.g., \cite{Boyd}) as
		$$
		\|x^*\|^2=\min_{x\in\L}\sup_{\lambda\ge0} L(x;\lambda)
		$$
		where
		$$
		L(x;\lambda)=\|x\|^2+\lambda(f(x)-f^*)=\lambda [f(x,2/\lambda)-f^*].
		$$
		By strong duality (Slater's condition holds since $\C\neq\emptyset$),
		we can interchange the ``$\min$'' and the ``$\sup$'':
		$$
		\|x^*\|^2=\sup_{\lambda\ge0} g(\lambda),
		$$
		where $ g(\lambda)=\min_{x\in\L} L(x;\lambda).  $ Now, for
		$\bar\lambda>\lambda\ge0$,
		\begin{align*}
			g(\bar\lambda)&=\min_{x\in\L}L(x,\bar\lambda)\\
			&=\min_{x\in\L}\left\{L(x,\lambda)+(\bar\lambda-\lambda)[f(x)-f^*]\right\}\\
			&\ge \min_{x\in\L}L(x,\lambda)+(\bar\lambda-\lambda)\min_{x\in\L}[f(x)-f^*]\\
			&\ge g(\lambda).
		\end{align*}
		Thus, $g$ is a non-decreasing function, and
		$$
		\lim_{\lambda\to\infty}g(\lambda)=\sup_{\lambda\ge0}
		g(\lambda)=\|x^*\|^2.
		$$
		Now,
		$$
		\|x^*\|^2=\lim_{\lambda\to\infty}g(\lambda)=\lim_{\lambda\to\infty}\left\{\lambda
		[f(x_{2/\lambda},2/\lambda)-f^*]\right\}=\lim_{\a\to0^+} (2/\a)
		[f(x_{\a},\a)-f^*]
		$$
		or, subtracting $\|x^*\|^2$ from both sides,
		$$
		0=\lim_{\a\to0^+} (2/\a) [f(x_{\a},\a)-f(x^*,\a)].
		$$
	\end{proof}

	By $\a$-strong convexity of $f(\cdot,\alpha)$,
	\begin{align}
		f(y,\a)\ge f(\xa,\a)+\frac{\a}{2}\|\xa-y\|^2\label{eq:scvx}
	\end{align}
	for any $y\in \L$.
	
	Applying this with $y=x^*$ gives
	$$
	f(x^*,\a)\ge f(\xa,\a)+\frac{\a}{2}\|\xa-x^*\|^2
	$$
	or
	$$
	\|\xa-x^*\|^2 \le (2/\a)[f(x^*,\a)- f(\xa,\a)].
	$$
	Taking the limit of both sides, Lemma \ref{lem:inf} gives
	$$  \lim_{\a\to0}\|\xa-x^*\|^2 \le0. $$
	Thus,
	$$
	\lim_{\a\to0}\xa=x^*.
	$$
\end{proof}
The above results are now easily applied to prove Proposition \ref{thm:identification}.
\begin{proof}[Proof of Proposition \ref{thm:identification}]
	Denote $x=(  \Phi,   \Theta)$. Consider the convex function
	$f(x) = \mathcal{P}_{\text{AR}}({  \Phi}) + \mathcal{P}_{\text{MA}}({  \Theta})$, and the affine space $\L$ in which  ${  \Phi}(L) = {  \Theta}(L)   \Pi(L)$ holds. 
	It follows from Proposition \ref{th_sparseid} that 
	$\lim_{\a\to0^+}({  \Phi}^{(\alpha)}, {  \Theta}^{(\alpha)})= ({  \Phi}^{(0)}, {  \Theta}^{(0)})$.
\end{proof}

\subsection{Identification for Multiple, Sparsest VARMA Representations}

\subsubsection{A Toy Example} \label{Tsay_example_Appendix}

In Section \ref{sparse_identif}, we refer to multiple equivalent, sparsest VARMA representations, as, for instance, discussed in  
Section 4.5.2 of  
\cite{Tsay14}. 
As an example, we consider  the VAR(1) and VMA(1) models 
$$
y_t = \begin{pmatrix} 0 & 1  \\
	0 & 0 \end{pmatrix} y_{t-1} + a_t \Leftrightarrow y_t = \begin{pmatrix} 0 & 1  \\
	0 & 0 \end{pmatrix} a_{t-1} + a_t.
$$
In this section, we establish our unique identification target for this example. 

Following the VARMA($p$,$q$) notation of our paper we write 
$
{\Phi}(L) {  y}_t =  {\Theta}(L) {  a}_t, \nonumber
$
where the AR and MA operators are respectively given by
\begin{equation}
	{\Phi}(L)  = {  I} -  \Phi_{1}L -  \Phi_{2}L^2 - \ldots -   \Phi_{p}L^p
	\ \ \text{and} \ \     {\Theta}(L)  = {  I} +  \Theta_{1}L +  \Theta_{2}L^2 + \ldots +   \Theta_{q}L^q, \nonumber
\end{equation}
with the lag operator $L^\ell$ defined as $L^\ell {  y}_t = {  y}_{t-\ell}$.
For the 
VMA(1) example 
with MA-coefficient matrix equal to 
\begin{equation} A = 
	\begin{pmatrix}
		0 & 1 \\
		0 & 0 \\
	\end{pmatrix} \nonumber
\end{equation}
we equivalently have ${\Phi}(L) = (I - 0L) = I$ and ${\Theta}(L) = (I + AL)$ in the VARMA(1,1) formulation. 
Further, since $\text{det}\{\Phi(z)\} \neq 0 $  and $\text{det}\{\Theta(z)\} \neq 0 $ for all $|z| \leq 1$ $(z \in \mathbb{C})$, this model is stable and invertible,  
and the process $ \{{  y}_t\}$ then has
an infinite-order VAR representation
$
\Pi(L) {  y}_t = {  a}_t, $
where 
$
\Pi(L) =  {\Theta}^{-1}(L)  {\Phi}(L) =  {  I} -  \Pi_{1}L -  \Pi_{2}L^2 - \cdots,  $
which in this case simplifies to 
$\Pi(L) = (I + AL)^{-1} I = I - AL$.
We then recognize this model is equivalent to a VAR(1) model with AR-coefficient matrix $A$, or we equivalently have $\widetilde{\Phi}(L) = (I - AL)$ and $\widetilde{\Theta}(L) = (I + 0L) = I$, in its VARMA(1,1) formulation.

In our paper, we therefore note that both models, as defined by their AR and MA coefficient matrix pairs $({\Phi},  {\Theta})$: $({I},  {A})$ and $({A},  {I})$, respectively, are in the same VARMA(1,1) equivalence class $\mathcal{E}_{1,1}$ with respect to $\Pi(L) = (I - AL)$. 
This is defined for the general VARMA(p,q) model as
$
\mathcal{E}_{p,q}( {\Pi}(L)) = \{ ( {\Phi},  {\Theta})  : 
{\Phi}(L) =  {\Theta}(L) {\Pi}(L)  \}, 
$    
and in this case we specifically have
\begin{equation}%\label{eqn:eqv-class1}
	\mathcal{E}_{1,1}( I - AL ) = \{ ( \overline{\Phi},  \overline{\Theta})  : 
	(I - \overline{\Phi}L) =  (I + \overline{\Theta} L) (I - AL)  \}. \nonumber
\end{equation}
For the equivalence relation of $\mathcal{E}_{1,1}$ to hold for the given $A$, 
any matrix pair  $(\overline{\Phi},  \overline{\Theta})$ in the set must also be of the form   
\begin{equation*} \overline{\Phi} = \overline{\Phi}(a) = 
	\begin{pmatrix}
		0 & a \\
		0 & 0 \\
	\end{pmatrix} 
	\quad and \quad
	\overline{\Theta} = \overline{\Theta}(b) =  
	\begin{pmatrix}
		0 & b \\
		0 & 0 \\
	\end{pmatrix}, 
	%\nonumber
\end{equation*}
for $a,b \in \mathbb{R}$ such that, $a+b =1$, and so there are many solutions, not just the two identified above. 

We now turn to considering this problem from the proposed optimization-based identification perspective (Section \ref{sec:id-new}) by using strongly convex optimization to establish identification. 
Among all feasible AR and MA matrix pairs $(\overline{\Phi},  \overline{\Theta})$, we look for the one that gives the most parsimonious representation of the VARMA. Specifically, we measure parsimony through a pair of  convex regularizers, $\mathcal{P}_{\text{AR}}(\overline\Phi)$ and $\mathcal{P}_{\text{MA}}(\overline\Theta)$.
Our identification results apply equally well to any convex function;  one can consider,  amongst others, the $\ell_1$-norm, the $\ell_2$-norm, the nuclear norm, and convex combinations thereof. 

To be concrete for this particular example, let us specifically consider using the $\ell_1$-norm:
$ \mathcal{P}_{\text{AR}}(\overline\Phi) = \|\overline\Phi\|_1 \ \text{and }  \ \mathcal{P}_{\text{MA}}(\overline\Theta) =  \|\overline\Theta\|_1.$
Then for any fixed $\alpha>0$, a uniquely identified solution $({  \overline\Phi}^{(\alpha)}, {  \overline\Theta}^{(\alpha)})$ is
\begin{equation*}
	\underset{{  \overline\Phi}, {  \overline\Theta}}{\operatorname{argmin}} \ \left\{\|\overline\Phi\|_1 + \|\overline\Theta\|_1 + \dfrac{\alpha}{2} \|{ \overline\Phi}\|_F^2 +  \dfrac{\alpha}{2}  \|{  \overline\Theta}\|_F^2  \; \st  \ {  (I - \overline{\Phi}L) =  (I + \overline{\Theta} L) (I - AL) } 
	\right\}  
\end{equation*}
(see general case, %page 10, 
Equation \eqref{Optim_pair}), and this is equivalent to 
\begin{equation*}
	\underset{{  \overline\Phi(a)}, {  \overline\Theta(b)}}{\operatorname{argmin}} \ \left\{\left\|\begin{pmatrix}
		0 & a \\
		0 & 0 \\
	\end{pmatrix}\right\|_1 + \left\|\begin{pmatrix}
		0 & b \\
		0 & 0 \\
	\end{pmatrix} \right\|_1 + \dfrac{\alpha}{2} \left\|{ \begin{pmatrix}
			0 & a \\
			0 & 0 \\
	\end{pmatrix}}\right\|_F^2 +  \dfrac{\alpha}{2}  \left\|{ \begin{pmatrix}
			0 & b \\
			0 & 0 \\
	\end{pmatrix} }\right\|_F^2 
	\st a+b =1
	\right\},  
\end{equation*}
or more simply, 
\begin{equation*}
	\underset{{  \overline\Phi(a)}, {  \overline\Theta(b)}}{\operatorname{argmin}} \ \left\{\
	|a| + |b| + \dfrac{\alpha}{2}|a|^2 + \dfrac{\alpha}{2}|b|^2     \st a+b =1
	\right\}, 
\end{equation*}
This optimization problem is strongly convex and thus has a unique solution pair $({  \overline\Phi}^{(\alpha)}, { \overline \Theta}^{(\alpha)})$ for each value of $\alpha>0$.
We further define our final \color{black} (unique) optimization-based  \color{black} identified VARMA representation as 
\begin{equation*}
	({  \overline\Phi}^{(0)}, {  \overline\Theta}^{(0)}) = \underset{\alpha \rightarrow 0^+}{\text{lim}} ({  \overline\Phi}^{(\alpha)}, {  \overline\Theta}^{(\alpha)}),
\end{equation*}
a result which is proved (in the general case) in Proposition \ref{thm:identification} 
to be the \textit{unique} pair of autoregressive and moving average matrices in the `regularized equivalent' class having smallest Frobenius norm, i.e., the regularized equivalent (sub-) class of $\mathcal{E}_{1,1}(I - AL)$ in this example is defined as $\mathcal{RE}_{1,1}( I - AL ) = $
\begin{equation*} 
	\underset{{  \overline\Phi}, {  \overline\Theta}}{\operatorname{argmin}} \ \left\{\|\overline\Phi\|_1 + \|\overline\Theta\|_1 \; \st  \ {  (I - \overline{\Phi}L) =  (I + \overline{\Theta} L) (I - AL)} 
	\right\}
	=
	\underset{{  \overline\Phi(a)}, {  \overline\Theta(b)}}{\operatorname{argmin}} \ \left\{\
	|a| + |b|   \st a+b = 1
	\right\}
\end{equation*}
(which has many solutions, i.e., $b = 1 - a, a \in [0,1]$). 
Then our final unique solution for this specific problem is 
\begin{align*}
	({  \overline\Phi}^{(0)}, {  \overline\Theta}^{(0)}) &= 
	\underset{{  \overline\Phi(a)}, {  \overline\Theta(b)}}{\operatorname{argmin}} \ \left\{\
	|a|^2 + |b|^2     \st ({  \overline\Phi(a)}, {  \overline\Theta(b)}) \in \mathcal{RE}_{1,1}( I - AL )
	\right\} \\
	&= 
	\underset{{  \overline\Phi(a)}, {  \overline\Theta(b)}}{\operatorname{argmin}} \ \left\{\
	|a|^2 + |b|^2    \st b = 1 - a, a \in [0,1] 
	\right\}. 
\end{align*}
This has the unique solution $a = b = 0.5$, or 
\begin{equation} \overline\Phi = \overline\Theta = 
	\begin{pmatrix}
		0 & 0.5 \\
		0 & 0 \\
	\end{pmatrix}, \nonumber
\end{equation}
and we can further confirm by hand this solution is in $\mathcal{E}_{1,1}( I - AL )$, since 
$$
\left(I +    
\begin{pmatrix}
	0 & 0.5 \\
	0 & 0 \\
\end{pmatrix}L\right)^{-1}\left(I -     
\begin{pmatrix}
	0 & 0.5 \\
	0 & 0 \\
\end{pmatrix}L \right) \\
=
\left(I -    
\begin{pmatrix}
	0 & 0.5 \\
	0 & 0 \\
\end{pmatrix}L\right)^2 \\
= 
\left(I -    
\begin{pmatrix}
	0 & 1 \\
	0 & 0 \\
\end{pmatrix}L\right) \\
= (I - AL).$$    
Finally, we note that although the proposed unique VARMA(1,1) solution above does not have as few non-zero parameters as either the pure VMA(1) or VAR(1) model 
(in which there was just one), in finding this solution via optimization with constraints there was still only one free parameter, and therefore the same overall model complexity in this regard. 
Furthermore, this is only the unique solution derived under the $\ell_1$-norm choice of regularization, and we reiterate that the flexible framework that we propose also allows any (user specified) convex function for regularization-based identification, including the $\ell_1$-norm, the $\ell_2$-norm, the nuclear norm, and convex combinations thereof. 

\subsubsection{Simulation} \label{sim_study_identification}
We further illustrate sparse identification with a small simulation study. Figure \ref{exampleID} (panel a) shows a $\text{VARMA}_{d=8}(1, 1)$ model 
$$\Phi_{\text{dense}} = \protect\begin{bmatrix} {\bf 0.2} & {\bf 0.05} \protect\\ {\bf 0}  & {\bf 0.1}\protect\end{bmatrix}, \ \text{and} \
\Theta_{\text{dense}} =  \protect\begin{bmatrix} {\bf 0} & {\bf -0.25} \protect\\ {\bf 0}  & {\bf -0.1}\protect\end{bmatrix},
$$
with the dense $(\Phi,\Theta)$  having 80 nonzero entries. However, this VARMA model can be alternatively expressed in terms of 
$$\Phi_{\text{sparse}} = \protect\begin{bmatrix} {\bf 0.2} & {\bf 0} \protect\\ {\bf 0}  & {\bf 0}\protect\end{bmatrix}, \ \text{and} \ \Theta_{\text{sparse}} =  \protect\begin{bmatrix} {\bf 0} & {\bf -0.2} \protect\\ {\bf 0}  & {\bf 0}\protect\end{bmatrix},$$
a sparse $(\Phi,\Theta)$ having only $32$ nonzero entries (panel b);
or 
$$\Phi^{(0)} = \protect\begin{bmatrix} {\bf 0.1} & {\bf -0.1} \protect\\ {\bf 0}  & {\bf 0}\protect\end{bmatrix}, \ \text{and} \  		\Theta^{(0)} =  \protect\begin{bmatrix} {\bf 0.1} & {\bf -0.1} \protect\\ {\bf 0}  & {\bf 0}\protect\end{bmatrix},$$
having only 64 nonzero entries (panel c).

Note that there are multiple equivalent, 
minimum-$\ell_1$ VARMA representations.
Two of these  are  visualized in panels (b) and (c) but others exist such as the pair where the AR and MA matrices of the ``sparse" design are swapped. 
All  have minimal $\ell_1$-norm (i.e.\ $||\Phi||_1 = ||\Theta||_1 = 3.2$). 
Panel (c) displays the unique pair $(\Phi^{(0)}, \Theta^{(0)})$, defined in Equation \eqref{uniquePhiTheta}, as the one having minimal $\ell_2$-norm (i.e. $||\Phi^{(0)}||^2_F = ||\Theta^{(0)}||^2_F = 0.32$). 
When choosing the $\ell_1$-norm as the convex regularizer, our optimization-based identification strategy would favor the sparser VARMA representations over the denser one since the former have a smaller $\ell_1$-norm (i.e. $\ell_1$-norm for the dense design is $||\Phi||_1 = ||\Theta||_1 = 5.6$). 

\begin{figure}
	\centering
	\includegraphics[width=\textwidth]{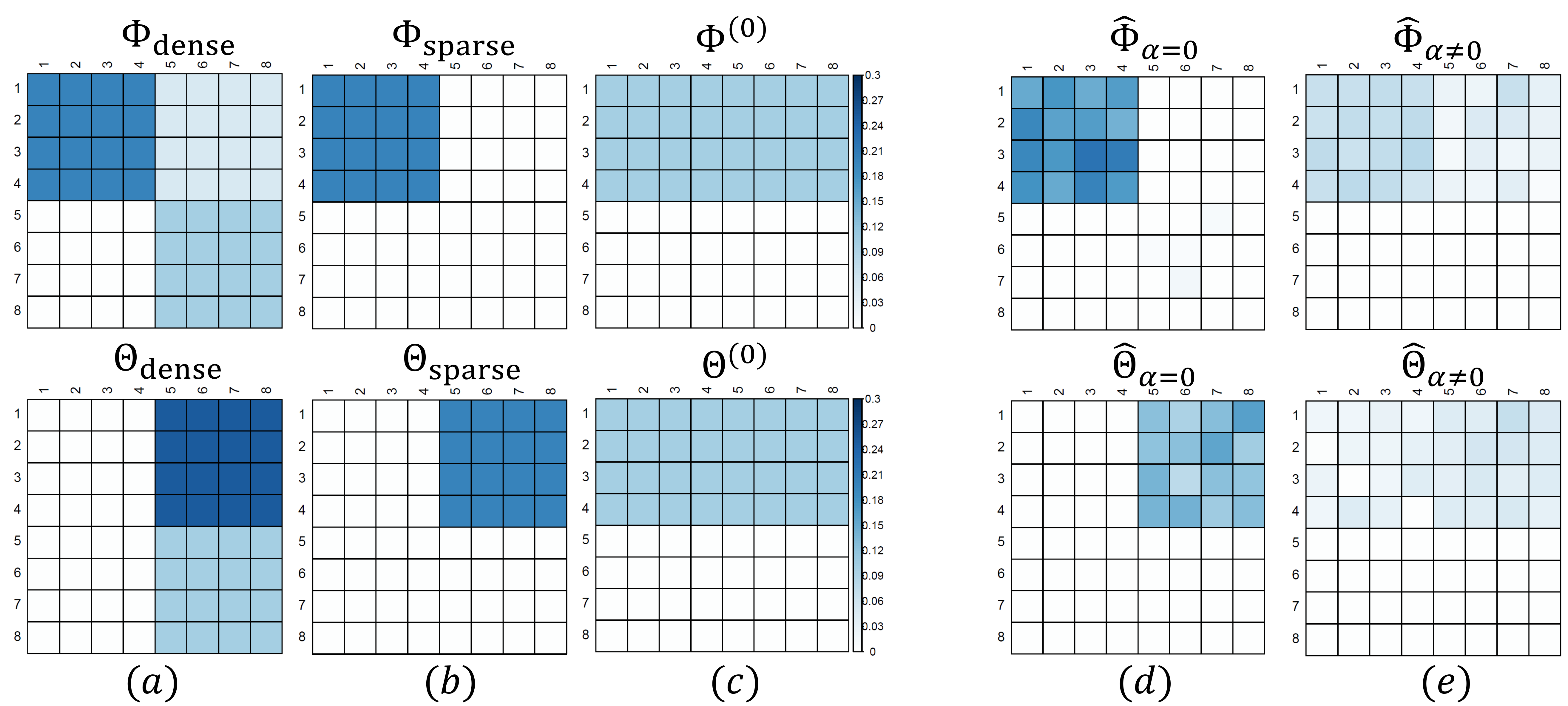}
	\caption{True AR and MA matrices  from three equivalent VARMA representations: (a) a dense, (b) a sparse and (c) the target VARMA.  
		The estimates obtained with our sparse VARMA estimation procedure for $\alpha=0$ are displayed in panel (d); for a small $\alpha \neq 0$ in panel (e). Darker shading of cells indicate parameters that are larger in (absolute) magnitude.
	}\label{exampleID}
\end{figure}

To illustrate the link between our identification and estimation stages, consider the following simulation experiment: 
We take $\Sigma_a ={I_d}$ and generate time series of length $T=1000$ (after 200 burn-in observations) from the \textit{dense} VARMA (Figure \ref{exampleID}, panel a). We then use our sparse VARMA procedure  with $\ell_1$-norm as convex regularizer and take $p=q=1$ to obtain the AR and MA parameter estimates. The number of simulations is $N=500$.

First, we estimate the VARMA with $\alpha=0$; the corresponding estimates are visualized in Figure \ref{exampleID} panel (d), for an illustrative simulation run. 
The results are very stable from one simulation run to another. 
Although we generate the time series from the dense DGP, our  procedure encourages  identification and estimation of sparser models and thus returns sparser estimates.
Since there are infinitely many equivalent ``true" $(\Phi,\Theta)$ pairs, we are not interested in comparing the estimates to the  dense $(\Phi,\Theta)$ pair used to originally generate the data.
Instead, we aim to produce estimates that are close to the sparse equivalence class.
In almost all simulation runs (496 out of 500), Matthews Correlation Coefficient (MCC) between the  sparse  $(\Phi,\Theta)$ (Figure  \ref{exampleID}, panel b) and the estimates equals one; thereby providing perfect recovery.
By taking $\alpha=0$, our simulation results thus show that our estimates are very close to 
one of the elements in the sparse equivalance class,
which is in line with our theoretical result on partially identified estimation. 

Next, we estimate the VARMA with a non-zero but small  $\alpha$ (we take $\alpha=10^{-2}$, thus small relative to the selected values for $\lambda_\Phi \approx 10^2$ and $\lambda_\Theta \approx 10$). 
By adding the $\ell_2$-norm to the objective function, we expect to produce estimates that are closer to the unique $(\Phi^{(0)},\Theta^{(0)})$. 
This expectation is confirmed by our results, as can be seen from the corresponding estimates, visualized in Figure \ref{exampleID} panel (e). 
The average (over the simulation runs) MCC  between the target (Figure  \ref{exampleID}, panel c)  and the estimates is 0.97 with a standard error smaller than 0.001.

Since there exist multiple equivalent, sparsest VARMA representations with different support, we do not focus on model selection consistency in the paper but instead on forecasting. For forecasting purposes, we are interested in obtaining a parsimonious VARMA representation with good out-of-sample performance. For this reason, we prefer to use $\alpha=0$ in the simulation study and forecast applications since our numerical experiments showed that this generally produces a sparser (i.e.\ with fewer non-zero coefficients) estimated VARMA compared to the estimates obtained when taking $\alpha$ non-zero but small.

\section{Key Technical Ingredients}\label{sec:technical-lemmas}

Our first technical ingredient provides a deviation bound (in element-wise maximum norm) for the product of two random matrices, whose rows consist of consecutive observations from two time series that are outputs of a linear filter applied on the same stationary Gaussian time series. In the analysis of both Phase-I and Phase-II, we use this result to control upper bounds on inner products of columns of the design matrix and the error matrix. This proposition generalizes a similar concentration bound in \citep{Basu15} for uncorrelated time series.   

\begin{proposition}\label{lem:yAy}
	Let $\{y_t\}_{t \in \mathbb{Z}}$ be a $d$-dimensional  stable, Gaussian, centered time series with spectral density $f_y$. Consider two time series  $X_t=\mathcal{A}(L)y_t$ and $Y_t = \mathcal{B}(L)y_t$, whose $d \times d$ matrix-valued lag polynomials $\mathcal{A}(L)$ and $\mathcal{B}(L)$ satisfy $\|\mathcal{A}\|_{2,1} < \infty$, $\|\mathcal{B}\|_{2,1} < \infty$. Let $\mathcal{X} = \left[X_T: X_{T-1}: \cdots: X_1\right]^\top$ and $\mathcal{Y} = \left[Y_T: Y_{T-1}:\cdots:Y_1 \right]^\top$ be two data matrices, each containing in its rows $T$ consecutive observations from the time series $\{X_t\}$ and $\{Y_t\}$, respectively. Then there exists a universal constant $c > 0$ such that for any $\eta > 0$ and any $u, v \in \mathbb{S}^{d-1}$, we have  
	\begin{equation}\label{eqn:devn-xy}
		\mathbb{P} \left[ \left| u^\top \left( \mathcal{X}^\top \mathcal{Y}/T - \Gamma_{X,Y}(0) \right)v \right| > 6\pi  \vertiii{f_y} \max \! \left\{ \vertiii{\mathcal{A}}^2, \vertiii{\mathcal{B}}^2 \right\} \eta \right]
	\end{equation}
	is at most $6 \exp[-c T \min \{ \eta, \eta^2 \}]$. 
	
	\noindent In addition, if $T \succsim \log d$, then for any $A >  0$, the following upper bound holds with probability at least $1 - 6 \exp \left[-2(c A^2-1) \log d \right]$: 
	\begin{eqnarray*}
		\left\|\mathcal{X}^\top \mathcal{Y}/T \right\|_\infty \le 
		2\pi  \vertiii{f_y} \left[3A \,\max \! \left\{ \vertiii{\mathcal{A}}^2, \vertiii{\mathcal{B}}^2 \right\} \sqrt{ 2 \log d / T} + \vertiii{\mathcal{A}} \left\| \mathcal{B} \right\|_{2,1} \right].
	\end{eqnarray*}
\end{proposition}

\begin{remark} The two terms in the above bound can be viewed as the variance and bias terms. The first term provides a bound on the deviation of $\mathcal{X}^\top \mathcal{Y}/T$ around its expectation in element-wise maximum norm. This bound scales with the dimension $d$ at a rate $\sqrt{\log d / T}$ similar to the case of i.i.d. random variables. In addition, the terms $\vertiii{f_y}$, $\vertiii{\mathcal{A}}$ and $\vertiii{\mathcal{B}}$ capture the effect of temporal dependence on the convergence rates. The second term provides a bound on the bias, i.e. the population covariance between the time series $X_t$ and $Y_t$. This H{\"o}lder-type bound involves the operator norms of the spectral density of $y_t$ (across frequencies), and the linear filters $\mathcal{A}(L)$ and $\mathcal{B}(L)$ applied on $y_t$. The bound on bias can be potentially improved using additional structures of the linear filters (see remark after proof below).
\end{remark}

\begin{proof}[Proof of Proposition \ref{lem:yAy}]
	
	In order to obtain a high probability concentration bound, we first state a generalized version of Proposition 2.4(b) in \citep{Basu15}, allowing for correlation between the two time series.  	The proof follows along the same line, only replacing $(2/n) \sum_{t=1}^n w^t z^t$ with $(2/n) \sum_{t=1}^n w^t z^t - Cov(z^t, w^t)$ in the left hand side of the first equation in their proof.
	
	Let $\{X_t\}_{t \in \mathbb{Z}}$ and $\{Y_t\}_{t \in \mathbb{Z}}$ be two $d$-dimensional stationary Gaussian centered time series, with autocovariance function $\Gamma_{X,Y}(h) = cov(X_t, Y_{t+h}) = \mathbb{E} [X_t Y_{t+h}^\top]$ and cross-spectral density $f_{X,Y}$. Assume the process $Z_t = [X_t^\top :Y_t^\top ]^\top $ is stable so that it has bounded cross-spectrum $\vertiii{f_{X,Y}} < \infty$. Let $\mathcal{X}$ and $\mathcal{Y}$ be $T \times d$ data matrices, with rows corresponding to consecutive observations from the time series $\{X_t\}$ and $\{Y_t\}$, respectively. Then, for any $u,\, v \in \mathbb{R}^d$ with $\|u\| \le 1$, $\|v\| \le 1$, and any $\eta > 0$, we have 
	\begin{equation}\label{eqn:dev-xycor}
		\mathbb{P}\left[\left|u^\top (\mathcal{X}^\top \mathcal{Y}/T - \Gamma_{X,Y}(0))v \right| > 2\pi \left[\vertiii{f_{X,Y}}+\vertiii{f_X} + \vertiii{f_Y} \right] \eta \right] 
	\end{equation}
	is at most $ 6\exp[-c \, T \min \{\eta, \eta^2 \}]$ for some universal constant $c > 0$. 
	
	Next, we 
	use the fact that $\vertiii{f_{X,Y}}^2$ is at most $\vertiii{f_X} \vertiii{f_Y}$, so that $\vertiii{f_{X,Y}}+\vertiii{f_X} + \vertiii{f_Y}$ is at most $3 \max\{\vertiii{f_X}, \vertiii{f_Y} \}$.

	\noindent By definition of $X_t$ and $Y_t$, the spectral densities take the form 
	\begin{eqnarray*}
		f_X(\theta) &=& \mathcal{A}(e^{i\theta}) f_y(\theta) \mathcal{A}^*(e^{i\theta}), \\
		f_Y(\theta) &=& \mathcal{B}(e^{i\theta}) f_y(\theta) \mathcal{B}^*(e^{i\theta}), \\
		f_{X,Y} (\theta) &=& \mathcal{A}(e^{i\theta}) f_y(\theta) \mathcal{B}^*(e^{i\theta}).
	\end{eqnarray*}
	
	\noindent This implies
	$\vertiii{f_X} \le \vertiii{\mathcal{A}}^2 \vertiii{f_y}$, $\vertiii{f_Y} \le \vertiii{\mathcal{B}}^2 \vertiii{f_y}$ and $\vertiii{f_{X,Y}} \le \vertiii{\mathcal{A}} \vertiii{\mathcal{B}} \vertiii{f_y} < \infty$, so that the above concentration bound can be applied. 
	Plugging in these upper bounds into the above concentration inequality, we prove the first part of our proposition. 
	
	In order to prove the second part, we set $\eta = A \sqrt{(\log d^2) / T}$ and take union bound of the event in \eqref{eqn:dev-xycor} over $d^2$ choices of $u, v \in \{e_1, \ldots, e_d \}$, the set of canonical unit vectors in $\mathbb{R}^d$. Since $T \succsim \log d$, we have $\min\{\eta, \eta^2 \} = \eta^2$ so that the above inequality implies
	\begin{eqnarray*}
		\mathbb{P} \left[ \left\| \mathcal{X}^\top \mathcal{Y}/T \right\|_\infty \, > \,  \left\| \Gamma_{X,Y}(0) \right\|_\infty + 6\pi A \vertiii{f_y} \max \! \left\{ \vertiii{\mathcal{A}}^2, \vertiii{\mathcal{B}}^2 \right\} \sqrt{ 2\log d / T} \right]
	\end{eqnarray*}
	is at most $6 d^2 \exp[ -c A^2 \log d^2 ] = 6 \exp \left[-(cA^2-1) \log d^2 \right]$. 
	
	Next, in order to get an upper bound on $\|\Gamma_{X,Y}(0)\|_\infty$, note that 
	\begin{eqnarray*}
		\Gamma_{X,Y}(0) &=& \cov(\mathcal{A}(L)y_t, \mathcal{B}(L)y_t) \\
		&=& \sum_{\ell \ge 0} \sum_{m \ge 0} A_\ell \Gamma(\ell-m) B_m^\top \\
		&=& \int_{-\pi}^\pi \sum_{\ell \ge 0} \sum_{m \ge 0} e^{i(\ell-m)\theta} A_\ell f(\theta) B_m^\top d\theta \\
		&=& \sum_{m \ge 0} \left[ \int_{-\pi}^\pi \left( \sum_{\ell \ge 0} A_\ell e^{i\ell\theta} \right) f(\theta) e^{-im \theta} d\theta \right] B_m^\top.
	\end{eqnarray*}
	
	Therefore,
	\begin{equation*}
		\| \Gamma_{X,Y}(0) \|_\infty \le \| \Gamma_{X,Y}(0) \| \le 2 \pi \vertiii{\mathcal{A}} \vertiii{f_y} \| \mathcal{B} \|_{2,1}. 
	\end{equation*}
\end{proof}

\begin{remark} Note that the bound on $\| \Gamma_{X,Y}(0) \|$ may be improved using information on the dependence between $X_t$ and $Y_t$. For instance, if we consider $X_t = y_{t-\ell}$ and $Y_t = y_t$, then we can expect that $\Gamma_{X,Y}(0)$, the covariance between $X_t$ and $Y_t$,  will decay with larger $\ell$, but our bound does not. A tighter bound on $\|\Gamma_{X,Y}(0)\|$ can potentially be obtained using special structures of $X_t$ and $Y_t$, as in our proof of Proposition \ref{prop:dev-bound}. 
\end{remark}

Our second key technical ingredient will be used to provide an upper bound on the operator norm of the spectral density of a time series of the form $z_t$ in Proposition \ref{prop:yule-walker-identification} in terms of the spectral density of $y_t$ and the linear filter used to generate $a_t$ from $y_t$. We use this to provide a finite-sample upper bound on the deviation of the sample Gram matrix in the Phase-II regression from its population analogue. 

\begin{proposition}\label{lem:f-y-eps}
	Consider a $d$-dimensional centered stable process $\{y_t\}$, and a $d \times d$ matrix-valued lag polynomial $\mathcal{C}(L)$ with finite $\| \mathcal{C} \|_{2,1}$. Then the spectral density of the $d(p+q)$-dimensional derived process 
	\begin{equation*}
		z_t = \left[ y_{t-1}^\top, y_{t-2}^\top, \ldots, y_{t-p}^\top, \mathcal{C}(L)y_{t-1}^\top, \ldots, \mathcal{C}(L)y_{t-q}^\top \right]^\top
	\end{equation*}
	satisfies $\vertiii{f_z} \le \left(p + q \vertiii{\mathcal{C}}^2\right) \,  \vertiii{f_y}$.
\end{proposition}

\begin{proof}[Proof of Proposition \ref{lem:f-y-eps}]
	Let $\mathcal{C}(L) = \sum_{\ell \ge 0} C_\ell L^\ell$ be a potentially infinite order $d \times d$ matrix-valued lag polynomial. The autocovariance function of the process $\{z_t\}$ takes the form 
	\begin{equation*}
		\Gamma_{z}(h) = \cov(z_t, z_{t+h}) = \cov \left( \left[\begin{array}{c} y_{t-1} \\ \vdots \\ y_{t-p} \\ \mathcal{C}(L)y_{t-1} \\ \vdots \\ \mathcal{C}(L)y_{t-q} \end{array} \right], \left[\begin{array}{c} y_{t-1+h} \\ \vdots \\ y_{t-p+h} \\ \mathcal{C}(L)y_{t-1+h} \\ \vdots \\ \mathcal{C}(L)y_{t-q+h} \end{array} \right] \right).
	\end{equation*}
	The $d(p+q) \times d(p+q)$ matrix on the right can be partitioned into four blocks. Since $\|\Gamma_y(h)\| \le 2 \pi \vertiii{f_y} < \infty$ for all $h \in \mathbb{Z}$, and $ \| \mathcal{C}\|_{2,1} = \sum_{\ell \ge 0} \|C_\ell \| < \infty$, using dominated convergence theorem we can express the four blocks as follows. 
	\begin{enumerate}
		\item \textit{Block (1,1), size $dp \times dp$}: consists of $p^2$ submatrices of size $d \times d$ each, the  $(r,s)^{th}$ submatrix given by $\cov \left(y_{t-r}, y_{t-s+h} \right) = \Gamma_y(r-s+h)$, for $1 \le r, s, \le p$;
		\item \textit{Block (1,2), size $dp \times dq$}: consists of $pq$ submatrices of size $d \times d$ each, the $(r,s)^{th}$ submatrix given by $\cov \left(y_{t-r}, \sum_{\ell \ge 0} C_\ell y_{t-s+h-\ell} \right) = \sum_{\ell \ge 0} \Gamma_y(r-s+h-\ell) C_\ell^\top$, for $1 \le r \le p$, $1 \le s \le q$;
		\item \textit{Block (2,1), size $dq \times dp$}: consists of $pq$ submatrices of size $d \times d$ each, the $(r,s)^{th}$ submatrix given by $\cov \left( \sum_{\ell \ge 0} C_\ell y_{t-r-\ell}, y_{t-s+h} \right) = \sum_{\ell \ge 0} C_\ell \Gamma_y(r-s+h+\ell)$, for $1 \le r \le q$, $1 \le s \le p$;
		\item \textit{Block (2,2), size $dq \times dq$}: consists of $q^2$ submatrices of size $d \times d$ each, the $(r,s)^{th}$ submatrix given by $\cov \left( \sum_{\ell \ge 0} C_{\ell} y_{t-r-\ell}, \sum_{\ell' \ge 0} C_{\ell'} y_{t-s+h-\ell'} \right) = \sum_{\ell , \ell' \ge 0} C_\ell \Gamma_y(h+r-s+\ell-\ell') C_{\ell'}$, for $1 \le r, s \le q$.  
	\end{enumerate}
	Similarly, the spectral density $f_z(\theta) = (1/2\pi) \sum_{h = -\infty}^\infty \Gamma_z(h) e^{-i h \theta}$, for any $\theta \in [-\pi, \pi]$ can be partitioned into four blocks as follows:\\
	\underline{\textit{Block (1,1)}}: the $(r,s)^{th}$ submatrix, for $1 \le r \le p$, $1 \le s \le q$, is given by \begin{equation*}
		\frac{1}{2\pi} \sum_{h = -\infty}^\infty \Gamma_y(h+r-s) e^{-ih\theta} = e^{i(r-s)\theta} f_y(\theta)
	\end{equation*}
	\underline{\textit{Block (1,2)}}: the $(r,s)^{th}$ submatrix, for $1 \le r \le p$, $1 \le s \le q$, is given by 
	\begin{eqnarray*}
		&~& \frac{1}{2\pi} \sum_{h = -\infty}^\infty \sum_{\ell \ge 0} \Gamma_y(r-s+h-\ell)C_\ell^\top e^{-ih\theta} \\
		&=& \sum_{\ell \ge 0} \left[ \frac{1}{2\pi} \sum_{h=-\infty}^\infty \Gamma_y(r-s+h - \ell) e^{-i(r-s+h-\ell)\theta} \right]C_\ell^\top e^{i(r-s-\ell)\theta } \\
		&=& f_y(\theta) \mathcal{C}^* (e^{i\theta}) e^{i(r-s)\theta}.
	\end{eqnarray*}
	\underline{\textit{Block (2,1)}}: the $(r,s)^{th}$ submatrix, for $1 \le r \le q$, $1 \le s \le p$ is given by 
	\begin{eqnarray*}
		&~& \frac{1}{2\pi} \sum_{\ell \ge 0} C_\ell \Gamma_y(r-s+h+\ell) e^{-ih\theta} \\
		&=& \sum_{\ell \ge 0} \left[ \frac{1}{2\pi} \sum_{h=-\infty}^{\infty} \Gamma_y(r-s+h+\ell) e^{-i(h+r-s+\ell) \theta } \right] e^{i(r-s+\ell)\theta } \\
		&=& e^{i(r-s)\theta } \left( \sum_{\ell \ge 0} C_\ell e^{i\ell \theta } \right) f_y(\theta) = e^{i(r-s)\theta} \mathcal{C}(e^{i\theta}) f_y(\theta).
	\end{eqnarray*}
	\underline{\textit{Block (2,2)}}: the $(r,s)^{th}$ submatrix, for $1 \le r \le q$, $1 \le s \le q$ is given by 
	\begin{eqnarray*}
		&~& \frac{1}{2\pi} \sum_{h=-\infty}^\infty \sum_{\ell, \ell' \ge 0} C_\ell \Gamma_y(h+r-s+\ell-\ell')C_{\ell'}^\top e^{-ih\theta} \\
		&=& \sum_{\ell, \ell' \ge 0} C_\ell \left(\frac{1}{2\pi} \sum_{h=-\infty}^\infty \Gamma_y(h+r-s+\ell-\ell') e^{-i(h+r-s+\ell-\ell')\theta} \right) 
		C_{\ell'}^\top e^{i(r-s+\ell-\ell')\theta } 
		\\
		&=& e^{i(r-s)\theta } \left(\sum_{\ell \ge 0} C_\ell e^{i\ell \theta} \right) f_y(\theta) \left(\sum_{ \ell' \ge 0} C^\top_{\ell'} e^{-i\ell'\theta} \right) 
		= e^{i(r-s)\theta} \mathcal{C}(e^{i\theta})f_y(\theta) \mathcal{C}^* (e^{i\theta}).
	\end{eqnarray*}
	
	Let $v_p$ and $v_q$ denote the vectors $[e^{i\theta}, \ldots, e^{ip\theta}]^\top$ and $[e^{i\theta}, \ldots, e^{iq\theta}]^\top$ respectively. Then the four blocks of $f_z(\theta)$ can be expressed as   $\left(v_p v_p^* \right) \otimes f_y(\theta)$, $\left(v_p v_q^* \right) \otimes \left(f_y(\theta) \mathcal{C}^*(e^{i\theta}) \right)$, $\left(v_q v_p^* \right) \otimes \left(\mathcal{C}(e^{i\theta}) f_y(\theta) \right)$ and $\left(v_q v_q^* \right) \otimes \left(\mathcal{C}(e^{i\theta}) f_y(\theta) \mathcal{C}^*(e^{i\theta})\right)$ respectively. Since $\|v_p \| = \sqrt{v_p^* v_p} = \sqrt{p}$, and $\|v_q \| = \sqrt{q}$, and $\| A \otimes B \| = \|A\| \|B\|$, by the norm compression inequality we obtain 
	\begin{eqnarray*}
		\vertiii{f_z} &\le& \left\| \left[ \begin{array}{cc}p \vertiii{f_y} & \sqrt{pq} \vertiii{f_y} \vertiii{\mathcal{C}} \\ \sqrt{pq} \vertiii{f_y} \vertiii{\mathcal{C}} & q \vertiii{f_y} \vertiii{\mathcal{C}}^2 \end{array} \right] \right\| \\
		&=& \vertiii{f_y} \left( p + q \vertiii{\mathcal{C}}^2 \right).
	\end{eqnarray*}
\end{proof}

\begin{lemma}[Controlling $\hat{\varepsilon}_t - \varepsilon_t$]\label{lem:ehat-e}
	Let $\{x_t\}$ be a $d$-dimensional, centered, stable Gaussian time series, and let $z_t = \mathcal{B}(L)(\hat{\varepsilon}_t - \varepsilon_t)$, where $\mathcal{B}(L)$ is a finite order lag polynomial of degree $q$ and $\{ \hat{\varepsilon}_t - \varepsilon_t \}$, $t = 1, , \ldots, n+q$ is a sequence of $d$-dimensional random vectors satisfying $\sum_{t=1}^{n+q} \| \hat{\varepsilon}_t - \varepsilon_t \|^2/(n+q) \le \Delta_\varepsilon^2$ on an event $\mathcal{E}$ such that $\mathbb{P}(\mathcal{E}) \ge 1 - c_0 \exp[-(c_1 A^2-1) \log d^2 \tilde{p}]$. Also, let $\{w_t\}_{t=1-j}^{n-j}$ be a sequence of random vectors given by $w_t = \hat{\varepsilon}_{t-j} - \varepsilon_{t-j}$, for some $j \in \{1, \ldots, q\}$. Consider data matrices $\mathcal{X}, \mathcal{Z}$ and $\mathcal{W}$ containing $n$ consecutive observations from the time series $x_t, z_t$ and $w_t$ respectively, and assume $n \succsim \log (d^2 \tilde{p})$. Then there exist constants $c_i > 0$ such that for any two unit vectors $u, v \in \mathbb{S}^{d-1}$, each of the following statements holds  with probability at least $1 - c_0 \exp[-(c_1 A^2-1) \log d^2 \tilde{p}]$:
	\begin{enumerate}
		\item[(i)] $\left| u^\top \left(\mathcal{X}^\top \mathcal{Z}/n \right)v \right| \le \left[ 2\pi \vertiii{f_x}  \left(1+A \, \sqrt{\log d^2 \tilde{p} \, / n} \right) \right]^{1/2} \sqrt{(1+q/n)} \Delta_\varepsilon \| \mathcal{B} \|_{2,1}$;
		\item[(ii)] $\left| u^\top \left(\mathcal{W}^\top \mathcal{Z}/n \right)v \right| \le \left(1+q/n \right)  \Delta_\varepsilon^2 \| \mathcal{B} \|_{2,1}$.  
	\end{enumerate}
\end{lemma}

\begin{proof}
	In order to prove (i), note that 
	\begin{eqnarray*}
		u^\top \left(\mathcal{X}^\top \mathcal{Z}/n \right)v &=& \frac{1}{n} \sum_{t=1}^n \left(u^\top x_t \right) \left[v^\top \sum_{k=0}^q B_k \left( \hat{\varepsilon}_{t-k} - \varepsilon_{t-k} \right) \right] \\
		&=& \sum_{k=0}^q \frac{1}{n} \sum_{t=1}^n \left(u^\top x_t \right) \left( v^\top B_k (\hat{\varepsilon}_{t-k} - \varepsilon_{t-k}) \right) \\
		&\le& \sum_{k=0}^q \left[\frac{1}{n} \sum_{t=1}^n \left(u^\top x_t \right)^2 \right]^{1/2} \left[ \frac{1}{n} \sum_{t=1}^n \left(v^\top B_k (\hat{\varepsilon}_{t-k} - \varepsilon_{t-k}) \right)^2 \right]^{1/2}.
	\end{eqnarray*}
	Using the concentration inequality in Proposition 2.4 and the upper bound on the spectral norm of population covariance matrix in Proposition 2.3 of  \citep{Basu15}, square of the first term in each summand is at most $2\pi \vertiii{f_x} (1+A\sqrt{\log d^2 \tilde{p} /n})$ with probability at least $1 - c_0 \exp \left[-(c_1A^2-1) \log d^2 \tilde{p} \right]$. Also, using the Cauchy-Schwarz inequality, square of the second term in the $k^{th}$ summand above satisfies, on the event $\mathcal{E}$, 
	\begin{equation*}
		\frac{1}{n} \sum_{t=1}^n \left[ \left(v^\top B_k (\hat{\varepsilon}_{t-k} - \varepsilon_{t-k}) \right)^2 \right] \le \frac{1}{n} \sum_{t=1}^n \|B_k \|^2 \| \hat{\varepsilon}_{t-k} - \varepsilon_{t-k} \|^2 = \|B_k \|^2 (1+q/n) \Delta_\varepsilon^2.
	\end{equation*}
	Together, this implies $|u^\top (\mathcal{X}^\top \mathcal{Z}/n)v|$ is upper bounded by \\
	$2\pi \vertiii{f_x} \sqrt{1+q/n}$ $ (1+A \sqrt{\log d^2 \tilde{p} / n})^{1/2} \left(\sum_{k=0}^q \|B_k\| \right) \Delta_{\varepsilon}$ with the specified  probability.
	\smallskip
	
	In order to prove (ii), note that 
	\begin{eqnarray*}
		&& u^\top \left(\mathcal{W}^\top \mathcal{Z}/n \right) v = \frac{1}{n} \sum_{t=1}^n \left(u^\top (\hat{\varepsilon}_{t-j} - \varepsilon_{t-j}) \right) \left( v^\top \sum_{k=0}^q B_k (\hat{\varepsilon}_{t-k} - \varepsilon_{t-k}) \right) \\
		&& \le \sum_{k=0}^q \left[ \frac{1}{n} \sum_{t=1}^n \left(u^\top (\hat{\varepsilon}_{t-j} - \varepsilon_{t-j}) \right)^2 \right]^{1/2} \left[ \frac{1}{n} \sum_{t=1}^n \left( v^\top B_k (\hat{\varepsilon}_{t-k} - \varepsilon_{t-k}) \right)^2 \right]^{1/2}.
	\end{eqnarray*}
	Using the argument above, we can check that on the event $\mathcal{E}$, the square of the first term in each summand is at most $(1+q/n) \Delta_{\varepsilon}^2$ and the square of the second term in the $k^{th}$ summand is at most $(1+q/n) \|B_k\|^2 \Delta_\varepsilon^2$. Putting things together, the right hand side of the above inequality is bounded above by $(1+q/n) \left( \sum_{k=0}^q \|B_k\| \right) \Delta_\varepsilon^2$. 
\end{proof}

\begin{lemma}\label{lem:conc_S_z}
	Consider $\hat{\varepsilon}_t$ and $\Delta_\varepsilon$ as in Lemma \ref{lem:ehat-e}, and $\Delta_a$ as defined in \eqref{eqn:delta_a}. Let $\tilde{\mathcal{Z}}$ be a data matrix consisting of $n$ consecutive observations from the time series  
	\newline
	$[y_{t-1}^\top, y_{t-2}^\top, \ldots, y_{t-p}^\top, \hat{\varepsilon}_{t-1}^\top, \hat{\varepsilon}_{t-2}^\top, \ldots, \hat{\varepsilon}_{t-q}^\top]^\top$, and $\mathcal{Z}$ a data matrix for \newline $\{z_t\} = [y_{t-1}^\top, y_{t-2}^\top, \ldots, y_{t-p}^\top, a_{t-1}^\top, a_{t-2}^\top, \ldots, a_{t-q}^\top]^\top$. Assume $n \succsim \log (d^2(p+q))$ and $\tilde{p} \ge p+q$. Then there exist universal constants $c_i > 0$ such that for any $u, v \in \mathbb{S}^{d(p+q)-1}$, with probability at least $1 - c_0 \exp[-(c_1 A^2-1)  \log d^2(p+q)]$, the following holds:
	\begin{align*}
		\left|u^\top \left(\tilde{\mathcal{Z}}^\top \tilde{\mathcal{Z}}/n - \Gamma_z(0) \right)v \right| &\le 2\pi \vertiii{f_z}(1+A \sqrt{\log d^2(p+q)/n}) + q(1+q/n)\Delta_a^2 \\
		+ 2 &\, \left[2\pi \vertiii{f_z} (1+A \sqrt{\log d^2(p+q)/n}) q(1+q/n) \right]^{1/2} \Delta_a.
	\end{align*}
\end{lemma}

\begin{proof}
	We begin by re-writing $\tilde{z}_t$ as $z_t + w_t$, where \\
	$w_t = \left[0^\top, \ldots, 0^\top, (\hat{\varepsilon}_{t-1} - a_{t-1})^\top, \ldots, (\hat{\varepsilon}_{t-q} - a_{t-q})^\top \right]^\top$. Then the following decomposition holds:
	\begin{eqnarray*}
		\left|u^\top \left(\tilde{\mathcal{Z}}^\top \tilde{\mathcal{Z}}/n - \Gamma_z(0)  \right)v \right| &\le&  \left|u^\top \left(\mathcal{Z}^\top \mathcal{Z}/n - \Gamma_z(0)  \right)v \right| + \left|u^\top \left(\mathcal{Z}^\top \mathcal{W}/n  \right)v \right|\\
		&+& \left|v^\top \left(\mathcal{Z}^\top \mathcal{W}/n  \right)u \right| + \left|u^\top \left(\mathcal{W}^\top \mathcal{W}/n \right)v \right|.
	\end{eqnarray*}
	Using Proposition 2.4 in \citep{Basu15}, we can obtain a high probability upper bound on the first term on the right hand side in terms of $\vertiii{f_z}$.
	In order to control the second and third terms, assume $v = [v_1^\top, v_2^\top, \ldots, v_{p+q}^\top]^\top$, where each $v_j \in \mathbb{R}^d$, and note that 
	\begin{eqnarray*}
		&~& \left|u^\top \left(\mathcal{Z}^\top \mathcal{W}/n \right)v \right| = \left|\frac{1}{n} \sum_{t=1}^n (u^\top z_t) (v^\top w_t) \right| \\
		&\le& \left[ \frac{1}{n} \sum_{t=1}^n (u^\top z_t)^2 \right]^{1/2} \left[ \frac{1}{n} \sum_{t=1}^n (v^\top w_t)^2 \right]^{1/2} \\
		&\le& \left[2\pi \vertiii{f_z} (1+A\sqrt{\log d(p+q)/n)} \right]^{1/2} \left[ \sum_{k=0}^q \frac{1}{n} \sum_{t=1}^n \left(v_{p+k}^\top (\hat{\varepsilon}_{t-k} - a_{t-k}) \right)^2 \right]^{1/2} \\
		&\le& \left[2\pi \vertiii{f_z} (1+A\sqrt{\log d(p+q)/n}) \right]^{1/2} \left[ \sum_{k=0}^q (1 + q/n) \| v_{p+k} \|^2 \Delta_a^2 \right]^{1/2}.
	\end{eqnarray*}
	The result follows by using the fact that on the event $\mathcal{E}$, square of the second term in the above product is upper bounded by $q(1+q/n) \Delta_{a}^2$, and noting that $u^\top (\mathcal{W}^\top \mathcal{W}/n)v = \frac{1}{n}\sum_{t=1}^n (u^\top w_t)(v^\top w_t) \le q(1+q/n) \Delta_a^2$. 
\end{proof}

\section{Proof of Proposition \ref{prop:elastic-net-fixed-X-E-v2} (Elastic Net)}\label{sec:pf-elnet}

\begin{proof}[Proof of Proposition \ref{prop:elastic-net-fixed-X-E-v2}]
	
	Set $\hat{\beta} \leftarrow \hat{\beta}^{(\alpha)}$,  
	and define $\beta^*_P:= \mathcal{P}_{\mathcal{S}} (\hat{\beta})$, the projection of $\hat{\beta}$ onto the affine space  $\mathcal{S}:= \left\{ \beta: \Sigma \beta = \rho \right\}$. Note that $\beta^* \in \mathcal{S}$. Set $v = \hat{\beta} - \beta^*$, $v_1 = \hat{\beta} - \beta^*_P$, $v_2 = \beta^*_P - \beta^*$. Then $v_2 \in \mathcal{N}(\Sigma)$, $v_1 \perp \mathcal{N}(\Sigma)$, and $\|v\|^2 = \|v_1\|^2+\|v_2\|^2$. Consider 
	\begin{equation*}
		\hat{\beta} \in \argmin_{\beta \in \mathbb{R}^{\bar{d}}} \frac{1}{n} \|Y - X \beta\|^2 + \lambda \left( \| \beta\|_1 + \frac{\alpha}{2 } \| \beta\|^2 \right) \mbox{ subject to } \|\beta\|_1 \le M
	\end{equation*}
	where $M \ge \| \beta^*\|_1$.
	
	Start with the basic inequality
	\begin{equation*}
		\frac{1}{n} \left\|Y - X \hat{\beta} \right\|^2 + \lambda \left( \| \hat{\beta} \|_1 + \frac{\alpha}{2 } \| \hat{\beta} \|^2 \right) \le \frac{1}{n} \left\|Y - X \beta^* \right\|^2 + \lambda \left( \| \beta^* \|_1 + \frac{\alpha}{2 } \| \beta^* \|^2 \right)
	\end{equation*}
	This implies
	\begin{equation*}
		\frac{1}{n} \| Xv \|^2 - \frac{2}{n} v^\top X^\top \varepsilon \le \lambda \left[ \left(  \| \beta^*\|_1 - \| \beta^* + v\|_1\right) + \frac{\alpha}{2 } \left( \| \beta^*\|^2 - \| \beta^* + v\|^2 \right)  \right]
	\end{equation*}
	Since $\|X^\top \varepsilon/n \|_\infty \le \lambda/2$, moving the second term to the right  we get 
	\begin{equation*}
		v^\top \left(X^\top X/n \right) v \le \lambda \|v\|_1 + \lambda \left[ \left(  \| \beta^*\|_1 - \| \beta^* + v\|_1\right) + \frac{\alpha}{2 } \left( \| \beta^*\|^2 - \| \beta^* + v\|^2 \right)  \right],
	\end{equation*}
	which in turn implies, by triangle inequality,
	\begin{equation*}
		v^\top \left(X^\top X/n \right) v \le \lambda \left[ 2\| \beta^*\|_1 + \frac{\alpha}{2 } \| \beta^*\|^2  \right] \le \lambda \left[ 2M + \alpha M^2/2  \right].
	\end{equation*}
	This implies 
	\begin{eqnarray*}
		v^\top \Sigma v &=& v^\top \left( \Sigma - X^\top X/n \right) v + v^\top \left(X^\top X/n \right) v\\
		&\le& \left\| \Sigma - X^\top X/n \right\|_\infty \|v\|^2_1 + \lambda \left[ 2M + \alpha M^2/2 \right] \\
		&\le& 4 q_n M^2 + \lambda \left[ 2M + \alpha M^2 / 2 \right], \mbox{ since } \|v\|_1 \le \| \hat{\beta}\|_1 + \| \beta^*\|_1 \le 2M.
	\end{eqnarray*}
	By the orthogonal decomposition $v = v_1+v_2$, we have 
	\begin{equation*}
		v^\top \Sigma v = v_1^\top \Sigma v_1 \ge \Lambda_{\min}^+\left(\Sigma \right) \|v_1\|^2.
	\end{equation*}
	Combining the above two inequalities, 
	\begin{equation*}
		\|v_1 \|^2 = \| \hat{\beta} - \beta^*_P\|^2  \le \frac{4 q_n M^2 + \lambda \left[ 2M + \alpha M^2 /2  \right]}{\Lambda_{\min}^+\left(\Sigma \right)}.
	\end{equation*}
	We restate the point identification result of part (c) in the form of a complete proposition \ref{prop:full-version-of-prop}.	
\end{proof}

\begin{proposition}
	\label{prop:full-version-of-prop} 
	Let $\Sigma\in\real^{D\times D}$ be a non-negative definite matrix with $\Lambda_{\min}(\Sigma) = 0$ and let $\rho\in\real^{D}$ be in the column space of $\Sigma$.  Consider the linear regression model  $y_{N \times 1} = X_{N \times {D}} \tba_{{D} \times 1} + \varepsilon_{N \times 1}$ with identified target
	\begin{equation*}
		\tba:= \argmin_{\beta} \left\{ \mathcal{P}_{\alpha} (\beta) \st \Sigma \beta = \rho \right\},
	\end{equation*}
	where $\mathcal{P}_{\alpha}(\beta) := \|\beta\|_1 + (\alpha/2) \|\beta\|^2$,
	and let
	\begin{equation*}
		\hba := \argmin_{\beta} \left\{ \mathcal{P}_{\alpha} (\beta) \st \frac{1}{n} \|y-X \beta\|^2 \le A_n, ~~ \|\beta\|_1 \le M \right\}
	\end{equation*}
	be the estimator.  On the event 
	\begin{equation*}
		\mathcal{E}:= \left\{ \left\|X^\top X/n - \Sigma \right\|_\infty \le q_n, \frac1{n}\left\|X^\top \varepsilon \right\|_\infty \le r_n, \left| \frac1{n}\left\| \varepsilon \right\|^2 - \sigma^2 \right| \le s_n \right\}
	\end{equation*}
	and choosing $A_n =\sigma^2 + s_n$ and $M\ge\|\tba\|_1$, the following holds:
	\begin{eqnarray*}
		&& \left\|\hba - \tba\right\|^2 \le 2v_n + 2(\sqrt{{D}}/\alpha+ M)v_n^{1/2}, \nonumber
	\end{eqnarray*}
	
	where $v_n:=\frac{4Mr_n + 2s_n + 4M^2q_n}{\Lambda_{\min}^+(\Sigma)}$
	and $\Lambda_{\min}^+(\Sigma)$ is the smallest non-zero eigenvalue of $\Sigma$. 
\end{proposition}

\begin{proof}[Proof of Proposition \ref{prop:full-version-of-prop}]	
	The estimator can be written as
	\begin{equation}\label{eqn:hba}
		\hba := \argmin_{\beta} \left\{ \mathcal{P}_{\alpha} (\beta) \st \beta\in\mathcal A_n, ~~ \|\beta\|_1 \le M \right\}, 
	\end{equation}
	where $\mathcal A_n=\{\beta:\frac{1}{n} \|y-X \beta\|^2 \le A_n\}$.  Our proof consists of a series of lemmas.
	We begin by relating the estimator's constraint set to the equivalence class of parameters that could have generated the data.
	\begin{lemma}\label{lem:contained}
		If $A_n \ge \sigma^2 + s_n$, then $\tba \in \mathcal A_n$ on the event $\mathcal{E}$.
	\end{lemma}
	\begin{proof}
		By the triangle inequality,
		$$
		\frac{1}{n} \|y-X \tba\|^2=\frac{1}{n} \|\varepsilon\|^2\le \sigma^2 + s_n.
		$$
	\end{proof}
	
	Our next lemma is a result about our estimator's in-sample prediction performance.
	\begin{lemma}\label{lem:pred-X}
		If we choose $A_n = \sigma^2 + s_n$, then on the event $\mathcal{E}$,  
		$$\frac1{n}\| X(\hba - \tba)\|^2\le 4Mr_n + 2s_n.
		$$
	\end{lemma}
	\begin{proof}
		We rewrite the inequality ${\frac1{n}\| y - X\hba\|^2 \leq A_n}$
		as $$\| X(\hba - \tba)\|^2 \leq n(A_n - \frac1{n}\|\varepsilon\|^2) + 2\varepsilon^\top X(\tba - \hba).$$
		Our choice of $A_n$ means that 
		$$A_n - \frac1{n}\|\varepsilon\|^2= \sigma^2 + s_n- \frac1{n}\|\varepsilon\|^2\le  \left|\sigma^2- \frac1{n}\|\varepsilon\|^2\right|  + s_n.$$
		On the event $\mathcal E$, we know that this is bounded by $2s_n$. Thus, 
		\begin{align*}
			\|X(\hba-\tba)\|^2\le 2ns_n + 2\varepsilon^\top X(\tba - \hba).
		\end{align*}
		Furthermore,
		\begin{align*}
			\|X(\hba-\tba)\|^2\le2ns_n + 2\|X^\top\varepsilon\|_\infty\cdot \| \tba - \hba\|_1.
		\end{align*}
		Dividing both sides by $n$ and 
		recalling the definition of $r_n$ (through the event $\mathcal E$) gives
		\begin{align*}
			\frac1{n}\|X(\hba - \tba)\|^2\le2s_n + 2r_n\| \tba - \hba\|_1.
		\end{align*}
		The triangle inequality and recalling that both vectors are bounded by $M$ in $\ell_1$ norm gives
		$$
		\frac1{n}\| X(\hba - \tba)\|^2 \le 4Mr_n + 2s_n.
		$$
	\end{proof}

	Our next lemma extends this prediction result from $X$ to $\Sigma^{1/2}$.
	\begin{lemma}
		If we choose $A_n = \sigma^2 + s_n$, then on the event $\mathcal{E}$,  
		$$\|\Sigma^{1/2}(\hba-\tba)\|^2\le 4Mr_n + 2s_n + 4M^2q_n.$$
	\end{lemma}
	\begin{proof}
		Writing
		$$
		\|\Sigma^{1/2}(\hba-\tba)\|^2=\frac{1}{n}\|X(\hba-\tba)\|^2+(\hba-\tba)^\top(\Sigma-\frac1{n}X^TX)(\hba-\tba),
		$$
		we apply Lemma \ref{lem:pred-X} to get
		$$
		\|\Sigma^{1/2}(\hba-\tba)\|^2\le 4Mr_n + 2s_n +(\hba-\tba)^\top(\Sigma-\frac1{n}X^TX)(\hba-\tba).
		$$
		Now, for a matrix $A$, $v^\top A v=\sum_{ij}v_iA_{ij}v_j\le \|A\|_\infty\sum_{ij}|v_i||v_j|=\|A\|_\infty\|v\|_1^2$, and recalling the definition of $q_n$ (through the event $\mathcal E$)  we have
		$$
		\|\Sigma^{1/2}(\hba-\tba)\|^2\le 4Mr_n + 2s_n +q_n\|\hba-\tba\|_1^2.
		$$
		The result follows by the triangle inequality and that both vectors are bounded by $M$ in $\ell_1$ norm.
	\end{proof}
	At this point, we move from prediction bounds to estimation bounds.
	Our next step is to translate the previous result to a statement about our estimator not being too far from the set of possible parameters that generated our data, that is the affine space $\{\beta:\Sigma\beta=\rho\}$.
	
	\begin{lemma}\label{lem:dist-perp}
		Let $\hba_P$ denote the projection of $\hba$ onto the affine subspace $\{\beta:\Sigma\beta=\rho\}$:
		$$
		\hba_P:=\arg\min_{\beta}\left\{ \|\hba-\beta\|^2\st \Sigma\beta=\rho\right\}.
		$$
		If we choose $A_n =\sigma^2 + s_n$, then on the event $\mathcal{E}$,  
		$$
		\|\hba-\hba_P\|^2\le v_n,
		$$
		where 
		$$v_n:=\frac{4Mr_n + 2s_n + 4M^2q_n}{\Lambda_{\min}^+(\Sigma)}$$
		and $\Lambda_{\min}^+(\Sigma)$ is the smallest non-zero eigenvalue of $\Sigma$.
	\end{lemma}
	\begin{proof}
		The distance of $\hba$ to the affine space is given by
		\begin{align*}
			\|\hba-\hba_P\|^2&=\min_{\beta}\left\{ \|\hba-\beta\|^2\st \Sigma\beta=\rho\right\}\\
			&=\min_{\delta}\left\{ \|\hba-\tba-\delta\|^2\st \Sigma\delta=0\right\}\\
			&=\|\Sigma\Sigma^{+}(\hba-\tba)\|^2
		\end{align*}
		where in the second equality we use that $\Sigma\tba=\rho$ and in the third equality we use that the row space and null space are orthogonal complements and therefore the residual after projecting onto the null space is equivalent to the projection onto the row space of $\Sigma$ (and here the row space and column space are identical).  Now, $\Sigma\Sigma^+=(\Sigma^{1/2})^+\Sigma^{1/2}$ and so
		\begin{align*}
			\|\hba-\hba_P\|^2&=\|(\Sigma^{1/2})^+\Sigma^{1/2}(\hba-\tba)\|^2\\
			&\le\|(\Sigma^{1/2})^+\|^2\|\Sigma^{1/2}(\hba-\tba)\|^2\\
			&\le\|\Sigma^{1/2}(\hba-\tba)\|^2/\Lambda_{\min}^+(\Sigma).
		\end{align*}
		The result follows from the previous lemma. 
	\end{proof}
	At this point, we have bounded the distance between our estimator and the identified target in the direction orthogonal to the affine space.  The remainder of the proof of the proposition is aimed at bounding the distance along the affine space.  To do so, we make use of the strong convexity of the objective function $\mathcal{P}_{\alpha}$.
	
	\begin{lemma}
		Under the same setup and conditions as the previous lemma,
		$$
		\mathcal{P}_{\alpha}(\hba_P) - \mathcal{P}_{\alpha}(\hba) \le (\sqrt{D}+\alpha M)v_n^{1/2} + \frac{\alpha}{2}v_n,
		$$
		where $v_n$ is defined in Lemma \ref{lem:dist-perp}.
	\end{lemma}
	\begin{proof}
		By the triangle inequality,
		$\|\hba_P\|_1-\|\hba\|_1\le\|\hba_P-\hba\|_1
		$ and 
		$\|\hba_P\|\le\|\hba\|+\|\hba_P-\hba\|$.
		Squaring this second inequality gives
		\begin{align*}
			\|\hba_P\|^2&\le\|\hba\|^2+\|\hba_P-\hba\|^2+2\|\hba\|\cdot\|\hba_P-\hba\|\\
			&\le \|\hba\|^2+\|\hba_P-\hba\|^2+2M\|\hba_P-\hba\|
		\end{align*}
		Thus,
		\begin{align*}
			\mathcal{P}_{\alpha}(\hba_P) - \mathcal{P}_{\alpha}(\hba)&\le \|\hba_P-\hba\|_1 + \frac{\alpha}{2}\left(\|\hba_P-\hba\|^2+2M\|\hba_P-\hba\|\right)\\
			&\le (\sqrt{D}+\alpha M)\|\hba_P-\hba\| + \frac{\alpha}{2}\|\hba_P-\hba\|^2.
		\end{align*}
		The result follows from the previous lemma.
	\end{proof}
	
	\begin{lemma}\label{lem:dist-parallel}
		Let $\hba_P$ be the projection of $\hba$ onto $\{\beta:\Sigma\beta=\rho\}$.  If we choose $A_n = \sigma^2 + s_n$ and $M\ge \|\tba\|_1$, then on the event $\mathcal{E}$, 
		$$
		\|\hba_P-\tba\|^2\le 2(\sqrt{D}/\alpha+ M)v_n^{1/2} + v_n,
		$$
		where $v_n:=\frac{4Mr_n + 2s_n + 4M^2q_n}{\Lambda_{\min}^+(\Sigma)}.$
	\end{lemma}
	\begin{proof}
		By $\alpha$-strong convexity of $\mathcal{P}_{\alpha}$ and the definition of $\tba$, we have that for any $\gamma$ such that $\Sigma\gamma=\rho$,
		$$
		\mathcal{P}_{\alpha}(\gamma)\ge \mathcal{P}_{\alpha}(\tba)+\frac{\alpha}{2}\|\gamma-\tba\|^2
		$$
		Substituting $\hba_P$ for $\gamma$ and rearranging terms gives
		$$
		\|\hba_P-\tba\|^2\le (2/\alpha)\left[\mathcal{P}_{\alpha}(\hba_P)- \mathcal{P}_{\alpha}(\tba)\right].
		$$
		By Lemma \ref{lem:contained}, $\tba\in\mathcal A_n$ and by assumption $\|\tba\|_1\le M$, thus $\tba$ is feasible for \eqref{eqn:hba}, meaning that $\mathcal{P}_{\alpha}(\hba)\le \mathcal{P}_{\alpha}(\tba)$.  Thus,
		$$
		\|\hba_P-\tba\|^2\le (2/\alpha)\left[\mathcal{P}_{\alpha}(\hba_P)- \mathcal{P}_{\alpha}(\hba)\right].
		$$
		We apply the previous lemma to the right-hand side to conclude the proof.
	\end{proof}
	
	The results of Lemma \ref{lem:dist-perp} and Lemma \ref{lem:dist-parallel} can now be combined to give the desired estimation result:
	\begin{align*}
		\|\hba-\tba\|^2 &= \|\hba-\hba_P\|^2 + \|\hba_P-\tba\|^2\\
		&\le v_n + 2(\sqrt{D}/\alpha+ M)v_n^{1/2} + v_n\\
		&\le 2v_n + 2(\sqrt{D}/\alpha+ M)v_n^{1/2}.
	\end{align*}
	This establishes the proposition.
\end{proof}

\section{Proof of Proposition \ref{prop:est-error-main} (Phase-I Analysis)}\label{app:phase-I}

We divide the proof of Proposition \ref{prop:est-error-main} in four steps. First, in Proposition \ref{prop:est-error}, we provide deterministic upper bounds on the estimation errors ($\widehat{\Pi} - \Pi)$ and approximation errors around the \textit{regression} residuals  $(\hat{\varepsilon}_t - \varepsilon_t)$ for a given realization of $(T+\tilde{p})$ consecutive observations from the VARMA process, under some sufficient conditions. Then we show in Propositions \ref{prop:dev-bound} and \ref{prop:re} that for a random realization from the VARMA process, these conditions are satisfied with high probability when the sample size is sufficiently large. Finally, we provide upper bound on the approximation errors around the \textit{true} VARMA errors  $(\hat{\varepsilon}_t - a_t)$ in Proposition \ref{prop:bound-eps-a}.

We start with the deterministic upper bound on the deviation of the estimated residuals $\hat{\varepsilon}_t$ around $\varepsilon_t$ without making any assumption on the design matrix $Z$. This is essentially a so-called ``slow rate'' bound, as appears in the lasso regression literature \citep{greenshtein2004persistence}. Then we provide a tighter upper bound on the above deviation, and an upper bound on the deviation of $\{\widehat{\Pi}_{\tau}\}_{\tau=1}^{\tilde{p}}$ around $\{\Pi_{\tau}\}_{\tau=1}^{\tilde{p}}$, under a restricted eigenvalue (RE) condition \citep{powai2012, Basu15}:
\begin{assumption}[Restricted Eigenvalue, RE]\label{assump:re}
	A symmetric matrix $G_{r \times r}$ satisfies the restricted eigenvalue (RE) condition with curvature $\gamma >0$ and tolerance  ${\delta} > 0$ if 
	\begin{equation}\label{eqn:re}
		v^\top G v \ge \gamma \|v\|^2 - {\delta} \|v\|_1^2, \mbox{ ~~ for all ~~} v \in \mathbb{R}^r.
	\end{equation}
\end{assumption}

These upper bounds involve the curvature and tolerance parameters $\gamma, \, {\delta}$ as well as the quantity $\|Z^\top\mathcal{E}/T\|_{\infty}$, and do not relate directly to the VARMA  parameters. Propositions \ref{prop:dev-bound} and \ref{prop:re} 
% in Appendix \ref{app:phase-I} 
then provide insight into how these quantities depend on VARMA parameters, when we have a random realization from a stable, invertible VARMA model \eqref{VARMA}.

\begin{proposition}\label{prop:est-error}
	Consider any solution $\hat{\beta}$ of \eqref{eqn:l1-ls} using a given realization of $\{ y_t\}_{t=1-\tilde{p}}^T$ from the VARMA model \eqref{VARMA}, and set $\mathcal{E} = \text{vec}(E)$. Then, for any choice of the penalty parameter $\lambda \ge 2 \left\| Z^\top \mathcal{E}/T \right\|_\infty$, we have
	\begin{eqnarray}\label{eqn:ehat-e}
		\frac{1}{T} \sum_{t=1}^T \left\|\hat{\varepsilon}_t - \varepsilon_t \right\|^2 \le 2 \lambda \sum_{\tau = 1}^{\tilde{p}} \left\| \Pi_\tau \right\|_1 =: \Delta_\varepsilon^2.
	\end{eqnarray}
	Further, assume $\{\Pi_1, \ldots, \Pi_{\tilde{p}} \}$ are sparse so that $k:= \sum_{\tau=1}^{\tilde{p}} \left\| \Pi_{\tau} \right\|_0$, and the sample Gram matrix $Z^\top Z/T$ satisfies $RE(\gamma, \delta)$ of Assumption \ref{assump:re} for some model dependent quantities $\gamma > 0, \delta > 0$ such that $k \delta \le \gamma/32$. Then for any choice of $\lambda \ge 4 \left\| Z^\top \mathcal{E}/T \right\|_\infty$, we have the following upper bounds
	\begin{eqnarray}\label{eqn:fastrate}
		&& \sum_{\tau=1}^{\tilde{p}} \left\| \widehat{\Pi}_\tau - \Pi_\tau \right\|_1 \le 64 k \lambda / \gamma, 
		\left[ \sum_{\tau=1}^{\tilde{p}} \left\| \widehat{\Pi}_\tau - \Pi_\tau \right\|_F^2 \right]^{1/2} \le 16 \sqrt{k} \lambda / \gamma, \nonumber\\
		&& \frac{1}{T}\sum_{t=1}^T \left\|\hat{\varepsilon}_t - \varepsilon_t \right\|^2 \le 128 {k} \lambda^2 / \gamma =: \Delta_\varepsilon^2.
	\end{eqnarray}
\end{proposition}

\begin{proof}[Proof of Proposition \ref{prop:est-error}]
	Since $\hat{\beta}$ is a minimizer of \eqref{eqn:l1-ls}, we have
	\begin{equation}
		\frac{1}{T} \left\| Y - Z \hat{\beta} \right\|^2 + \lambda \left\| \hat{\beta} \right\|_1 \le \frac{1}{T} \left\|Y - Z \beta^* \right\|^2 + \lambda \left\| \beta^* \right \|_1. \nonumber
	\end{equation}
	Let $v = \hat{\beta} - \beta^*$ denote the error vector. Substituting $Y = Z \beta^* + \mathcal{E}$ in the above, we obtain 
	\begin{equation}
		\frac{1}{T} \left\|\mathcal{E} - Zv \right\|^2 + \lambda \left\|\beta^* + v \right\|_1 \le \frac{1}{T} \left\| \mathcal{E} \right\|^2 + \lambda \left\| \beta^* \right\|_1. \nonumber
	\end{equation}
	Moving some terms to the right hand side of the inequality, we get
	\begin{equation}\label{eqn:est-error-eqn1}
		v^\top \left(Z^\top Z/T\right)v \le 2 v^\top  \left(Z^\top \mathcal{E}/T\right) + \lambda \left(\| \beta^*\|_1 - \| \beta^* + v\|_1 \right). 
	\end{equation}
	Since $\lambda \ge 2\|Z^\top \mathcal{E}/T\|_{\infty}$, and the first term on the right is at most $2\|v\|_1 \| Z^\top \mathcal{E}/T\|_{\infty}$, we have
	\begin{equation}
		v^\top \left(Z^\top Z/T\right)v \le \lambda \left( \|v\|_1 + \| \beta^* \|_1 - \| \beta^* + v \|_1 \right) \le 2 \lambda \| \beta^* \|_1 = 2 \lambda \sum_{\tau=1}^{\tilde{p}} \| \Pi_\tau \|_1 \nonumber.
	\end{equation}
	Then \eqref{eqn:ehat-e} follows from the fact that 
	\begin{eqnarray*}
		v^\top \left( Z^\top Z/T\right)v = \frac{1}{T} \left\|\mathcal{X} \widehat{B} - \mathcal{X} B \right\|_F^2 = \frac{1}{T} \sum_{t=1}^T \| \hat{\varepsilon}_t - \varepsilon_t \|^2.
	\end{eqnarray*}
	Next, suppose $J$ denotes the support of $\beta^*$, i.e. $J = \left\{j \in \{1, \ldots, d^2 \tilde{p} \}: \beta^*_j \neq 0 \right\}$. By our assumption, $|J| \le k$. Inequality \eqref{eqn:est-error-eqn1}, together with our choice of $\lambda$, then leads to 
	\begin{eqnarray}
		0 \le v^\top \left( Z^\top Z/T \right)v & \le& \frac{\lambda}{2} \left(\|v_J\|_1 + \|v_{J^c}\|_1 \right) + \lambda \left ( \| \beta^*_J\|_1 - \| \beta^*_J + v_J \|_1 - \|v_{J^c} \|_1 \right) \nonumber \\
		&\le& \frac{\lambda}{2} \left(\|v_J\|_1 + \|v_{J^c}\|_1 \right) + \lambda \left ( \| v_J\|_1 - \|v_{J^c} \|_1 \right) \nonumber \\
		&\le& \frac{3\lambda}{2} \|v_J\|_1 - \frac{\lambda}{2} \|v_{J^c}\|_1  \le 2 \lambda \|v_J \|_1 \le 2 \lambda \|v\|_1. \nonumber
	\end{eqnarray}
	Since $\lambda > 0$, the first inequality on the last line ensures $\|v_{J^c}\|_1 \le 3 \|v_J\|_1$, so that $\|v\|_1 \le 4 \|v_J\|_1 \le 4\sqrt{k} \|v\|$. Using the RE condition \eqref{eqn:re} and the upper bound on $k\delta$, we have 
	\begin{equation}
		v^\top \left( Z^\top Z/T\right)v \ge \gamma \|v\|^2 - \delta \|v\|_1^2 \ge (\gamma - 16k \delta) \|v\|^2 \ge \frac{\gamma}{2} \|v\|^2. \nonumber
	\end{equation}
	Combining these upper and lower bounds on $v^\top \left( Z^\top Z/T\right)v$, we obtain the final inequalities as follows:
	\begin{eqnarray*}
		&& \gamma \|v\|^2/2 \le v^\top \left( Z^\top Z/T\right)v \le 8 \lambda \sqrt{k} \|v\| \\
		&\Rightarrow& \|v\| \le 16 \lambda \sqrt{k} / \gamma, 
	\end{eqnarray*}
	and consequently $\|v\|_1 \le 4 \sqrt{k} \|v\| \le 64k \lambda / \gamma$. 
	
	Together with $v^\top \left( Z^\top Z/T\right)v \le 2 \lambda \|v\|_1$, we obtain the final in-sample prediction error bound $128 k \lambda^2 / \gamma$.
\end{proof}

\smallskip

Our next proposition provides a non-asymptotic upper bound on $\|Z^\top\mathcal{E}/T\|_{\infty}$ which holds with high probability for large $d, \tilde{p}$. If $\lambda$ is chosen as the same order of this bound, Proposition \ref{prop:est-error}
%inequalities \eqref{eqn:slowrate}-\eqref{eqn:fastrate}  
then shows how the upper bounds of estimation and approximation errors vary with model parameters. 
\begin{proposition}[Deviation Condition: Phase-I]\label{prop:dev-bound}
	If $\{y_{-(\tilde{p}-1)}, \ldots, y_T \}$ is a random realization from a stable, invertible VARMA model \eqref{VARMA}, then there exist  universal constants $c_i > 0$ such that for any $A > 1$, with probability at least $1 - c_0 \exp \! \left[ -(c_1 A^2-1) \log d^2 \tilde{p}\right]$, 
	\begin{eqnarray*}
		\|Z^\top\mathcal{E}/T \|_{\infty} \le 2\pi   \vertiii{f_y} \left[ 3A \, \max \left\{ \vertiii{\Pi_{[\tilde{p}]}}^2, \, 1 \right\} \sqrt{\log (d^2\tilde{p})/T} + \left\| \Pi_{-[\tilde{p}]} \right\|_{2,1} \right].
	\end{eqnarray*}
\end{proposition}
\smallskip

\begin{proof}[Proof of Proposition \ref{prop:dev-bound}]
	Note that $\|\tilde{Z}^\top \mathcal{\mathcal{E}}/T \|_{\infty} = \|\mathcal{X}^\top E/T \|_\infty = \max_{1 \le h \le \tilde{p}} \|\mathcal{X}_{(h)}^\top E/T \|_{\infty}$, where $\mathcal{X}_{(h)} = [(y_{T-h}) :\ldots:(y_{1-h}) ]^\top $. 
	
	Define $X_t = y_{t-h} = L^hy_t$ and $Y_t = \varepsilon_t = a_t + \sum_{\tau = \tilde{p}+1}^\infty \Pi_\tau y_{t-\tau} = \Pi_{[\tilde{p}]}(L)y_t$. The first term in our upper bound follows from \eqref{eqn:devn-xy} in Proposition \ref{lem:yAy}, by using $X_t = L^h y_t$, $Y_t = \varepsilon_t = \Pi_{[\tilde{p}]}(L)y_t$ and $\eta = A \sqrt{\log d^2 \tilde{p}/T}$. To obtain the second term, i.e. the bound on the bias term $\Gamma_{X,Y}(0)$, we use the representation $Y_t = \varepsilon_t =  a_t + \sum_{t = \tilde{p}+1}^\infty \Pi_\tau y_{t-\tau}$ as follows:
	\begin{equation}
		\Gamma_{X,Y}(0) = \cov \left( y_{t-h}, a_t + \sum_{\tau = \tilde{p}+1}^\infty \Pi_\tau y_{t-\tau} \right) = \sum_{\tau = \tilde{p}+1}^\infty \Gamma_y(h-\tau) \Pi_\tau^\top . \nonumber
	\end{equation}
	First, note that the entries of $\Gamma_{X,Y}(0)$ are upper bounded as follows:
	\begin{equation}
		\left\| \Gamma_{X,Y}(0) \right\|_\infty \le \left\| \Gamma_{X,Y}(0) \right\| \le \left( \max_{h \in \mathbb{Z}} \| \Gamma_y(h) \| \right) \left\|\Pi_{-[\tilde{p}]} \right\|_{2,1} \le 2 \pi \vertiii{f_y} \left\|\Pi_{-[\tilde{p}]} \right\|_{2,1} \nonumber
	\end{equation}
	The last inequality holds since for any $h \in \mathbb{Z}$, 	$\Gamma_y(h) = \int_{-\pi}^{\pi} e^{ih \theta} f_y(\theta) d\theta$. 
	
\end{proof}

\smallskip

The next proposition investigates sample size requirements for the RE condition to hold with high probability, and also provides insight into how the tolerance and curvature parameters depend on the VARMA model parameters. 

\begin{proposition}[Verifying Restricted Eigenvalue Condition]\label{prop:re}
	Consider a random realization of $(T+\tilde{p})$ data points $\{y_{-(\tilde{p}-1)}, \ldots, y_T\}$ from a stable, invertible VARMA model \eqref{VARMA} with $\Lambda_{\min}(\Sigma_a)>0$. 
	Then there exist universal constants $c_i > 0$ such that for $T \succsim \max\{\omega^2, 1\} k (\log d + \log \tilde{p})$, the matrix $Z^\top Z/T$ satisfies RE($\gamma, \delta$) with probability at least $1-c_1 \exp(-c_2 T min\{ \omega^{-2}, 1\})$, where 
	\begin{equation*}
		\gamma = \pi /\vertiii{f_y^{-1}}, \, \omega = c_3 \tilde{p} \vertiii{f_y} \vertiii{f^{-1}_y}, \, \delta = \gamma \max\{\omega^2, 1\}\log (d\tilde{p})/T.
	\end{equation*}
\end{proposition}

\begin{proof}[Proof of Proposition \ref{prop:re}]
	The proof follows along the same line of arguments as in Proposition 4.2 of  \citep{Basu15}, where the restricted eigenvalue condition was verified for processes  $\{y_t\}$ generated according to a \textit{finite-order} VAR process. In particular, rows of the design matrix were generated from  the process $\tilde{y}_t = [y_t^\top , \ldots, y_{t-\tilde{p}+1}^\top ]^\top $ allowing a VAR(1) representation with closed form expressions of spectral density and autocovariance. In the present context, $\{\tilde{y}_t\}$ does not have a VAR representation. However, a close inspection of the proof in \citep{Basu15} shows that it is sufficient to derive a lower bound on $\Lambda_{\min}(\Gamma_{\tilde{y}}(0))$ and an upper bound on $\vertiii{f_{\tilde{y}}}$, and the rest of the argument follows. Next, we derive these two bounds for the process $\{\tilde{y}_t\}$.
	
	First we consider $\Lambda_{\min}(\Gamma_{\tilde{y}}(0))$. Note that $\Gamma_{\tilde{y}}(0)$ can be viewed as the variance-covariance of a vectorized data matrix containing $\tilde{p}$ consecutive observations from the process $y_t$. Hence, using Proposition 2.3 and Equation (2.6) of  \citep{Basu15}, we can show that 
	\begin{equation*}
		\Lambda_{\min}\left(\Gamma_{\tilde{y}}(0) \right) \ge \min_{\theta \in [-\pi, \pi]} 2\pi \Lambda_{\min} \left(f_{y}(\theta) \right) = 2 \pi \vertiii{f^{-1}_y}^{-1}. \nonumber
	\end{equation*}
	
	The upper bound on $\vertiii{f_{\tilde{y}}}$ follows from Proposition \ref{lem:f-y-eps},  by setting $\mathcal{C}(L)=0$, which implies  $\vertiii{f_{\tilde{y}}} \le \tilde{p} \vertiii{f_y}$. 
\end{proof}

\smallskip

\begin{proposition}\label{prop:bound-eps-a}
	Consider the Phase-I regression residuals $\hat{\varepsilon}_t$ in Proposition \ref{prop:est-error}. Assume  $(1/T)\sum_{t=1}^T \left\| \hat{\varepsilon}_t - \varepsilon_t \right\|^2 \le \Delta_{\varepsilon}^2$ with probability at least $1 - c_0 \exp[-(c_1 A^2 - 1) \log d^2 \tilde{p}]$ for some universal constants $c_i > 0$, and $T \succsim \log (d^2 \tilde{p})$. Then there exist $c_i > 0$ such that 
	\begin{enumerate}[label=(\alph*)]
		\item For any $v \in \mathbb{S}^{d-1}$, with probability at least $1 - c_0 \exp[-(c_1 A^2 - 1) \log (d^2 \tilde{p})]$,
		\begin{equation}\label{eqn:delta_a}
			\frac{1}{T} \sum_{t=1}^T \left( v^\top (\hat{\varepsilon}_{t} - a_{t}) \right)^2 \le 4 \max \left\{\Delta_\varepsilon^2, 4 \pi \left\| \Pi_{-[\tilde{p}]} \right\|^2_{2,1} \vertiii{f_y} \right\} =: \Delta_a^2.
		\end{equation}
		\item In particular, with probability at least $1 - c_0 \exp[-(c_1 A^2 - 2) \log (d^2 \tilde{p})]$,
		\begin{equation}\label{eqn:delta_a_eltwise}
			\max_{1 \le j \le d} \, \,  \frac{1}{T} \sum_{t=1}^T \left( \hat{\varepsilon}_{tj} - a_{tj}\right)^2 \le  \Delta_a^2.
		\end{equation}
		\item With probability at least $1 - c_0 \exp[-(c_1 A^2 - 2) \log (d^2 \tilde{p})]$,
		\begin{equation}\label{eqn:delta_a_total}
			\frac{1}{T} \sum_{t=1}^T \left\| \hat{\varepsilon}_t - a_t\right\|^2 \le 4 \max \left\{\Delta_\varepsilon^2, 4 \pi d \left\| \Pi_{-[\tilde{p}]} \right\|^2_{2,1} \vertiii{f_y} \right\}.
		\end{equation}
	\end{enumerate}
\end{proposition}

\begin{proof}
	We use the decomposition $\hat{\varepsilon}_t - a_t = (\hat{\varepsilon}_t - \varepsilon_t) + (\varepsilon_t- a_t)$ and analyze the sum of squares for the two parts separately. 
	In particular, note that 
	\begin{equation*}
		\frac{1}{T} \sum_{t=1}^T \left( v^\top (\hat{\varepsilon}_t - a_t) \right)^2 \le 4 \max \left\{ \frac{1}{T} \sum_{t=1}^T \left(  v^\top (\hat{\varepsilon}_t - \varepsilon_t)\right)^2, \frac{1}{T} \sum_{t=1}^T \left(v^\top (\varepsilon_t - a_t) \right)^2\right\}.
	\end{equation*}
	
	By assumption, the first part is at most $\Delta_\varepsilon^2$ with probability at least $1 - c_0 \exp \left[-(c_1 A^2 - 1) \log d^2 \tilde{p} \right]$. To work with the second part, note that $\varepsilon_t - a_t = \Pi_{-[\tilde{p}]}(L)y_t =:w_t$, say. The spectral density of $w_t$ satisfies $\vertiii{f_w} \le \left\| \Pi_{-[\tilde{p}]}\right\|_{2,1} \vertiii{f_y}$. Using Propositions 2.3 and 2.4 of \citep{Basu15}, we obtain the following upper bound for any $\eta > 0$,
	\begin{equation*}
		\mathbb{P} \left[v^\top \left(\frac{1}{T}w_t w_t^\top \right)v  > 2 \pi \vertiii{f_w}(1 + \eta)\right] \le 2 \exp \left[-c_0 T \min \left\{\eta, \eta^2 \right\} \right].
	\end{equation*}
	Setting $\eta = (c_1 A^2 - 1) \log d^2 \tilde{p}/T $ (note that $\eta < 1$ when $T \succsim \log d^2 \tilde{p}$), we conclude that the second term is at most $4 \pi \vertiii{f_w}$ with probability at least $1 - 2 \exp \left[-(c_1 A^2 - 1) \log d^2 \tilde{p} \right]$.

	The second inequality \eqref{eqn:delta_a_eltwise} follows by taking an union bound over the choices $v = e_1, \ldots, e_d$, the unit vectors, and multiplying the tail probability by $d^2 \tilde{p} \ge d$. The third inequality \eqref{eqn:delta_a_total} follows by adding up these $d$ terms corresponding to the $d$ unit vectors.
\end{proof}

\begin{proof}[Proof of Proposition \ref{prop:est-error-main}]
	The slow rate bounds follow from Propositions \ref{prop:est-error},  \ref{prop:dev-bound} and \ref{prop:bound-eps-a}. To establish the fast rate, note that by Proposition \ref{prop:re}, the RE condition with $\gamma = \pi/\vertiii{f_y^{-1}}$, $\omega = c_3 \tilde{p} \vertiii{f_y} \vertiii{f_y^{-1}}$ and $\delta = \gamma \max \{ \omega^2, 1\} \log (d \tilde{p} )/T$ holds with probability at least $1 - c_1 \exp \left[ -c_2 k \log (d \tilde{p}) \right]$, for $T \succsim \max \{\omega^2, 1 \} k (\log d \tilde{p})$. Since $k (\log d \tilde{p}) \ge \log (d^2 \tilde{p})$ for $k \ge 2$, the event where both RE and deviation condition of Proposition \ref{prop:dev-bound} hold has probability at least $1 - c_0 \exp \left[-(c_1A^2 - 1) \log (d^2 \tilde{p}) \right]$ for some universal constants $c_i > 0$. These choices also ensure $k \delta / \gamma = \max \{\omega^2, 1\} \log (d \tilde{p})/T \le 1/32$ for large enough $T$. Plugging in the value of $\gamma$ in the final inequality of Proposition \ref{prop:est-error} leads to the tighter upper bound $\Delta^2_\varepsilon = 128 k \lambda^2 / \gamma = (128/\pi) \vertiii{f_y^{-1}} k \lambda^2$.
\end{proof}

\section{Propositions and Proofs for Phase-II Analysis}\label{app:phase-II}
Before presenting the proof of \ref{prop:phase2-qrs},  in Proposition \ref{prop:prediction-phase2} we provide a high probability upper bound on $\left\| \tilde{\mathcal{Z}}^\top \mathcal{U}/n \right\|_\infty$, which is required for the choices of both $\lambda$ in the penalized version and $r_n$ in the constrained version. Deriving upper bounds for the other quantities $q_n$, $s_n$ follow similar arguments, and an outline is provided in the proof of Proposition \ref{prop:phase2-qrs}.
\begin{proposition}[Deviation Bound: Phase-II]\label{prop:prediction-phase2}
	There exist universal constants $c_i > 0$ such that if $n \succsim \log \, d^2(p+q)$, then for any $A > 0$ the following holds with probability at least $1 - c_0 \exp \! \left[-(c_1 A^2-2)\, \log \, d^2(p+q) \right]$:
	\begin{equation*}
		\left\| \tilde{\mathcal{Z}}^\top \mathcal{U}/n \right\|_\infty \le \varphi_1 \sqrt{\frac{\log d^2 (p +q)}{n}} + \varphi_2\cdot \left(\Delta_\varepsilon + \Delta_\varepsilon^2 + \left\|\Pi_{-[\tilde{p}]} \right\|_{2,1} \right),
	\end{equation*}
	where 
	\begin{eqnarray*}
		\varphi_1 &=& c_1 \vertiii{f_y} A \,  \max \left\{1, \vertiii{\Theta}^2 \left\| \Pi_{-[\tilde{p}]} \right\|_{2,1}^2, \vertiii{\Pi_{[\tilde{p}]}}^2 \right\}, \\
		\varphi_2 &=& c_2 \vertiii{f_y} \left\|  \Theta \right\|_{2,1} \max \! \{ 1, \left\| \Pi_{[\tilde{p}]} \right\|_{2,1} \}.
	\end{eqnarray*} 
\end{proposition}

\begin{proof}[Proof of Proposition \ref{prop:prediction-phase2}]
	Recall $n = T-q$ is the number of observations in the Phase-II regression. The element-wise maximum norm can be expressed as 
	\begin{equation*}
		\left\| \mathcal{Z}^\top \mathcal{U}/n \right\|_\infty = \max_{\begin{array}{c} 1 \le \ell \le p \\ 1 \le m \le q \end{array}} ~ 
		\max \left\{ \left\|\mathcal{Y}_{(\ell)}^\top \mathcal{U}/n \right\|_\infty, \, \, \left\| \hat{E}_{(m)}^\top \mathcal{U}/n \right\|_{\infty} \right\},
	\end{equation*}
	where $\mathcal{Y}_{(\ell)} = \left[ y_{n-\ell}: \ldots : y_{1-\ell} \right]^\top$ is a data matrix with $n$ consecutive observations from the process $\{y_t \}$,  $\hat{E}_{(m)} = \left[\hat{\varepsilon}_{n-m}: \ldots : \hat{\varepsilon}_{1-m} \right]^\top$ is a data matrix with $n$ consecutive observations from the process $\{\hat{\varepsilon}_t\}$, and $\mathcal{U}$ is a data matrix with $n$ consecutive observations from the process $\{ u_t\}$. Also, the process $\{u_t\}$ can be alternately expressed as \begin{eqnarray*}
		u_t &=& \Phi(L)y_t - \Theta(L) \hat{\varepsilon}_{t}  \\
		&=& \Theta(L) (a_t - \hat{\varepsilon}_t) \\
		&=& \Theta(L)(a_t - \varepsilon_t) - \Theta(L)(\hat{\varepsilon}_t - \varepsilon_t) \\
		&=& \Theta(L) \left(\Pi(L) - \Pi_{\tilde{[p]}}(L) \right) y_t - \Theta(L) (\hat{\varepsilon}_t - \varepsilon_t) \\
		&=& \mathcal{A}(L) y_t + \mathcal{B}(L)(\hat{\varepsilon}_t - \varepsilon_t), \mbox{  say.}
	\end{eqnarray*}

	\noindent The lag polynomial $\mathcal{A}(L) = \Theta(L) \Pi_{-[\tilde{p}]}(L)$ satisfies $\vertiii{\mathcal{A}} \le \vertiii{\Theta} \| \Pi_{-[\tilde{p}]} \|_{2,1}$, and $\mathcal{B}(L) = -\Theta(L)$ is a finite order lag polynomial. 
	
	Now note that each term $\mathcal{Y}_{(\ell)}^\top \mathcal{U}/n$ can be expressed in the form of a sample covariance matrix $\widehat{\cov}(L^\ell y_t, u_t):= \sum_{t=1}^n y_{t-\ell} u_t^\top/n$. With this notation, we can decompose this into two terms and apply deviation bounds from Proposition \ref{lem:yAy} and Lemma  \ref{lem:ehat-e} on each term separately. To be precise, for any $\ell$, $1 \le \ell \le p$, we have
	\begin{equation*}
		\widehat{\cov}(y_{t-\ell}, u_t) = \widehat{\cov}(L^\ell y_t, \mathcal{A}(L)y_t) + \widehat{\cov}(L^\ell y_t, \mathcal{B}(L)(\hat{\varepsilon}_t - \varepsilon_t)).
	\end{equation*}
	Similarly, for any $m$, $1 \le m \le q$, we can decompose $\widehat{\cov}(\hat{\varepsilon}_{t-m}, u_t)$ into four parts as 
	\begin{eqnarray*}
		&& \widehat{\cov}(\hat{\varepsilon}_{t-m} - \varepsilon_{t-m}, \mathcal{A}(L)y_t) + \widehat{\cov}(\hat{\varepsilon}_{t-m} - \varepsilon_{t-m}, \mathcal{B}(L)(\hat{\varepsilon}_t - \varepsilon_t)) \\
		&& + \widehat{\cov}(\Pi_{[\tilde{p}]}(L)L^m y_t, \mathcal{A}(L)y_t) + \widehat{\cov}(\Pi_{[\tilde{p}]}(L)L^m y_t, \mathcal{B}(L)(\hat{\varepsilon}_t - \varepsilon_t)).
	\end{eqnarray*}
	
	Using bounds from Proposition \ref{lem:yAy} and Lemma \ref{lem:ehat-e} then implies that there are universal constants $c_i > 0$ such that each of the following events hold with probability at least $1 - c_0 d^2 \exp[-(c_1 A^2-1) \log d^2(p+q)]$  as long as $n > q$, $\tilde{p} \ge p+q$ and $n \succsim \log d^2(p+q)$:
	\begin{align*}
		\left\| \widehat{\cov}(L^\ell y_t, \mathcal{B}(L)(\hat{\varepsilon}_t - \varepsilon_t)) \right\|_\infty &\le  2\sqrt{2 \pi} \vertiii{f_y}^{1/2} \Delta_{\varepsilon} \left\|\Theta \right\|_{2,1}\\
		\left\| \widehat{\cov}(L^\ell y_t, \mathcal{A}(L)y_t) \right\|_\infty &\le 2\pi \vertiii{f_y} \left[ \vertiii{\Theta} \left\| \Pi_{-[\tilde{p}]}\right\|_{2,1} + \right. \\
		3A \max &\,\left.\{1, \vertiii{\Theta}^2 \left\| \Pi_{-[\tilde{p}]}\right\|_{2,1}^2 \}  \sqrt{\log d^2(p+q) / n} \right]\\
		\left\| \widehat{\cov}(\hat{\varepsilon}_{t-m} - \varepsilon_{t-m}, \mathcal{A}(L)y_t) \right\|_\infty &\le  2\sqrt{2 \pi} \vertiii{f_y}^{1/2} \vertiii{\Theta} \left\| \Pi_{-[\tilde{p}]} \right\|_{2,1} \Delta_{\varepsilon}\\
		\left\| \widehat{\cov}(\hat{\varepsilon}_{t-m} - \varepsilon_{t-m}, \mathcal{B}(L)(\hat{\varepsilon}_t - \varepsilon_t)) \right\|_\infty &\le  2 \left\| \Theta \right\|_{2,1} \Delta_{\varepsilon}^2 \\
		\left\| \widehat{\cov}(\Pi_{[\tilde{p}]}(L)L^m y_t, \mathcal{B}(L)(\hat{\varepsilon}_t - \varepsilon_t)) \right\|_\infty &\le 2\sqrt{2 \pi} \vertiii{f_y}^{1/2} \vertiii{\Pi_{[\tilde{p}]}}^{1/2} \Delta_{\varepsilon} \left\| \Theta \right\|_{2,1} \\
		\left\| \widehat{\cov}(\Pi_{[\tilde{p}]}(L)L^m y_t, \mathcal{A}(L)y_t) \right\|_\infty &\le  2\pi \vertiii{f_y} \left[ \vertiii{\Theta} \left\| \Pi_{-[\tilde{p}]}\right\|_{2,1} \vertiii{\Pi_{[\tilde{p}]}} + \right. \\
		3A \max  \{\vertiii{\Pi_{[\tilde{p}]}}^2, &\, \left. \vertiii{\Theta}^2 \left\| \Pi_{-[\tilde{p}]}\right\|_{2,1}^2 \}\sqrt{\log d^2(p+q) / n} \right].
	\end{align*}
	
	Summing up the six terms above and taking a union bound over $1 \le \ell \le p$, $1 \le m \le q$, we obtain the final upper bound. 
\end{proof}

\begin{proof}[Proof of Proposition \ref{prop:phase2-qrs}]
	We start by deriving a suitable choice of $s_n$. To this end, note that for each $j$, $1 \le j \le d$, we have  $\wh{\var}(u_{tj})$ can be expressed as 
	\begin{eqnarray*}
		e_j^\top \wh{\cov}(u_t, u_t) e_j &=& e_j^\top \wh{\cov}\left(\Theta(L) \Pi_{-[\tilde{p}]}(L)y_t, \Theta(L) \Pi_{-[\tilde{p}]}(L)y_t \right) e_j \\
		&& - 2 e_j^\top \wh{\cov}\left( \Theta(L) \Pi_{-[\tilde{p}]}(L)y_t, \Theta(L)(\hat{\varepsilon}_t -\varepsilon_t) \right) e_j \\
		&& + e_j^\top  \wh{\cov}\left( \Theta(L)(\hat{\varepsilon}_t -\varepsilon_t), \Theta(L)(\hat{\varepsilon}_t -\varepsilon_t)  \right) e_j.
	\end{eqnarray*}
	We then use upper bounds on the individual terms using the deviation bounds provided in our technical ingredients.
	
	In particular, set $w_t := \Theta(L) \Pi_{-[\tilde{p}]}(L)y_t$. Then $\vertiii{f_w} \le \vertiii{\Theta} \left\| \Pi_{-[\tilde{p}]}  \right\|_{2,1}^2 \vertiii{f_y}$. Setting $\sigma^2_j = e_j^\top \Gamma_w(0) e_j$, Proposition 2.4 of \citep{Basu15} implies, with probability at least $1 - c_1 \exp \left[ -(c_2 A^2-1) \log d^2(p+q) \right]$, the following holds:
	\begin{equation*}
		\left| e_j^\top \wh{\cov} (w_t, w_t) e_j - \sigma^2_j \right| \le 2 \pi \vertiii{f_w} A \sqrt{\log d^2(p+q)/n}.
	\end{equation*}
	The second term in the above expansion, $e_j^\top \wh{\cov} \left(w_t, \Theta(L)(\hat{\varepsilon}_t - \varepsilon_t) \right) e_j$, can be bounded in absolute value (use Lemma \ref{lem:ehat-e} and note that $n > q$, $n \succsim \log d^2(p+q)$) by the following:
	\begin{equation*}
		2 \sqrt{2\pi} \vertiii{f_w}^{1/2} \Delta_\varepsilon \left\| \Theta \right\|_{2,1}.
	\end{equation*}
	The last term in the above expansion, $e_j^\top  \wh{\cov}\left( \Theta(L)(\hat{\varepsilon}_t -\varepsilon_t), \Theta(L)(\hat{\varepsilon}_t -\varepsilon_t)  \right) e_j$, can be bounded in absolute value (see proof of Lemma \ref{lem:ehat-e}(ii)) by the following:
	\begin{equation*}
		2 \left\| \Theta \right\|_{2,1}^2 \Delta_{\varepsilon}^2.
	\end{equation*}
	
	Combining these, we obtain the following choice of $s_n$ (with $\sigma^2$ as $\sum_{j=1}^d \sigma_j^2$):
	\begin{eqnarray*}
		s_n &=& 2 \pi \vertiii{\Theta} \left\| \Pi_{-[\tilde{p}]}  \right\|_{2,1}^2 \vertiii{f_y} A\, d \,\sqrt{\log d^2(p+q)/n} + \\
		&~& 4 \sqrt{2 \pi \vertiii{\Theta} \left\| \Pi_{-[\tilde{p}]}  \right\|_{2,1}^2\vertiii{f_y}} \, d \,  \Delta_\varepsilon \left\| \Theta \right\|_{2,1} + 2 d \, \left\| \Theta \right\|_{2,1}^2 \Delta_\varepsilon^2.
	\end{eqnarray*}

	The choice of $q_n$ follows from Lemma \ref{lem:conc_S_z}, with a union bound over $d^2(p+q)^2$ choices of $u, v$ as canonical unit vectors in $\mathbb{R}^{d(p+q)}$. In particular, we have 
	
	\begin{equation*}
		q_n = 2 \pi \vertiii{f_z} A \sqrt{\log d^2(p+q)/n} + 2q \Delta_a^2 + 2 \sqrt{2 \pi \vertiii{f_z} q} \Delta_a.
	\end{equation*}
	
	\noindent The choice of $r_n$ follows directly from the Proposition \ref{prop:prediction-phase2}.	
\end{proof}

\section{Implementation of the Sparse VARMA Procedure}\label{Algorithmdetail}
\textbf{Phase-II Proximal Gradient Algorithm.}
\color{black}
The objective function in \eqref{HVARXeq} is separable over the $d$ rows of $\Phi, \Theta$ and can thus be solved in parallel by solving the ``one-row" subproblems, see e.g., \cite{Nicholson16}. Denote the $i^{th}$ row of $Y$ by $Y_{i\cdot} = \mathbb{R}^{1\times(T-\bar{o})}$,  the $i^{th}$ row of $\Phi$ by $\Phi_{i\cdot} \in \mathbb{R}^{1\times dp}$ and the $i^{th}$ row of $\Theta$ by $\Theta_{i\cdot} \in \mathbb{R}^{1\times dq}$. The Proximal Gradient Algorithm for the one-row subproblems is given in Algorithm 1. \color{black} 

\begin{algorithm}[t]
	\scriptsize
	\caption{Proximal Gradient Algorithm to solve Phase-II \label{PGA_PhaseII}}
	\begin{description}
		\item [{Input}] $Y_{i\cdot}$, $Z$, $X$, $p$, $q$, ${  \Phi}_{i\cdot}[0]$, ${  \Theta}_{i\cdot}[0]$, $\lambda_{{  \Phi}}$, $\lambda_{{  \Theta}}$,   $\alpha, \color{black} \mathcal{P}_{\text{AR}}(\Phi), \mathcal{P}_{\text{MA}}(\Theta)$, \color{black} $\epsilon$
		\item [{Initialization}] Set
		
		\begin{itemize}
			\item ${  \Phi}_{i\cdot}[2] \leftarrow {  \Phi}_{i\cdot}[1] \leftarrow {  \Phi}_{i\cdot}[0]$
			\item ${  \Theta}_{i\cdot}[2] \leftarrow {  \Theta}_{i\cdot}[1] \leftarrow {  \Theta}_{i\cdot}[0]$
			\item step size $s=1/\sigma_1({  A})^2$, with  $\sigma_1({  A})$ the largest singular value of the matrix ${  A} = \left(\begin{smallmatrix} {  Z}\\ {  X}\end{smallmatrix} \right)$
		\end{itemize}
		\item[{Iteration}] For $r=3,4,\ldots$
		\begin{itemize}
			\item $\widehat{{ \phi}} \leftarrow  {  \Phi}_{i\cdot}[r-1] +  \dfrac{r-2}{r+1} \left({  \Phi}_{i\cdot}[r-1] -{  \Phi}_{i\cdot}[r-2]\right)$
			\item 	${  \Phi}_{i\cdot}[r] \leftarrow \dfrac{1}{1+\alpha  \cdot \lambda_{{ \Phi}} } \cdot \color{black} \text{Prox}_{s\lambda_{{ \Phi}} P_{i}^{({ {\Phi}})}}	\left(\widehat{{ \phi}} -s \nabla_{{ {\Phi}}} \mathcal{L}_i (\widehat{{ \phi}})\right)$, 
			
			where 
			\begin{itemize}
				\renewcommand{\labelitemii}{\tiny$\blacksquare$} 
				\item $\nabla_{{ {\Phi}}} \mathcal{L}_i (\widehat{ \phi}) = -({  Y}_{i\cdot} - \widehat{ \phi} {  Z}- { {\Theta}}_{i\cdot}[r-1] {  X}){  Z}^\top,$
				\item 	\color{black}	$ \text{Prox}_{s\lambda_{{ \Phi}} P_i^{({ {\Phi}})}}(\cdot)$  the proximal operator of the function $s\lambda_{{ \Phi}} P_i^{({ {\Phi}})}(\cdot)$ where $\mathcal{P}_{\text{AR}}(\Phi) = \sum_i P_i^{(\Phi)}(\Phi_{i\cdot}).$
				\color{black}
			\end{itemize}

			\item $\widehat{{ \theta}} \leftarrow  {  \Theta}_{i\cdot}[r-1] +  \dfrac{r-2}{r+1} \left({  \Theta}_{i\cdot}[r-1] -{  \Theta}_{i\cdot}[r-2]\right)$
			
			\item 	${  \Theta}_{i\cdot}[r] \leftarrow \color{black} \dfrac{1}{1+\alpha  \cdot \lambda_{{ \Theta}}  } \cdot \text{Prox}_{s\lambda_{{ \Theta}} P_i^{({ {\Theta}})}}	\left(\widehat{{ \theta}} -s \nabla_{{ {\Theta}}} \mathcal{L}_i (\widehat{{ \theta}})\right)$, 
			
			where
			\begin{itemize}
				\renewcommand{\labelitemii}{\tiny$\blacksquare$} 
				\item $\nabla_{{ {\Theta}}} \mathcal{L}_i (\widehat{{ \theta}}) = -({  Y}_{i\cdot} - { {\Phi}}_{i\cdot}[r] {  Z}- \widehat{{ \theta}} {  X}){  X}^\top,$
				\item 	\color{black} $ \text{Prox}_{s\lambda_{{ \Theta}} P_i^{({ {\Theta}})}}(\cdot)$ the proximal operator of the function $s\lambda_{{ \Theta}} P_i^{({ {\Theta}})}(\cdot)$
				where $\mathcal{P}_{\text{MA}}(\Theta) = \sum_i P_i^{(\Theta)}(\Theta_{i\cdot}).$
				\color{black}
			\end{itemize}
			
		\end{itemize}

		\item [{Convergence}] Iterate until $||{  \Phi}_{i\cdot}[r] - {  \Phi}_{i\cdot}[r-1]||_{\infty} \leq \epsilon$ and $||{  \Theta}_{i\cdot}[r] - {  \Theta}_{i\cdot}[r-1]||_{\infty} \leq \epsilon$
		\item[Output] $\widehat{{  \Phi}}_{i\cdot} \leftarrow {  \Phi}_{i\cdot}[r]; \widehat{{  \Theta}}_{i\cdot} \leftarrow {  \Theta}_{i\cdot}[r]$
	\end{description}
\end{algorithm}

\color{black}
\textbf{Choice of convex regularizers.} 
As indicated in Section \ref{methodology}, we focus on the $\ell_1$-norm and HLag penalty as choices of convex regularizers. For the $\ell_1$-norm, $$P_i^{({ {\Phi}})}({ {\Phi}}_i)=\sum_{j=1}^d \sum_{\ell=1}^{p} |   {\Phi}_{\ell, ij} | \ \text{and} \ P_i^{({ {\Phi}})}({ {\Theta}}_i)=\sum_{j=1}^d \sum_{m=1}^{q} | {  \Theta}_{m, ij} |.$$ For the HLag penalty, $$P_i^{({ {\Phi}})}({ {\Phi}}_i)=\sum_{j=1}^d \sum_{\ell=1}^{p} ||   {\Phi}_{(\ell:{p}), ij} || \ \text{and} \ P_i^{({ {\Theta}})}({ {\Theta}}_i)=\sum_{j=1}^d \sum_{m=1}^{q} || {  \Theta}_{(m:q), ij} ||.$$
\color{black}

\section{\label{Simulation}Simulation Study}
We investigate the performance of the proposed VARMA estimator through a simulation study.
We generate data from a $\text{VARMA}_{d}(p,q)$ with time series length $T=100$.  
To ensure identification, we take $  \Phi_{\ell}$, $1\leq \ell \leq p$, diagonal matrices and set each diagonal element of $  \Phi_{\ell}$ equal to $0.4/\ell $.  For the autoregressive order, we take $p=4$. For the error covariance matrix, we take ${ {\Sigma}}_{a}= {  I}_d$. \color{black}To reduce the influence of initial conditions on the data generating processes, the first 200
observations were discarded as burn-in for each simulation run. \color{black}

We consider several settings for the moving average parameters.
We take banded matrices for $  \Theta_{m}$, $1\leq m \leq q$ with the diagonal elements of $  \Theta_{m}$ equal to $\theta/m$, the elements on the first lower and upper subdiagonals equal to $\theta/(10m)$, and the elements on the second lower and upper subdiagonals equal to $\theta/(100m)$.
The parameter $\theta$ regulates the strength of the moving average signal. The parameter $q$ regulates the moving average order.
We investigate the effect of the following features.  
\textit{(i) The MA signal strength}: we vary the parameter $\theta \in \{0,0.4,0.6,0.8\}$. The larger $\theta$, the stronger the moving average signal. Note that for $\theta=0$, the true model is a VAR. 
\textit{(ii) The MA order}: we vary the parameter $q \in \{4,6,8,10\}$. 
\textit{(iii) The number of time series}: we vary the number of time series $d \in \{5,10,20,40\}$. 
In all considered settings, the  VARMA models are invertible and stable.

\textbf{Estimators.}
We compare the following estimators.
(i) ``VARMA($p,q; {  a}_t$)": the VARMA estimator of model \eqref{VARMA} with an oracle providing the true errors ${  a}_t$ and orders $p$ and $q$. 
(ii) ``VARMA($p,q; \widehat{{  \varepsilon}}_t$)" the VARMA estimator of model \eqref{VARX} with approximated errors and an oracle providing the orders $p$ and $q$.
(iii) ``VARMA($\widehat{p},\widehat{q}; \widehat{{  \varepsilon}}_t$)": the VARMA estimator of model \eqref{VARX} with approximated errors and specified orders $\widehat{p}=\widehat{q}= \lfloor 0.75\sqrt{T} \rfloor $.
(iv) ``VAR($\widetilde{p})$": the VAR estimator of model \eqref{VARptilde} with specified order $\widetilde{p}= \lfloor 1.5\sqrt{T} \rfloor$.
\color{black} We use both the $\ell_1$-norm and the HLag penalty to obtain our estimates. 
\color{black}

\textbf{Performance Measure.}
We compare the performance of the  estimators in terms of out-of-sample forecast accuracy. 
We generate time series of length $T+1$ and use the last observation to measure forecast accuracy. We compute the Mean Squared Forecast Error
\begin{equation}
	\text{MSFE}= \dfrac{1}{N} \sum_{s=1}^{N}\dfrac{1}{d}\|{y}^{(s)}_{T+1}-\widehat{y}_{T+1}^{(s)}\|^2, \nonumber
\end{equation}
where 
${  y}_t^{(s)}$  is the vector of time series at time point $t$ in the $s^{th}$ simulation run, and  $\widehat{y}_{t}^{(s)}$ is its predicted value. The number of simulations is $N=500$. We focus on out-of-sample forecast accuracy in the simulation study, in line with the discussion of the applications in Section \ref{applications}. \color{black}
We did also compare the estimators in terms of the estimation accuracy of the $\Pi$-matrices; similar conclusions are obtained and available from the authors upon request. \color{black}

\subsection{\label{masignal.effect}Effect of the Moving Average Signal Strength}
Figure \ref{MSFEsims} panel (a) shows the MSFEs (averaged over the simulation runs) of the four estimators for different values of the moving average parameter $\theta$, which regulates the moving average signal strength. We report the results for  the  HLag penalty and $d=10, q=4$.

\begin{figure}
	\includegraphics[width=0.35\textwidth]{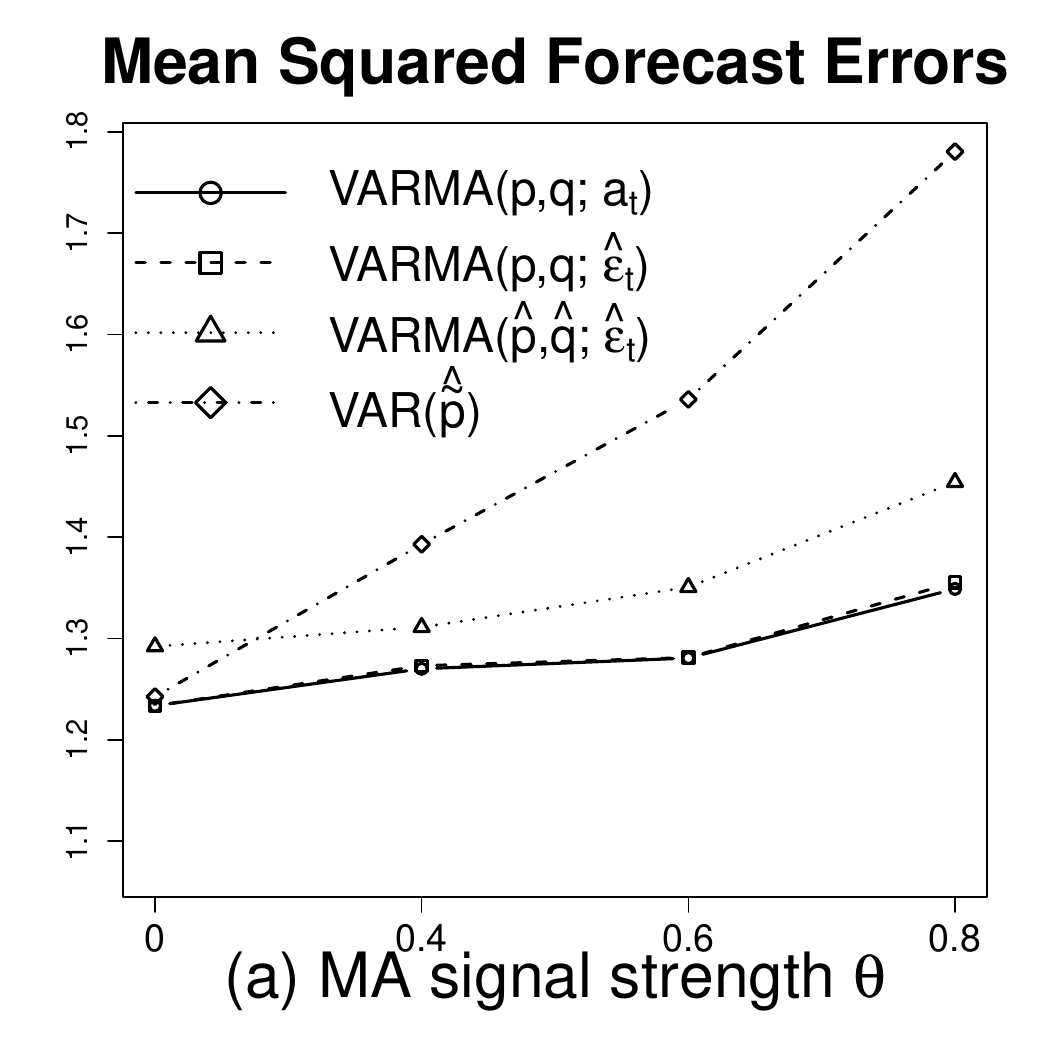}
	\hspace{-0.7cm}
	\includegraphics[width=0.35\textwidth]{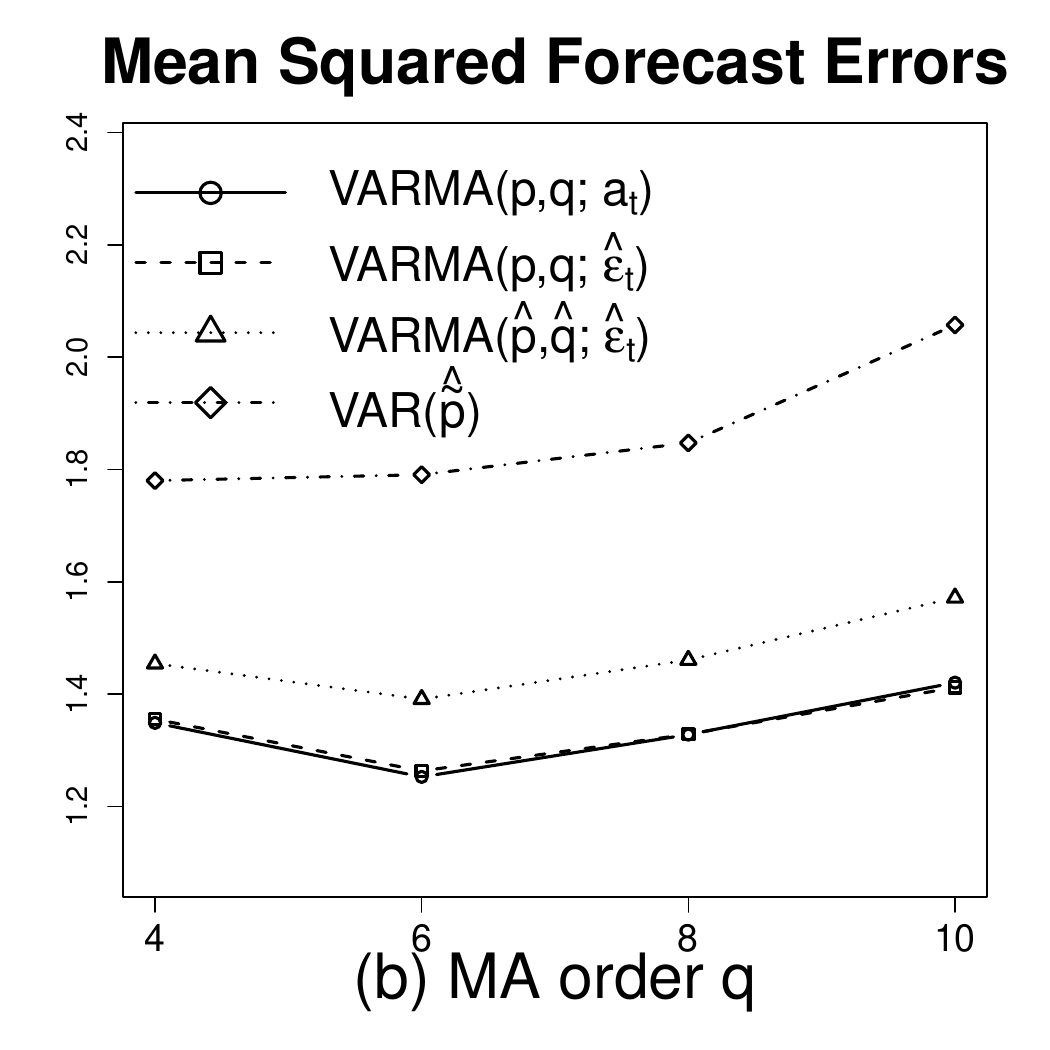}
	\hspace{-0.7cm}
	\includegraphics[width=0.35\textwidth]{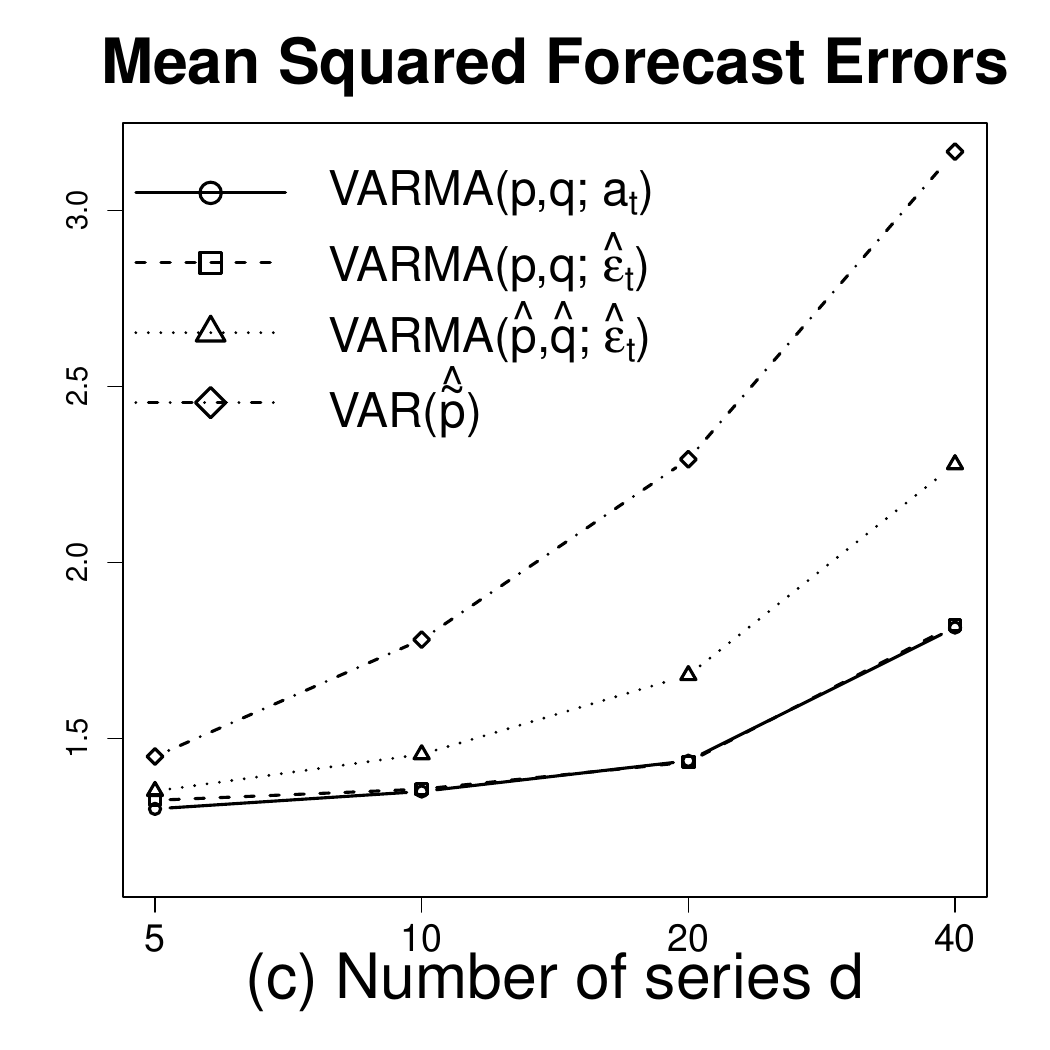}
	\caption{Mean Squared Forecast Errors (averaged over the simulation runs) of the four estimators for different values of (a) the moving average parameter $\theta$, 
		(b) the moving average order $q$, and (c) the number of time series $d$.  \label{MSFEsims}}
\end{figure}

If the true model is a VARMA (i.e.\ $\theta\neq0$), the VARMA estimators perform better than the VAR, as expected. The larger $\theta$, the larger the gain of VARMA over VAR.  
The differences in forecast accuracy between the VARMA estimators and the VAR estimator are all significant, as confirmed by paired $t$-tests (at the 5\% significance level). 
\color{black}
Among the VARMA estimators, there is no significant difference between 
``VARMA($p,q; {  a}_t$)" and ``VARMA($p,q; \widehat{{  \varepsilon}}_t$)"  
thus supporting the validity of the two-phase approach.
\color{black}
The VARMA estimator with estimated errors and selected orders (i.e.,  ``VARMA($\widehat{p},\widehat{q}; \widehat{{  \varepsilon}}_t$)") performs, for all values of $\theta$, very similarly to the one with known orders. 
\color{black} The loss in forecast accuracy of not knowing the autoregressive or moving average order is limited to 5\% on average. \color{black}

If the true model is a VAR (i.e.\ $\theta=0$), the VARMA estimators with known orders both reduce to a VAR($p$) estimator since $\theta=0$, hence $q=0$. They give the lowest MSFE. 
However, in practice, the orders of the model are not known. \color{black}
For unknown orders, the VARMA estimator is competitive to the VAR estimator. 
The VARMA estimator attains this competitveness since, in general, it returns a more parsimonious model (i.e.\  the estimated VARMA has more sparse AR coefficients with some sparse MA coefficients than the number of sparse coefficients in the estimated VAR).

\color{black}
The relative performance of the four estimators with HLag penalty are compared to the results with $\ell_1$-norm in Table \ref{HLagvsL1}.
For the estimators with unknown maximal lag orders (i.e.\ VARMA($\widehat{p},\widehat{q}; \widehat{{  \varepsilon}}_t)$ and VAR($\widetilde{p})$),	HLag  outperforms the $\ell_1$-norm in all considered cases ($p$-values paired $t$-test $< 0.05$). 
For the estimators with known maximal lag orders, (e.g., VARMA($p,q; {  a}_t$) and VARMA($p,q; \widehat{{  \varepsilon}}_t$)), HLag performs, overall, as good as the $\ell_1$-norm. These results are in line with the findings of \cite{Nicholson16}.

\color{black}

\begin{table}
	\caption{Mean Square Forecast Errors (averaged over the simulation runs) of the four estimators with either HLag penalty or $\ell_1$-norm and for different values of the moving average parameter $\theta$. $P$-values of a paired $t$-test are in parentheses. \label{HLagvsL1}}
	\centering
	\begin{tabular}{lccccccccccc}
		\hline
		& \multicolumn{2}{c}{VARMA($p,q; {  a}_t$)} 	&& \multicolumn{2}{c}{VARMA($p,q; \widehat{{  \varepsilon}}_t$)} 	&& \multicolumn{2}{c}{VARMA($\widehat{p},\widehat{q}; \widehat{{  \varepsilon}}_t$)} 	&& \multicolumn{2}{c}{VAR($\widetilde{p})$} \\
		& HLag &  $\ell_1$ && HLag&  $\ell_1$ && HLag &  $\ell_1$ && HLag &  $\ell_1$ \\ 
		\hline
		$\theta = 0$   & 1.234 & $\underset{(<0.01)}{1.263}$ && 1.234 & $\underset{(<0.01)}{1.263}$ && 1.292 & $\underset{(<0.01)}{1.334}$ && 1.243 & $\underset{(<0.01)}{1.317}$ \\
		$\theta = 0.4$ & 1.270 & $\underset{(0.415)}{1.299}$ && 1.273 & $\underset{(0.396)}{1.303}$ && 1.311 & $\underset{(0.040)}{1.387}$ && 1.393 & $\underset{(<0.01)}{1.558}$ \\ 
		$\theta = 0.6$ & 1.281 & $\underset{(0.360)}{1.315}$ && 1.281 & $\underset{(0.293)}{1.321}$ && 1.351 & $\underset{(<0.01)}{1.459}$ && 1.536 & $\underset{(<0.01)}{1.802}$ \\  
		$\theta = 0.8$ & 1.349 & $\underset{(0.275)}{1.383}$ && 1.355 & $\underset{(0.170)}{1.399}$ && 1.454 & $\underset{(<0.01)}{1.582}$ && 1.780 & $\underset{(<0.01)}{2.159}$ \\ 
		\hline
	\end{tabular}
\end{table}

\subsection{\label{maorder.effect}Effect of the Moving Average Order}
Figure \ref{MSFEsims} panel (b)  shows the MSFEs of the four estimators for different values of the moving average order $q$. We report the results for \color{black} the HLag penalty and $d=10, \theta=0.8$. Similar conclusions are obtained with the $\ell_1$-norm and other values of $d$ and $\theta$, therefore omitted. \color{black}
For all values of $q$, the VARMA estimators perform significantly better than the VAR estimator. \color{black}
The oracle VARMA estimators perform equally good and are closely followed by the VARMA estimator with approximated errors and unknown orders. 
The latter improves forecast accuracy over the VAR estimator by about 20\% on average. \color{black}

\subsection{\label{timeseries.effect}Effect of the Number of Time Series}
Figure \ref{MSFEsims} panel (c)  shows the MSFEs  for different values of the number of time series $d$. We report the \color{black} results for  the HLag penalty and  $q=4,  \theta=0.8$. \color{black} As the number of time series increases relative to the fixed time series length $T$, it becomes more difficult to accurately estimate the model. Hence, the MSFEs of all estimators increase.
For all values of $d$, the VARMA estimators attain lower values of the MSFE than the VAR estimator. 
All differences are significant. 
The loss in forecast accuracy of not knowing the AR and MA order is only 2\% for $k=5$ and remains limited to 20\% for $k=40$.
The margin by which the VARMA estimator (with approximated errors and unknown orders) improves forecast accuracy over the VAR increases from around 7\% for $k=5$ to around 30\% for $k=40$. \color{black}

\begin{figure}
	\centering
	\includegraphics[width = 0.36\textwidth]{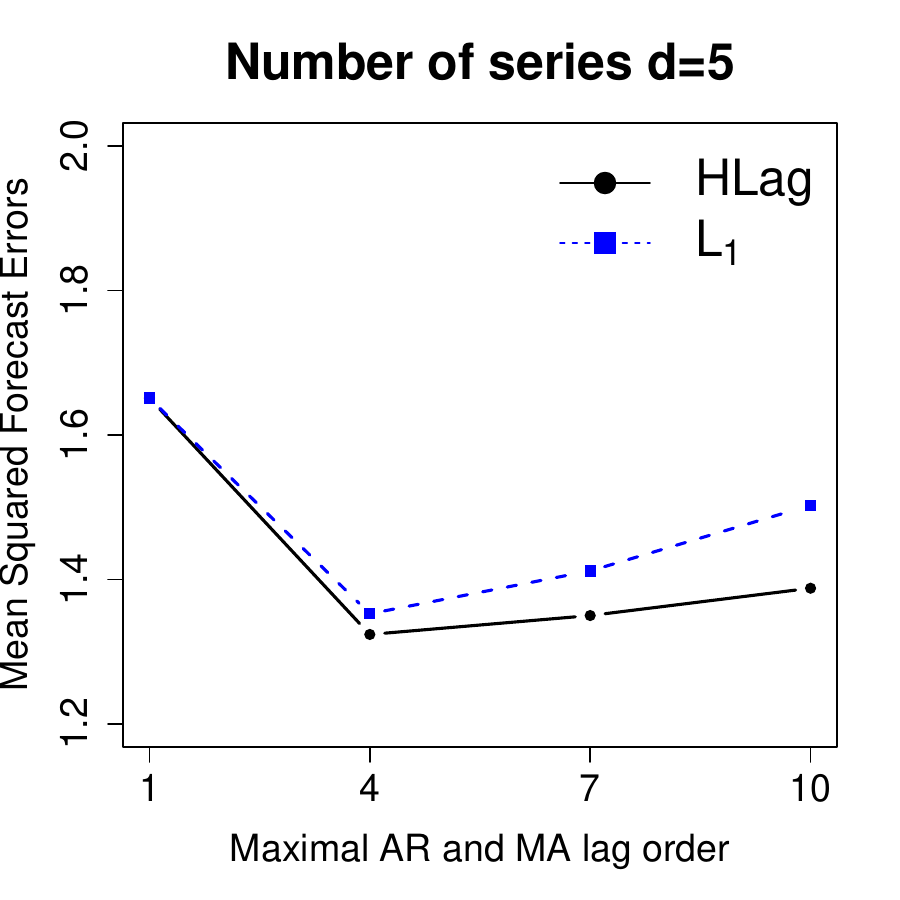}
	\hspace{0.5cm}
	\includegraphics[width = 0.36\textwidth]{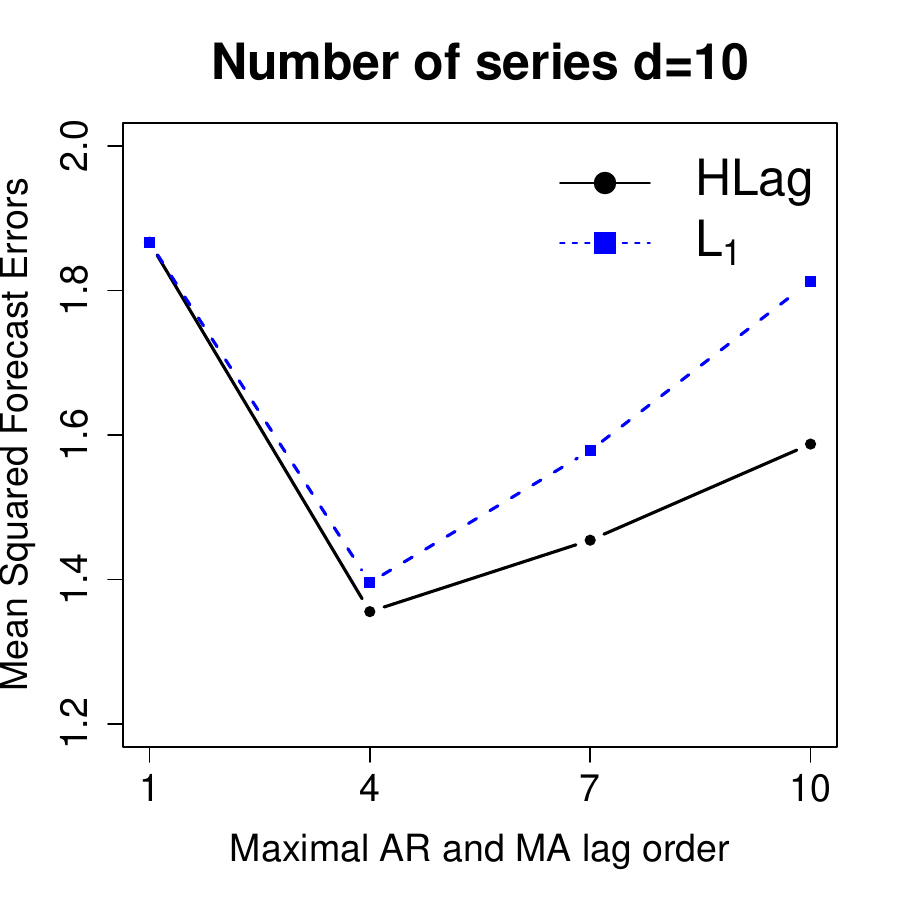}
	\caption{Mean Squared Forecast Errors (averaged over the simulation runs) of the VARMA($\widehat{p}$, $\widehat{q}$; $\widehat{\epsilon}_t$) estimator with HLag penalty (black solid line) and $\ell_1$-penalty (blue dashed line)  for different values of the maximal lag orders $p$ and $q$ (horizontal axis) and number of time series $d=5$ (left), $d=10$ (right). \label{sim_misspecify}}
\end{figure}
\subsection{Implications of Misspecifying the Maximal Lag Orders} \label{misspecify_orders}
Next, we investigate the implications of misspecifying the maximal AR and MA lag orders of the VARMA.  We generate data from the VARMA model with $p=q=4, \theta = 0.8$, $d=5, 10$ and
estimate a sparse VARMA model with maximal lag orders smaller than, equal to and larger than the true orders. Note that a maximal lag order of seven corresponds to our recommendation ($\hat{p} = \hat{q} = \lfloor 0.75 \sqrt{100} \rfloor = 7$). Figure \ref{sim_misspecify} shows the MSFEs for the sparse VARMA estimator with HLag penalty and $\ell_1$ penalty,  different values of the maximal AR and MA lag orders (horizontal axis) and number of time series (panels).

The lowest MSFEs are attained at the true maximal lag order of four, as expected. At a maximal lag order of one, all models are misspecified and the MSFEs are the largest. 
Using too small maximal lag orders thus has more severe consequences than using too large maximal lag orders.
Furthermore, the drop in MSFE at our recommended maximal lag orders (of seven) remains small provided that one uses an HLag penalty. Indeed, the price to pay for too large maximal lag orders is smaller for HLag than the standard $\ell_1$ penalty since HLag encourages low maximal lag orders.

\subsection{Data-based Simulation Design}
As a last experiment, we consider a data-based design \cite{ho1996}. 
Similar to \cite{carriero2012}, we carry out a simulation by bootstrapping the actual demand set with $d=16$ and $T=76$ as discussed in Section \ref{applications} of the paper.
We start from the autoregressive and moving average estimates obtained with the sparse VARMA method with HLag penalties and 
$\hat{p} = \hat{q} = \lfloor 0.75\sqrt{T} \rfloor = 6$
We then  generate data from a VARMA$_d(\hat{p}, \hat{q} )$  using a non-parametric residual bootstrap procedure (e.g., \citep{Kreiss12}) with bootstrap errors an i.i.d.\ sequence of discrete random variables uniformly distributed on $\{1, \ldots, T\}$.

\begin{table} 
	\centering
	\caption{Data-based Simulation Design: Mean Squared Forecast Errors (averaged over the simulation runs) of the four estimators with HLag penalty and different forecast horizons. Standard errors around the reported results are in parentheses. \label{sim_data_based}}
	\begin{tabular}{lcccc} \hline
		Forecast horizon& VARMA($p,q; {  a}_t$) 	&VARMA($p,q; \widehat{{  \varepsilon}}_t$) 	&VARMA($\widehat{p},\widehat{q}; \widehat{{  \varepsilon}}_t$) 	&VAR($\widetilde{p})$ \\ \hline 
		$h=1$   & $\underset{(0.021)}{0.764}$ & $\underset{(0.022)}{0.763}$ & $\underset{(0.022)}{0.765}$ & $\underset{(0.025)}{0.758}$ \\
		$h=8$   & $\underset{(0.022)}{0.782}$& $\underset{(0.022)}{0.783}$& $\underset{(0.022)}{0.785}$& $\underset{(0.026)}{0.877}$\\
		$h=13$  & $\underset{(0.022)}{0.784}$& $\underset{(0.022)}{0.785}$& $\underset{(0.022)}{0.787}$& $\underset{(0.026)}{0.877}$\\ \hline 
	\end{tabular}
\end{table}

Table \ref{sim_data_based} gives the MSFEs of the four estimators at  different forecast horizons $h=1, 8, 13$, as used in Section \ref{applications}. For the VARMA estimators with known orders, we use $p=q=3$, in line with the largest reported values in Figure \ref{Demand_lhat}.
First of all, note that it becomes more difficult to obtain accurate forecasts for longer horizons;  the 
MSFEs of all estimators increases with $h$.
The relative performance of VARMA compared to VAR is tied to the forecast horizon: 
at $h=1$, all estimators perform equally well (i.e.\ there are no significant differences, as confirmed through paired $t$-tests).
At longer forecast horizons, the VARMA estimators still perform equally well but statistically outperform the VAR estimator.
These findings support the results from Section \ref{applications}.
\end{appendices}

%%%% References %%%%	
	\begingroup
	\bibliographystyle{asa}
	\setstretch{0.05}
	\linespread{0.5}
	\bibliography{VARMAref}
	\endgroup	
\end{document}